\documentclass[hidelinks, 11pt]{article}

\usepackage{array}
\newcolumntype{P}[1]{>{\centering\arraybackslash}p{#1}}

\usepackage{booktabs} 
\newcommand{\tabitem}{\llap{\textbullet}~}
\usepackage{longtable}

\usepackage{hyperref}

\usepackage{color,soul}
\usepackage{xcolor}
\usepackage{framed}
\colorlet{shadecolor}{yellow!20}

\usepackage{tikz}
\usepackage{verbatim}

\usepackage{algorithm} 
\usepackage{algpseudocode} 

\usepackage{amssymb}
\usepackage{import}
\usepackage[pagewise]{lineno}

\usepackage{mathtools,bbm}
\usepackage{graphicx,url}
\usepackage{graphics}
\usepackage{multirow}
\usepackage{tabularx}
\usepackage{threeparttable}
\usepackage{siunitx}
\usepackage[justification=centering, labelsep=period]{caption}
\usepackage[noabbrev, capitalise]{cleveref}
\usepackage{makecell}
\usepackage{cases} 

\DeclareMathOperator*{\E}{\mathbb{E}}

\makeatletter 
\def\underbar#1{\underline{\sbox\tw@{$#1$}\dp\tw@\z@\box\tw@}} 
\makeatother 

\usepackage{setspace,graphicx,epstopdf,amsmath,amsfonts,amssymb,amsthm,versionPO,bm} 
\usepackage{marginnote,datetime,enumitem,rotating,fancyvrb,scalerel}
\usepackage{hyperref}

\excludeversion{notes}    
\includeversion{links}          

\iflinks{}{\hypersetup{draft=true}}

\ifnotes{%
\usepackage[margin=1in,paperwidth=10in,right=2.5in]{geometry}%
\usepackage[textwidth=1.4in,shadow,colorinlistoftodos]{todonotes}%
}{%
\usepackage[margin=1in]{geometry}%
\usepackage[disable]{todonotes}%
}

\makeatletter\let\chapter\@undefined\makeatother 

\usepackage{indentfirst} 
\usepackage{jfe}          

\newtheorem{theorem}{Theorem}

\newtheorem{proposition}{Proposition}

\newtheorem{lemma}{Lemma}

\newtheorem{definition}{Definition}

\usepackage[
  backend=biber,natbib,
  sorting=nyt, 
   style=apa,
   uniquename=false 
]{biblatex}
\DefineBibliographyStrings{german}{%
  andothers = {et al.},
}

\usepackage{float}
\usepackage[caption = false]{subfig}
\usepackage{blindtext}

\newcommand*\linenomathpatch[1]{%
  \cspreto{#1}{\linenomath}%
  \cspreto{#1*}{\linenomath}%
  \csappto{end#1}{\endlinenomath}%
  \csappto{end#1*}{\endlinenomath}%
}

\linenomathpatch{equation}
\linenomathpatch{gather}
\linenomathpatch{multline}
\linenomathpatch{align}
\linenomathpatch{alignat}
\linenomathpatch{flalign}

\RequirePackage{color}\definecolor{RED}{rgb}{1,0,0}\definecolor{BLUE}{rgb}{0,0,1} 
\definecolor{burgundy}{rgb}{0.7, 0.11, 0.11}
\RequirePackage[stable]{footmisc} 
\providecommand{\DIFaddtex}[1]{{\protect\color{black} #1}} 
\providecommand{\DIFdeltex}[1]{} 
\providecommand{\DIFaddbegin}{} 
\providecommand{\DIFaddend}{} 
\providecommand{\DIFdelbegin}{} 
\providecommand{\DIFdelend}{} 
\providecommand{\DIFaddFL}[1]{\DIFadd{#1}} 
\providecommand{\DIFdelFL}[1]{} 
\providecommand{\DIFaddbeginFL}{} 
\providecommand{\DIFaddendFL}{} 
\providecommand{\DIFdelbeginFL}{} 
\providecommand{\DIFdelendFL}{} 
\providecommand{\DIFadd}[1]{\texorpdfstring{\DIFaddtex{#1}}{#1}} 
\providecommand{\DIFdel}[1]{} 

\RequirePackage{settobox} 
\RequirePackage{letltxmacro} 
\newsavebox{\DIFdelgraphicsbox} 
\LetLtxMacro{\DIFOincludegraphics}{\includegraphics} 
\newcommand{\DIFaddincludegraphics}[2][]{{\color{black}\fbox{\DIFOincludegraphics[#1]{#2}}}} 
\newcommand{\DIFdelincludegraphics}[2][]{
} 
\LetLtxMacro{\DIFOaddbegin}{\DIFaddbegin} 
\LetLtxMacro{\DIFOaddend}{\DIFaddend} 
\LetLtxMacro{\DIFOdelbegin}{\DIFdelbegin} 
\LetLtxMacro{\DIFOdelend}{\DIFdelend} 
\DeclareRobustCommand{\DIFaddbegin}{\DIFOaddbegin \let\includegraphics\DIFaddincludegraphics} 
\DeclareRobustCommand{\DIFaddend}{\DIFOaddend \let\includegraphics\DIFOincludegraphics} 
\DeclareRobustCommand{\DIFdelbegin}{\DIFOdelbegin \let\includegraphics\DIFdelincludegraphics} 
\DeclareRobustCommand{\DIFdelend}{\DIFOaddend \let\includegraphics\DIFOincludegraphics} 
\LetLtxMacro{\DIFOaddbeginFL}{\DIFaddbeginFL} 
\LetLtxMacro{\DIFOaddendFL}{\DIFaddendFL} 
\LetLtxMacro{\DIFOdelbeginFL}{\DIFdelbeginFL} 
\LetLtxMacro{\DIFOdelendFL}{\DIFdelendFL} 
\DeclareRobustCommand{\DIFaddbeginFL}{\DIFOaddbeginFL \let\includegraphics\DIFaddincludegraphics} 
\DeclareRobustCommand{\DIFaddendFL}{\DIFOaddendFL \let\includegraphics\DIFOincludegraphics} 
\DeclareRobustCommand{\DIFdelbeginFL}{\DIFOdelbeginFL \let\includegraphics\DIFdelincludegraphics} 
\DeclareRobustCommand{\DIFdelendFL}{\DIFOaddendFL \let\includegraphics\DIFOincludegraphics} 
\RequirePackage{listings} 
\RequirePackage{color} 
\lstdefinelanguage{DIFcode}{ 
  moredelim=[il][\color{black}\scriptsize]{\%DIF\ <\ }, 
  moredelim=[il][\color{black}\sffamily]{\%DIF\ >\ } 
} 
\lstdefinestyle{DIFverbatimstyle}{ 
  language=DIFcode, 
  basicstyle=\ttfamily, 
  columns=fullflexible, 
  keepspaces=true 
} 
\lstnewenvironment{DIFverbatim}{\lstset{style=DIFverbatimstyle}}{} 
\lstnewenvironment{DIFverbatim*}{\lstset{style=DIFverbatimstyle,showspaces=true}}{} 

\begin{document}

\allowdisplaybreaks
\setlist{noitemsep}  

\title{\textbf{A general equilibrium model for multi-passenger ridesharing systems with stable matching}}

\author{Rui Yao\textsuperscript{1}, Shlomo Bekhor\textsuperscript{1,}\footnote{Corresponding author. \\E-mail address: sbekhor@technion.ac.il (Shlomo Bekhor)}
  \\
  1 - Department of Civil and Environmental Engineering, \\Technion – Israel Institute of Technology, Haifa, Israel}

\date{} 

\renewcommand{\thefootnote}{\fnsymbol{footnote}}
\singlespacing
\maketitle

\vspace{-.2in}
\begin{abstract}
\noindent
This paper proposes a general equilibrium model for multi-passenger ridesharing systems, in which interactions between ridesharing drivers, passengers, platforms, and transportation networks are endogenously captured. Stable matching is modeled as an equilibrium problem in which no ridesharing driver or passenger can reduce his/her ridesharing disutility by unilaterally switching to another matching sequence. 

This paper is one of the first studies that explicitly integrates the ridesharing platform's multi-passenger matching problem into the model. By \DIFdelbegin \DIFdel{proposing a }\DIFdelend \DIFaddbegin \DIFadd{integrating matching sequence with }\DIFaddend hyper-network\DIFdelbegin \DIFdel{approach}\DIFdelend , ridesharing-passenger transfers are avoided in a multi-passenger ridesharing system\DIFdelbegin \DIFdel{, and ridesharing driver's routes can be determined endogenously}\DIFdelend . Moreover, the matching stability between the ridesharing drivers and passengers is extended to address the multi-OD multi-passenger case \DIFaddbegin \DIFadd{in terms of matching sequence}\DIFaddend.  The paper provides a proof for the existence of the proposed general equilibrium. 

A sequence-bush algorithm is developed for solving the multi-passenger ridesharing equilibrium problem. This algorithm is capable to handle complex ridesharing constraints implicitly.
Results illustrate \DIFaddbegin \DIFadd{that the proposed sequence-bush algorithm outperforms general-purpose solver, and provides insights into the equilibrium of }\DIFaddend the \DIFdelbegin \DIFdel{equilibrium for the }\DIFdelend joint stable matching and route choice problem. Numerical experiments indicate that ridesharing trips are typically longer than average trip lengths. Sensitivity analysis suggests that a properly designed ridesharing unit price is necessary to achieve network benefits, and travelers with relatively lower values of time are more likely to participate in ridesharing.
\\

\end{abstract}

\medskip

\textit{Keywords}: General equilibrium; Multi-passenger ridesharing; Hyper-network; Stable matching\DIFaddbegin \DIFadd{; Sequence-bush assignment}\DIFaddend .

\hspace{\fill}

\hspace{\fill}

\hspace{\fill}

\hspace{\fill}

\hspace{\fill}

\hspace{\fill}

\hspace{\fill}

\noindent\fbox{%
    \parbox{\textwidth}{%
        \textbf{Please cite the published version:}

\noindent Yao, R., and Bekhor, S. (2023). A general equilibrium model for multi-passenger ridesharing systems with stable matching. Transportation Research Part B: Methodological, 175, 102775. doi:10.1016/j.trb.2023.05.012
    }%
}

\thispagestyle{empty}

\clearpage

\onehalfspacing
\setcounter{footnote}{0}
\renewcommand{\thefootnote}{\arabic{footnote}}
\setcounter{page}{1}

\section{Introduction}
\label{sec:Intro}

Shared mobility services, namely ride-sourcing services, have become an integral part of our transportation systems. Through ride-sourcing platforms like Uber, Lyft, and DiDi, passengers requesting these services are connected with drivers providing rides. 
One type of these services is ridesharing, in which peer drivers take multiple passengers at the same time. Different from dedicated drivers (as in the case of ride-hailing), peer drivers in ridesharing are also travelers, who are assumed to perform some activities at their designated destinations and not cruise in the network (\cite{wang2019ridesourcing}). Consequently, ridesharing services are generally expected to alleviate traffic congestion and reduce the costs of the participants (\cite{ma2022general}). In contrast, ride-hailing services have been criticized for inducing traffic in urban areas (\cite{erhardt2019transportation}, \cite{diao2021impacts}). However, the market share of ridesharing(pooling) services remains low, for example, only 20\% of the ride-sourcing trips of Uber are shared (\cite{Uber2018}). 
In order to improve the effectiveness of ridesharing services, it is necessary to investigate the intrinsic relations between traveler decisions, platform operations, and network congestion in a ridesharing system, so as to assess their long-term impacts on network congestion within a unified framework. On the other hand, these interactions pose additional challenges in finding equilibria of the overall system. 
In the following paragraphs, we detail the interactions between traveler decisions and platform operations over the underlying transportation networks, as well as the complications concerning the solution methods.   

\textbf{Traveler decisions.} Early studies in ridesharing(carpooling) modeling consider the driver's carpooling decision as a mode choice problem within multi-modal network equilibrium models, in which only vehicular flows are modeled and passenger decisions are ignored (\cite{daganzo1981equilibrium}; \cite{huang2000models}). To account for passenger mode choice decisions in ridesharing, passenger flows are incorporated in network models using super-networks (\cite{xu2015complementarity}; \cite{di2019unified}), or assumed traversing on a set of passenger routes (\cite{bahat2016incorporating}; \cite{li2020path}). By considering traveler mode choice behaviors, ridesharing driver supply and passenger demand are coupled to resemble the supply-demand imbalance in ridesharing systems. 

Alternatively, the interactions between driver supply and passenger demand are modeled as a two-sided market equilibrium problem, which typically assumes constant travel times and omits spatial imbalance (\cite{he2018pricing}; \cite{ke2020pricing}; \cite{zhang2021pool}). However, such assumptions might not be satisfied when planning to promote ridesharing. In this case, the increased number of ridesharing vehicles could have significant impacts on network congestion, and the route choice behavior of ridesharing drivers in congested networks would impact passenger travel times (\cite{chen2022unified}). In order to evaluate the impacts of ridesharing on network congestion, there is a need to not only model the matching between ridesharing drivers and passengers (as in the case of market equilibrium models), but also to consider the route choice behaviors in network models.  

\textbf{Platform operations.} The imbalance between ridesharing supply and demand might hinder the benefits of ridesharing. For example, ridesharing drivers, who cannot find passengers needing rides, could eventually leave ridesharing services. Consequently, ridesharing passengers, who experience longer waiting times due to low driver supply, might quit ridesharing as well (\cite{wang2019ridesourcing}). Ridesharing platforms are expected to (partially) resolve such imbalance by matching ridesharing drivers and passengers through mobile applications. However, improper platform operations (e.g., matching with long detours) could reduce both the service qualities and network benefits.

Platform operations are particularly challenging in a flexible setting, where one ridesharing driver is allowed to take multiple passengers from different origin-destination (OD) pairs at the same time. Such setting could bring additional network benefits due to efficient vehicle capacity utilization, but create challenges for platform operations and equilibrium modelings, as well as their solution methods. 

For platform operations, \DIFdelbegin \DIFdel{one of the challenges is to solve the }\DIFdelend \DIFaddbegin \DIFadd{efficiently solving the }\DIFaddend multi-passenger ridesharing matching problem \DIFdelbegin \DIFdel{efficiently }\DIFdelend \DIFaddbegin \DIFadd{is a challenge }\DIFaddend (\cite{agatz2012optimization}). \textcite{alonso2017demand} propose an efficient bi-level solution scheme, where the upper-level finds shareable driver-passenger groups based on graph cliques (\cite{santi2014quantifying}), and the lower-level solves the vehicle routing problem (VRP) for each group. 
\DIFdelbegin \DIFdel{As an extension, \mbox{
\textcite{yao2021dynamic} }\hskip0pt
incorporate }\DIFdelend \DIFaddbegin \DIFadd{\mbox{
\textcite{yao2021dynamic} }\hskip0pt
extend this approach by integrating }\DIFaddend a dynamic tree structure to efficiently determine feasible pickup and drop-off sequences for the lower-level VRP\DIFdelbegin \DIFdel{, where they also found network benefits of ridesharing are dependent on operational parameter settings. }\DIFdelend \DIFaddbegin \DIFadd{. Other studies incorporate additional realistic factors in the ridesharing matching problem. For example, \mbox{
\textcite{liu2020integrated} }\hskip0pt
and \mbox{
\textcite{zhou2022scalable} }\hskip0pt
consider endogenously the traffic congestion induced by ridesharing in the matching problem, with the goal to reduce platform operational costs. \mbox{
\textcite{mahmoudi2021many} }\hskip0pt
further consider optimizing jointly the passenger activity planning and the matching problem. In contrast, some studies simplify the matching problem by solving successive single-passenger matching problems }\citep{simonetto2019real}\DIFadd{, or reformulating the matching problem into a maximum flow problem to obtain only aggregate matching results }\citep{seo2021multi}\DIFadd{.
}\DIFaddend 

In terms of equilibrium modeling with platform operations, one of the questions is how should the matching problem be incorporated. Many studies consider an \textit{ideal} ridesharing service without any matching frictions, in which ridesharing drivers can take passengers as long as vehicle capacity permits (\cite{xu2015complementarity}; \cite{bahat2016incorporating}; \cite{di2018link}; \cite{di2019unified}; \cite{li2020restricted}). Only a handful of studies consider exogenous matching functions to account for matching frictions in ridesharing (\cite{ma2020ridesharing}; \cite{ma2022general}; \cite{noruzoliaee2022one}). However, matching functions provide aggregate performance for platforms with fixed strategies, and are not flexible enough to accommodate different operational strategies. Such flexibility is indeed important for evaluating the impacts of platform operations on traveler decisions and network congestion (\cite{liu2019framework}), and the impacts of regulatory policies on platform operations. There is a need to integrate explicitly the ridesharing matching problem in the equilibrium modeling framework.

Another challenge in equilibrium modeling with platform operations is related to ridesharing passenger pickup and drop-off handling. The existing literature on ridesharing network equilibrium can be broadly divided into link-based and path-based approaches. Link-based models \DIFaddbegin \DIFadd{typically }\DIFaddend accommodate multi-OD ridesharing by allowing passengers to be dropped off on the way, and picked up by another ridesharing driver within a hyper-network (\cite{xu2015complementarity}; \cite{di2018link}; \cite{di2019unified}). It is expected that such en-route transfers should be avoided, since they bring additional inconveniences to the ridesharing passengers. \DIFaddbegin \DIFadd{\mbox{
\textcite{mahmoudi2016finding} }\hskip0pt
and \mbox{
\textcite{liu2020integrated} }\hskip0pt
suggest enumerating all possible combinations of passenger pickup and drop-off }\textit{\DIFadd{states}} \DIFadd{within a high-dimensional hyper-network, to avoid passenger transfers in a link-based setting. }\DIFaddend Path-based models can prevent en-route transfers, but typically cannot handle ridesharing with travelers from multiple OD pairs (\cite{bahat2016incorporating}; \cite{li2020restricted}; \cite{ma2020ridesharing}; \cite{ma2022general}). Only a few path-based models relax the same-OD sharing setting to a certain extent. \textcite{li2020path} consider a multi-OD sharing scheme, in which ridesharing driver's route connects multiple same-OD passenger trips, i.e., passengers with different ODs cannot travel at the same time. In contrast, \textcite{chen2022unified} considers the case that ridesharing passengers cannot have the same origin. \textcite{noruzoliaee2022one} consider a more general case that travelers from different OD pairs can travel together in a \textit{transit-style} ride-pooling vehicle with predetermined routes. It is necessary to develop new models for multi-OD ridesharing without transfer, which also determine ridesharing driver's routes endogenously.

The other challenge in equilibrium modeling with platform operations is how should the interactions between travelers and ridesharing matching be considered. Conventional models assume ridesharing drivers and passengers accept ridesharing matching by default (\cite{alonso2017demand}; \cite{di2019unified}; \cite{ma2022general}). However, empirical studies have shown matching acceptance/rejection behaviors could have significant impacts on ridesharing system performances (\cite{alonso2021determinants}; \cite{ashkrof2022ride}; \cite{yao2022}). There are only a few studies that consider matching stability between passengers and drivers, in which a matching is considered stable if no driver or passenger prefers to be matched together other than their current matching. \textcite{wang2018stable} considers the problem of finding stable one-to-one matching from the platform operational perspective, in which ridesharing drivers and passengers are assigned with optimal cost-saving matching. \textcite{peng2022many} extends it for the stable many-to-one ride-pooling matching problem. \textcite{noruzoliaee2022one} takes into account the passenger matching stability in an equilibrium model for ride-pooling services. \textcite{li2020path} model the stable matching between ridesharing driver and same-OD passengers. It is important to consider matching stability in multi-OD multi-passenger ridesharing.

\DIFaddbegin \textbf{\DIFadd{Equilibrium solution.}} \DIFadd{Incorporating ridesharing into network models creates challenges for solving the equilibrium problem. The ridesharing network equilibrium models are related to the classic traffic assignment problems (TAP), but they include additional ridesharing constraints such as pickup and drop-off. Among assignment algorithms developed for solving TAP, Bush-based algorithms are one of the most efficient methods. These algorithms exploit the acyclic property of equilibrium link flows }\citep[][]{bar2002origin,nie2010class} \DIFadd{and avoid path enumeration with a compact link-based bush representation }\citep{dial2006path}\DIFadd{. Recently, \mbox{
\textcite{xu2022hyperbush} }\hskip0pt
extend the bush approach for solving transit assignment problems with hyper-bush. 
}

\DIFadd{However, the coupling ridesharing constraints make it challenging to apply existing TAP algorithms directly. These algorithms typically require decomposing the TAP into independent OD-based or origin-based subproblems, known as }\textit{\DIFadd{block}}\DIFadd{. 
The classical Augmented Lagrangian (AL) method addresses this issue by dualizing the coupling constraints into the objective function to allow decomposition, and adding quadratic penalties for robustness~}\citep{hestenes1969multiplier,powell1969method}\DIFadd{. The AL method iteratively solves the optimal primary variables for the given set of multipliers (also called the }\textit{\DIFadd{restricted problem}}\DIFadd{) and updates the multipliers, for which local convergence is established without requiring convexity~}\citep{rockafellar1974augmented,bertsekas1982constrained,kanzow2016augmented}\DIFadd{. 
}

\DIFadd{Nonetheless, solving large }\textit{\DIFadd{restricted problems}} \DIFadd{all at once can still be difficult. To tackle this issue, block coordinate descent (BCD) methods are combined to solve the AL problem, which, similar to Gauss-Seidel method, successively update each variable block while keeping other blocks fixed~}\citep{bertsekas1997nonlinear}\DIFadd{. One BCD variant is the Alternating Direction of Multiplier Method (ADMM), which simplifies the AL method by using only one pass of block minimization to approximately solve the }\textit{\DIFadd{restricted problem}} \DIFadd{at each iteration~}\citep{eckstein2012augmented,boyd2011distributed,parikh2014proximal}\DIFadd{. ADMM's simplicity has made it a popular choice for various transportation problems, such as vehicle routing and crew scheduling~}\citep[e.g.,][]{liu2020integrated,yao2019admm,feng2023admm}\DIFadd{. However, ADMM has limitations. It cannot guarantee convergence for problems with more than two variable blocks~}\citep{chen2016direct}\DIFadd{, and it requires strong convexity assumptions on the objective function~}\citep{boyd2011distributed,parikh2014proximal}\DIFadd{. 
}

\DIFadd{In contrast to ADMM, another type of BCD, the block proximal descent (BPD) method guarantees local convergence without these restrictive assumptions~}\citep{auslender1992asymptotic,grippo2000convergence,tseng2001convergence,bolte2014proximal}\DIFadd{. The BPD method essentially adds proximal regularization~}\citep{parikh2014proximal} \DIFadd{for updating each variable block, and performs multiple passes of block minimization to solve the }\textit{\DIFadd{restricted problem}} \DIFadd{optimally. As a result of this regularization, each update of a variable block takes only one gradient step, instead of jumping to the }\textit{\DIFadd{minimal}}\DIFadd{~}\citep{bolte2014proximal}\DIFadd{. Several studies in machine learning and large-scale optimization have shown that proximal regularization improves solution stability and reduces overall runtime~}\citep[e.g.,][]{beck2009fast,combettes2011proximal}\DIFadd{. It is worth noting that some state-of-the-art traffic assignment algorithms~}\citep[e.g.,][]{dial2006path,bar2002origin} \DIFadd{also implement a BPD-style solution scheme, even though they are not explicitly classified as BPD.
}

\DIFadd{Most existing studies solve the equilibrium problems using general-purpose solvers like GAMS \citep[e.g.,][]{di2018link,ban2019general,li2020path,chen2022unified}, there is no network assignment algorithm in the literature developed for solving multi-passenger ridesharing network equilibrium problems}\DIFadd{. However, we could not find a network assignment algorithm that is able to solve this problem. Therefore, there is a need to develop such an algorithm for solving ridesharing network equilibrium problems in a real-size network.
}

\DIFaddend 

A summary of existing ride-sourcing network equilibrium models is given in \cref{tab:literture}. 

\begin{table}[H]
    \caption{\\A comparison between the proposed and existing ride-sourcing network equilibrium models.}
    \label{tab:literture}
    \centering
    \scriptsize
    \setlength\tabcolsep{8pt}
    \begin{threeparttable}
        \begin{tabular}{p{0.14\linewidth}
                        *{8}{S[table-format=0.8]}
                       }
        \toprule
        {\multirowcell{3}{\textbf{Literature}}}
           & \multicolumn{3}{l}{\textbf{\makecell[l]{Modeling components$^{*}$}}}
              & \multicolumn{4}{l}{\textbf{\makecell[l]{Modeling features}}} 
                & {\multirowcell{3}{\textbf{Solution}\\\textbf{algorithm}}$^{***}$}\\
        \cmidrule{2-4}
        \cmidrule{5-8}
        & {\multirowcell{2}{Driver}} 
        & {\multirowcell{2}{Passenger}}  
        & {\multirowcell{2}{Platform \\ operations}}
        & {\multirowcell{2}{Multiple \\ OD sharing}}
        & {\multirowcell{2}{Avoid \\ transfer}}
        & {\multirowcell{2}{Explicit \\ matching}}
        & {\multirowcell{2}{Matching \\ stability}}
        \\
        \\
        \midrule
        \multicolumn{8}{l}{\textbf{Ride-hailing}}                                                                                                                                                                                                                                                                                                                                                                                                         \\
\textcite{ban2019general}                     &                 & \checkmark                  & \checkmark                                                                     & NR$^{**}$                                                                      & \checkmark                                                                  & \checkmark                                                                     & NR$^{**}$                                                                  &\\
\textcite{di2019unified}                     &                 & \checkmark                  & \checkmark                                                                     & NR$^{**}$                                                                      & \checkmark                                                                  & \checkmark                                                                     & NR$^{**}$                                                                  &\\
\multicolumn{8}{c}{}                                                                                                                                                                                                                                                                                                                                                                                                                              \\
\multicolumn{8}{l}{\textbf{Ride-pooling}}                                                                                                                                                                                                                                                                                                                                                                                                         \\
\textcite{noruzoliaee2022one}            &                 & \checkmark                  & \checkmark                                                                     & \checkmark                                                                       & \checkmark                                                                  &                                                                       & NR$^{**}$                                                                   &\\
\textcite{chen2022unified}                    &                 & \checkmark                  & \checkmark                                                                     & \checkmark                                                                       & \checkmark                                                                  & \checkmark                                                                     & NR$^{**}$                                                                  &\\
\multicolumn{8}{c}{}                                                                                                                                                                                                                                                                                                                                                                                                                              \\
\multicolumn{8}{l}{\textbf{Ride-sharing}}                                                                                                                                                                                                                                                                                                                                                                                                         \\
\textcite{xu2015complementarity}                      & \checkmark               & \checkmark                  &                                                                       & \checkmark                                                                       &                                                                    &                                                                       &                                                                     &\\
\textcite{bahat2016incorporating}                & \checkmark               & \checkmark                  &                                                                       &                                                                         & \checkmark                                                                  &                                                                       &                                                                     &\\
\textcite{di2018link}                      & \checkmark               & \checkmark                  &                                                                       & \checkmark                                                                       &                                                                    &                                                                       &                                                                     &\\
\textcite{di2019unified}                    & \checkmark               & \checkmark                  &                                                                       & \checkmark                                                                       &                                                                    &                                                                       &                                                                     &\\
\textcite{li2020restricted}                      & \checkmark               & \checkmark                  &                                                                       &                                                                         & \checkmark                                                                  &                                                                       &                                                                     &\\
\textcite{ma2020ridesharing}                     & \checkmark               & \checkmark                  & \checkmark                                                                     &                                                                         & \checkmark                                                                  &                                                                       &                                                                     &\\
\textcite{li2020path}                      & \checkmark               & \checkmark                  &                                                                       & \checkmark                                                                       & \checkmark                                                                  &                                                                       & \checkmark                                                                   &\\
\textcite{ma2022general}                      & \checkmark               & \checkmark                  & \checkmark                                                                     &                                                                         & \checkmark                                                                  &                                                                       &                                                                     &\\
\\
\textbf{This paper}                   & \textbf{\checkmark}      & \textbf{\checkmark}         & \textbf{\checkmark}                                                            & \textbf{\checkmark}                                                              & \textbf{\checkmark}                                                         & \textbf{\checkmark}                                                            & \textbf{\checkmark}                                                        &\textbf{\checkmark}\\                              
        \bottomrule
        \end{tabular}
    \begin{tablenotes}
        \item[*]All network equilibrium models consider network congestion.
        \item[**]NR: {\textit{Not relevant} for the service type.}
        \item[***]: Solution algorithm for multi-passenger ridesharing equilibrium problems.
    \end{tablenotes}
    \end{threeparttable}
\end{table}

As indicated in \DIFdelbegin \DIFdel{the table}\DIFdelend \DIFaddbegin \DIFadd{\mbox{
\cref{tab:literture}}\hskip0pt
}\DIFaddend , there is no comprehensive equilibrium model that accounts explicitly \DIFaddbegin \DIFadd{for }\DIFaddend the platform matching operations, multi-OD multi-passenger ridesharing without transfer, and the stability in ridesharing matching. Moreover, there is no solution algorithm dedicated \DIFdelbegin \DIFdel{for }\DIFdelend \DIFaddbegin \DIFadd{to the }\DIFaddend multi-passenger ridesharing equilibrium problems. The objectives and contributions of the paper are summarized in the following subsection.

\subsection{Objective and contributions}
\label{sec:contributions}
We propose a general network-based equilibrium model for multi-OD multi-passenger ridesharing systems\DIFdelbegin \DIFdel{, in which traveler }\DIFdelend \DIFaddbegin \DIFadd{. The proposed model endogenously accounts for driver and passenger }\DIFaddend decisions (i.e., mode choice, route choice, and matching choice)\DIFdelbegin \DIFdel{of both drivers and passengers, as well as }\DIFdelend \DIFaddbegin \DIFadd{, }\DIFaddend platform operations, \DIFdelbegin \DIFdel{are considered to influence traffic patterns in the network}\DIFdelend \DIFaddbegin \DIFadd{as well as their influences on traffic congestion}\DIFaddend . To the best of our knowledge, this paper is one of the first studies that explicitly consider \DIFdelbegin \DIFdel{multi-passenger ridesharing matching in ridesharing equilibrium models}\DIFdelend \DIFaddbegin \DIFadd{platform matching decisions in a multi-passenger ridesharing equilibrium problem}\DIFaddend . Moreover, \DIFdelbegin \DIFdel{the proposed model }\DIFdelend \DIFaddbegin \DIFadd{in contrast to most existing models~}\citep[e.g.,][]{xu2015complementarity,di2019unified}\DIFadd{, the proposed link-based formulation }\DIFaddend can accommodate multi-OD ridesharing without the need for en-route transfer\DIFdelbegin \DIFdel{, and consider the matching stability in the presence of multi-OD ridesharing. }\DIFdelend \DIFaddbegin \DIFadd{. This is achieved by integrating matching sequence with hyper-network, instead of enumerating pickup and drop-off combinations~}\citep[as required in][]{mahmoudi2016finding,liu2020integrated}\DIFadd{. Furthermore, we extend matching stability for multi-OD multi-passenger ridesharing in terms of matching sequence. Existing studies only considered stable matching for either ride-pooling services (i.e., cooperative dedicated drivers) or same-OD ridesharing. Our stability condition explicitly considers the driver and passenger detours for taking peer passengers at different locations, which depends not only on the matching }\citep{peng2022many} \DIFadd{but also the matching sequence.  The stable matching problem is then cast and solved as a route choice problem in the hyper-network, which allows ridesharing disutilities to be determined endogenously.
}\DIFaddend 

To capture the interactions between three entities: travelers, platform, and transportation network, the proposed equilibrium problem is composed of multiple players optimizing their own objectives within a non-cooperative game framework. For example, travelers are assumed to minimize their travel disutilities while fulfilling their travel demands, and ridesharing platforms aim at maximizing their company objectives. Travelers' ridesharing decisions influence platform operation decisions, which in turn impact the modal costs of the travelers. The proposed equilibrium model considers that platform matching decisions are constrained by the ridesharing demands, which makes the resulting model a generalized Nash equilibrium (\cite{rosen1965existence}, \cite{facchinei2007generalized}). \textcite{ban2019general} establishes the solution existence of a generalized equilibrium model for ride-hailing services, by using a penalty-based method to handle the complex coupling constraints (which makes it difficult to directly apply fixed-point theorems, \cite{Nagurney2009}). In this paper, we extend the work of \textcite{ban2019general} to show solution existence for the proposed multi-passenger ridesharing general equilibrium model. 

A sequence-bush assignment algorithm is developed in this paper for solving \DIFaddbegin \DIFadd{the }\DIFaddend multi-passenger ridesharing equilibrium \DIFdelbegin \DIFdel{problems}\DIFdelend \DIFaddbegin \DIFadd{problem}\DIFaddend . The proposed algorithm \DIFdelbegin \DIFdel{adapts the }\DIFdelend \DIFaddbegin \DIFadd{decomposes the problem into simpler subproblems, and solves each subproblem efficiently. Specifically, in contrast to existing decompositions that are based on a single vehicle~}\citep[e.g.,][]{liu2020integrated,yao2019admm}\DIFadd{, we apply Augmented Lagrangian method~}\citep{kanzow2016augmented} \DIFadd{to decompose the network problem for each driver-passenger group who are matched together, by exploiting the problem structure. To solve the network subproblem efficiently, we integrate matching sequences with }\DIFaddend bush-based \DIFdelbegin \DIFdel{algorithm for efficiently handling traveler routing problems between two tasks (pickup/drop-off), and combines a matching-sequence hyper-network approach to tackle the ridesharing constraints}\DIFdelend \DIFaddbegin \DIFadd{algorithms~}\citep{bar2002origin,dial2006path,nie2010class}\DIFadd{, and solve them with the block proximal algorithm~}\citep{bolte2014proximal}\DIFadd{. As a result, the complex ridesharing constraints are handled implicitly in the sequence-bush, and solving the network subproblem is similar to a classic TAP}\DIFaddend .  

The contributions of this paper are highlighted as follows:
\begin{itemize}
   \item Matching stability is extended for multi-OD multi-passenger ridesharing in \DIFdelbegin \DIFdel{our network equilibrium model, in which ridesharing driver and passenger matching-option preferences are cast }\DIFdelend \DIFaddbegin \DIFadd{terms of matching sequences, which particularly considers the differences in ridesharing detours due to variations in matching sequences. Moreover, the stable matching problem is cast and solved }\DIFaddend as a route choice problem \DIFdelbegin \DIFdel{within }\DIFdelend \DIFaddbegin \DIFadd{in }\DIFaddend the hyper-network. 
\item
 A sequence-bush algorithm is developed for solving the multi-passenger ridesharing equilibrium problem\DIFdelbegin \DIFdel{, in which bushes are connected in sequence to handle ridesharing constraints}\DIFdelend \DIFaddbegin \DIFadd{. The proposed algorithm decomposes the network problem based on ridesharing travelers who are matched together, which exploits the coupling constraints in multi-passenger ridesharing. Furthermore, matching sequences are integrated with bush-based algorithms to implicitly handle the complex ridesharing constraints, such that solving the assignment subproblem is simpler}\DIFaddend .
    \item \DIFdelbegin \DIFdel{A }\DIFdelend \DIFaddbegin \DIFadd{Hyper network is integrated with }\DIFaddend  matching sequences \DIFdelbegin \DIFdel{approach is developed to handle }\DIFdelend \DIFaddbegin \DIFadd{to model transfer-free }\DIFaddend multi-OD multi-passenger ridesharing \DIFdelbegin \DIFdel{. Specifically, by applying the }\DIFdelend \DIFaddbegin \DIFadd{in a link-based formulation. In contrast to existing methods, the proposed expanded }\DIFaddend hyper-network \DIFdelbegin \DIFdel{, the requirement for making transfers in multi-OD ridesharing is relaxed, and ridesharing driver routes can be determined endogenously}\DIFdelend \DIFaddbegin \DIFadd{is compact: its size increases linearly with respect to the number of pickups and drop-offs in a matching sequence}\DIFaddend . 
   \item \DIFdelbegin \DIFdel{A comprehensive planning tool for evaluating the impacts of traveler decisions, ridesharing matching on network congestion, and vice versa}\DIFdelend \DIFaddbegin \DIFadd{Solution existence is established for multi-passenger ridesharing general equilibrium under mild assumptions, by adapting the penalty-based method of \mbox{
\textcite{ban2019general}}\hskip0pt
}\DIFaddend .
\end{itemize}
The rest of the paper is organized as follows: \cref{sec:Setup} presents the general equilibrium scheme for multi-passenger ridesharing systems. \cref{sec:Hypernet} introduces the hyper-network approach for modeling the ridesharing system; this is followed by the traveler mode choice model in \cref{sec:ModeChoice}, and the platform matching decisions in \cref{sec:Matching}. The matching stability and route choice problems are jointly handled as a network problem in \cref{sec:RidesharingNetworkModel}. The overall equilibrium model is summarized in \cref{sec:Overview}, which also includes a preview of the proof of solution existence. \DIFdelbegin \DIFdel{\mbox{
\cref{sec:Existence} }\hskip0pt
presents this proof in detail. }\DIFdelend The sequence-bush assignment method is developed in \cref{sec:algorithm}. \DIFaddbegin \DIFadd{The contributions in model formulation and solution algorithm are illustrated with an example in \mbox{
\cref{sec:illustration}}\hskip0pt
.
}\DIFaddend We present selected model characteristics and numerical results \DIFaddbegin \DIFadd{with a real-size network }\DIFaddend in \cref{sec:Results}, discuss and summarize the paper in \cref{sec:Discussion} and \cref{sec:Summary}. \DIFaddbegin \DIFadd{The list of key mathematical notations used in the paper is provided in Appendix~\cref{tab:index_list}.}\DIFaddend

\section{A general equilibrium modeling scheme}
\label{sec:Setup}

\begin{figure}[H]
    \centering
    \includegraphics[width=0.6\textwidth]{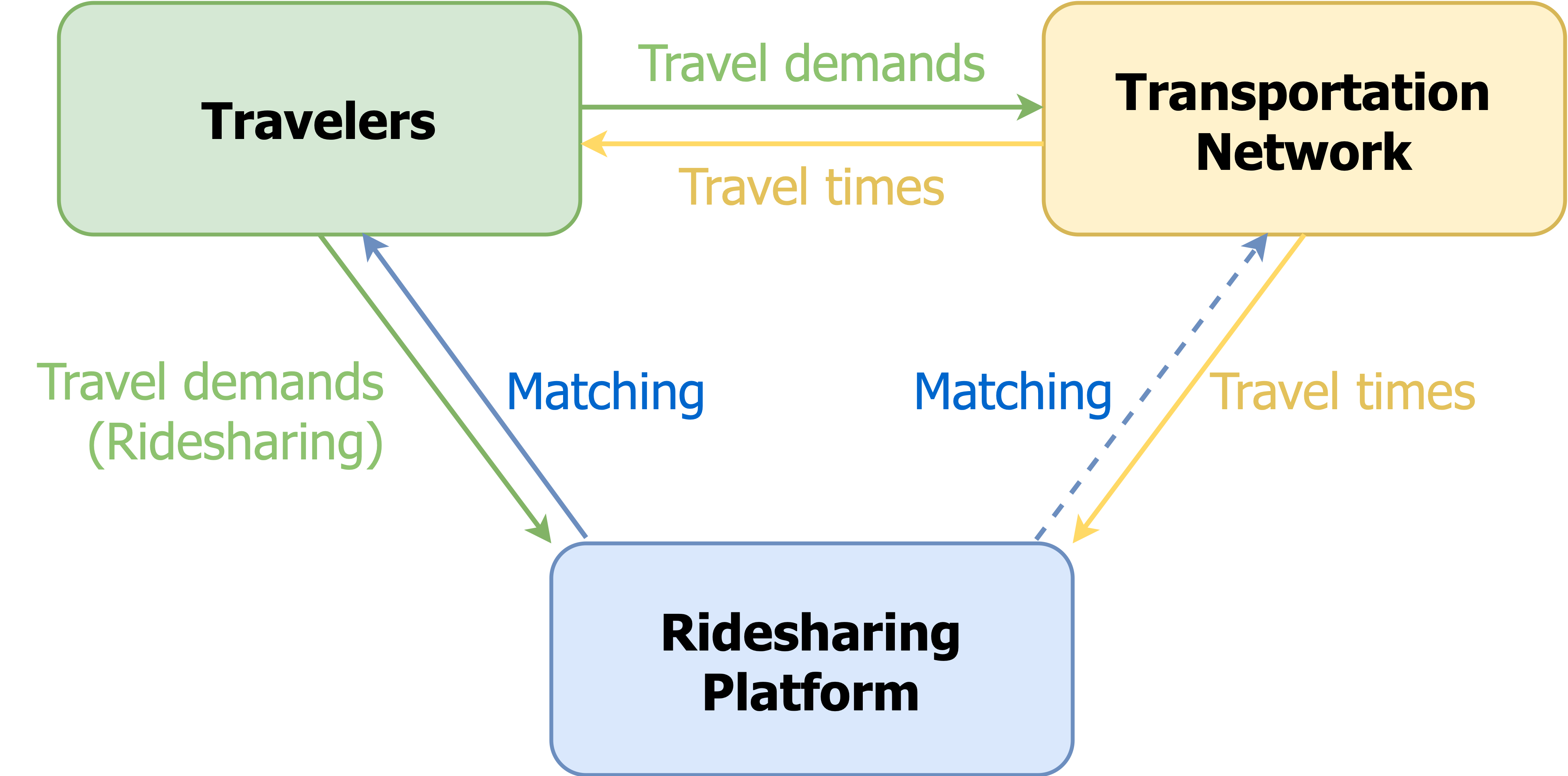}
    \caption{A general equilibrium modeling scheme with interactions between three entities}
    \label{fig:1.general_scheme}
\end{figure}

\DIFdelbegin \DIFdel{A }\DIFdelend \DIFaddbegin \DIFadd{As }\DIFaddend shown in \cref{fig:1.general_scheme}, we consider a general ridesharing system including three entities: \textbf{travelers} (drivers and passengers), \textbf{ridesharing platform}, and \textbf{transportation networks}, in which these entities make various decisions and interact with each other. For example, travelers make trips (travel demands) in the transportation network and experience traffic congestion (travel times). Ridesharing drivers and passengers requesting shared rides are matched by the platform based on ridesharing travel demands and travel times in the network. Consequently, ridesharing platform indirectly influences travel times in the network by matching ridesharing drivers and passengers. Furthermore, we consider a ridesharing service that allows: 1) a ridesharing driver to take passengers besides from the driver's own OD pair; 2) multiple ridesharing passengers from the same/different OD pairs to travel together without any transfer in a ridesharing vehicle. The following paragraphs detail the decisions of each entity and their interactions, which are also summarized in \cref{fig:2.process}.  

\textbf{\textit{Mode choice.}} We assume one traveler chooses the travel mode with minimum travel disutility among four alternatives: driving alone with a private vehicle (DA); ridesharing as a driver (RD); ridesharing as a passenger (RP); or taking public transport (PT). 

\textbf{\textit{Platform operations.}} Ridesharing platform then performs ridesharing matching between RD and RP and provides both RD and RP several options of matching sequence (for pickup and drop-off), and  limits the maximum number of matched RD and RP for each matching sequence.

\textbf{\textit{Matching stability}} We consider a matching is stable if the matched ridesharing driver and the ridesharing passengers consider the proposed matching as the best choice \textit{they can get}(\cite{li2020path}). To accommodate stable matching, we consider a flexible ridesharing setting where RD and RP choose either one of the matching-sequence options or even quit ridesharing, to minimize their generalized travel costs. In case of quitting ridesharing, travelers, who previously decided not to drive (i.e., chose RP), are assumed to take public transport. Similarly, for RD leaving ridesharing, they are assumed to continue driving as DA. 

\textbf{\textit{Route choice}} In the networks, DA and PT traverse their own selected routes, while ridesharing routes are jointly determined by RD and RP. It is assumed that only private vehicles (DA) and ridesharing vehicles (RD) contribute to network congestion, and public transport \DIFdelbegin \DIFdel{are }\DIFdelend \DIFaddbegin \DIFadd{is }\DIFaddend operated in a segregated system (e.g., BRT and metro). Furthermore, we consider a setting where RDs are only required to follow their chosen pickup and drop-off sequences, and can decide their routes between pickup/drop-off tasks. As a result, congestion costs affect the travel mode decision of all travelers, and matching-sequence choice of ridesharing travelers (RD, RP).

\DIFdelbegin 
{
\DIFdelFL{Model components and interactions}}

\DIFdelend We consider a general equilibrium, at which:
\begin{itemize}
   \item (\textit{Mode choice}) No traveler can reduce his/her generalized travel disutility by unilaterally switching his/her travel mode;
   \item (\textit{Platform operations}) Ridesharing platform cannot improve its matching objective by unilaterally switching its ridesharing driver-passenger matching sequences, and the limits on the maximum number of matched RD and RP for each matching sequence;
   \item (\textit{Matching stability}) No ridesharing traveler (RD, RP) can reduce his/her generalized travel costs by unilaterally switching his/her choice of matching sequence;
   \item (\textit{Route choice}) No traveler can reduce his/her generalized travel costs by unilaterally switching his/her route. 
\end{itemize}

\begin{figure}[H]
    \centering
    \includegraphics[width=0.65\textwidth]{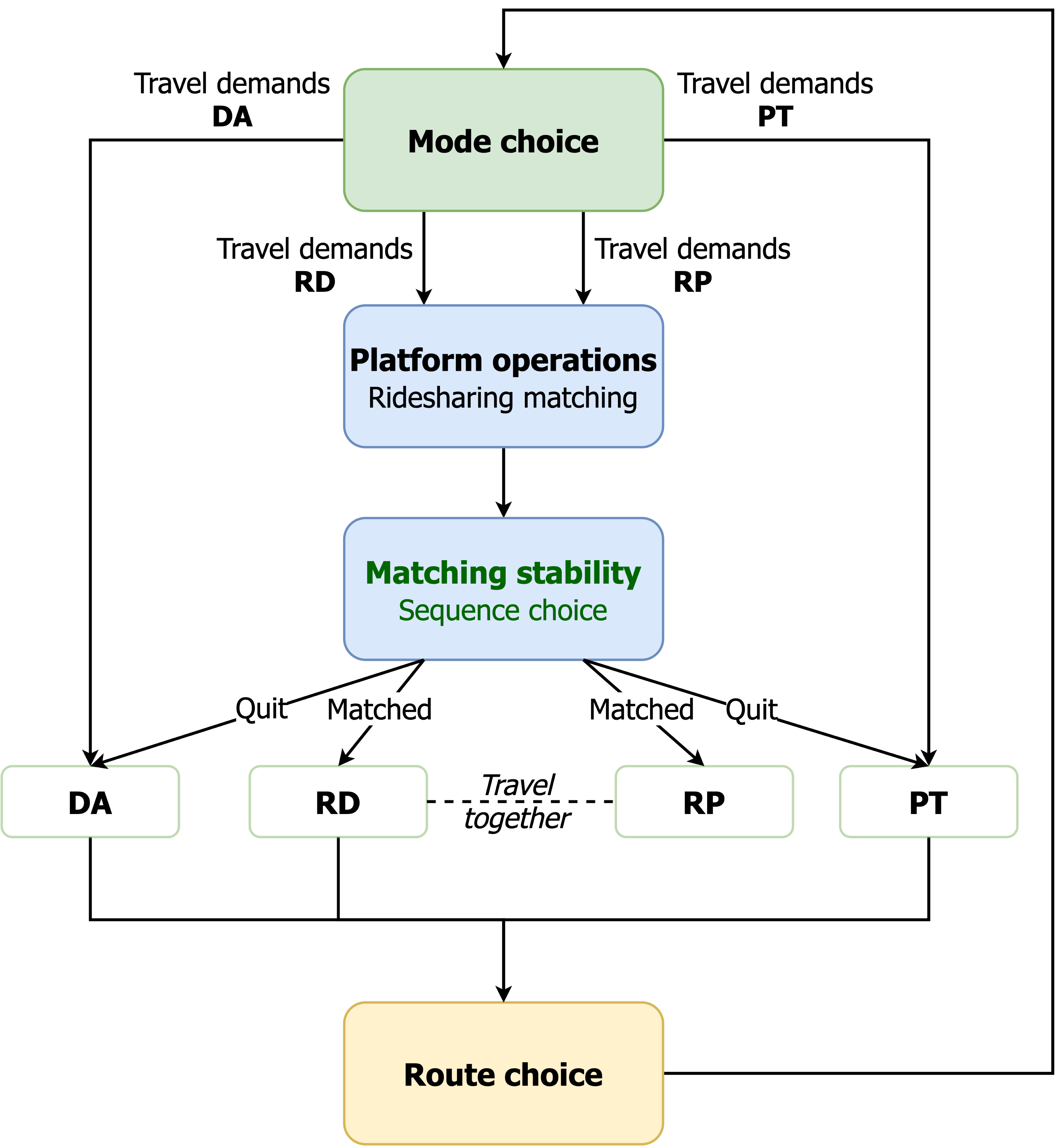}
    \caption{{Model components and interactions}}
    \label{fig:2.process}
\end{figure}
\section{Hyper-network approach}
\label{sec:Hypernet}
A transportation network typically consists of nodes and links, where nodes represent origins, destinations, and intermediate stops, and links represent directed roads connecting two nodes. Individuals travel from origins to destinations by traversing links in the network. Depending on their mode choice and route choice, individuals may experience different travel costs.

To incorporate mode choice within a unified network-based framework, super-networks are introduced, which augment the original networks with virtual links to represent mode choice decisions (\cite{sheffi1985urban}). \textcite{xu2015complementarity} and \textcite{di2018link} adapt the super-network approach for a multi-modal equilibrium problem with self-organized ridesharing (i.e., no ridesharing platform) by duplicating links for each mode, where passengers are assumed to make transfers. 

As outlined in the introduction section, we propose a network model that not only determines how travelers will choose their \textbf{modes}, \textbf{routes}, but also their choices of \textbf{matching sequences}. To relax the requirement for making transfers in multi-passenger ridesharing and to capture the ridesharing travelers' preferences for matching sequences, a multi-modal hyper-network is developed. This hyper-network is constructed by extending the network with additional levels (layers) to represent matching sequences. The proposed approach embeds the stable matching and route choice problem within a network modeling framework, which requires a detailed representation of the problems. Therefore, before introducing the hyper-network, we present the definitions and notations of matching sequences, which are the key building blocks of the proposed hyper-network.

\subsection{Matching sequence}
\label{sec:Matching sequence}
We consider a ridesharing setting that, ridesharing drivers depart from their origins, and finish their trips at their destinations. Similarly, a ridesharing passenger is picked up from his/her origin by a driver, and dropped off at his/her destination by the same driver without any transfer (\cite{yao2021dynamic}). 

Under such setting, RD's ridesharing itinerary is defined as a \textit{feasible matching sequence} denoted as $n$, which starts from one driver's origin $o$, and finishes at the driver's destination $d$, with each served RP trip being the \textit{subsequence} of $n$. Specifically, a matching sequence $n$ is served by only one RD and composed of several tasks, in which the $l^{th}$ task corresponds to one of the four task types: 1) departing from the driver's origin $o$; 2) arriving at the driver's destination $d$; 3) picking up one passenger from his/her origin $o$; 4) dropping off one passenger at his/her destination $d$. A matching sequence $n$ is considered feasible if a) the RD's vehicle occupancy never exceeds the given vehicle capacities at any point in $n$; b) all RP in $n$ are picked up at their origins and only dropped off at their destinations (i.e., without transfer); and c) the corresponding RD departs from his/her origin as the first task and reaches his/her destination as the last task.

\begin{figure}[H]
    \centering
    \includegraphics[width=1\textwidth]{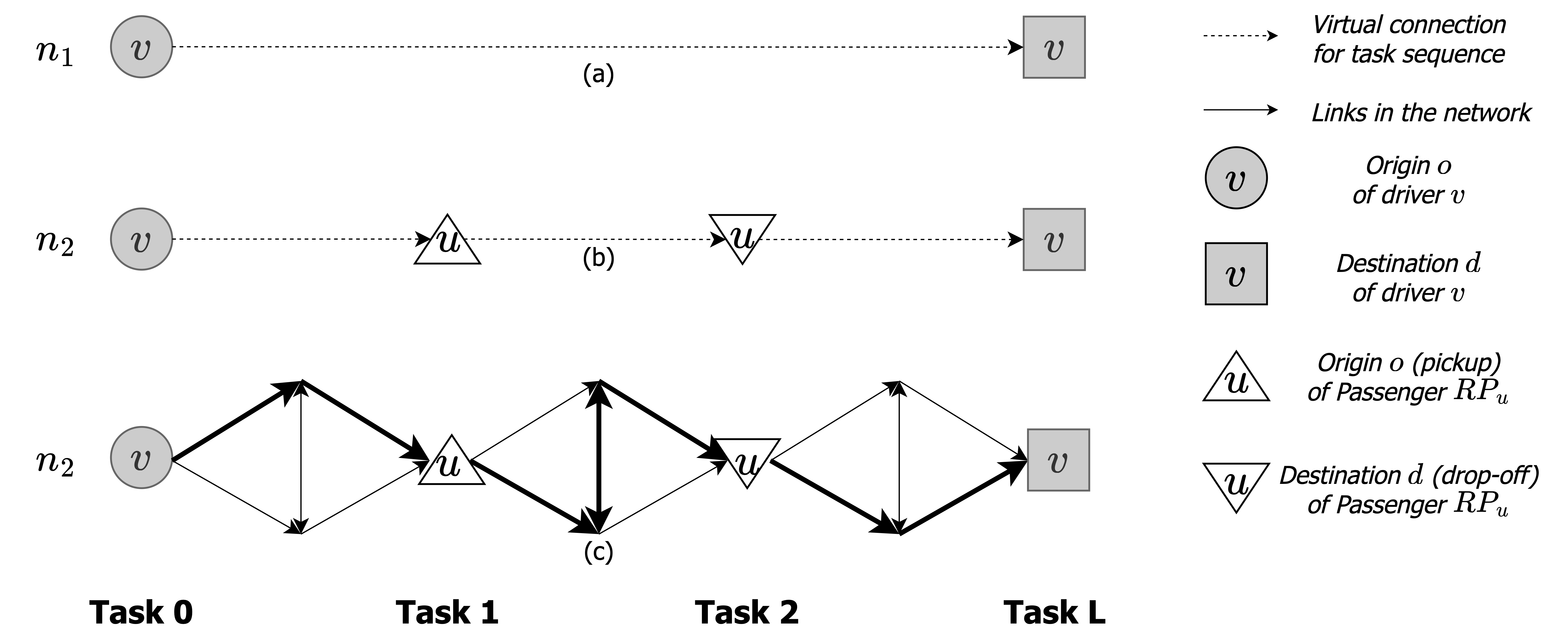}
    \caption{Matching sequences. (a) Sequence without any pickup/drop-off; (b) sequence with one pickup and drop-off; (c) Routing example of a sequence}
    \label{fig:3.matching_seqeunce}
\end{figure}

We illustrate the matching sequences for a toy example of 1 RD and 1 RP in \cref{fig:3.matching_seqeunce}, in which the RD has two feasible matching sequences $n_1$ representing traveling alone (\cref{fig:3.matching_seqeunce}a), and $n_2$ representing serving 1 RP (\cref{fig:3.matching_seqeunce}b). The virtual connection between two consecutive tasks represents RD (and RP) traverse links in the network to the next task, where the actual route is illustrated in \cref{fig:3.matching_seqeunce}(c). 

The proposed matching sequence representation resembles the bilevel structure adapted for solving dial-a-ride problems (DARP) \citep{alonso2017demand,yao2021dynamic}, where ridesharing platforms solve the optimal routing of a matching sequence at the lower level, and determine the system optimal matching at the upper level. However, in this paper, such approach is generalized and applied for ridesharing travelers (RD and RP) to determine their user optimal routes and matching sequences (as detailed in \cref{sec:RidesharingNetworkModel}). 

From \DIFaddbegin \DIFadd{a }\DIFaddend modeling perspective, the proposed matching sequence approach implies RD can freely choose their routes between two consecutive tasks, as long as they follow the matching sequence. This routing flexibility could be more favorable for ridesharing drivers, as opposed to dedicated drivers who are required to follow centralized routing decisions from their employed ridesharing platforms.

From \DIFaddbegin \DIFadd{an }\DIFaddend implementation perspective, by adapting a link-based route choice formulation, the proposed matching sequence approach avoids path set enumeration and reduces the number of decision variables. As illustrated in \cref{fig:3.matching_seqeunce}(c), matching sequence $n_2$ has in total of 64 simple paths (3 segments and each with 4 paths $=4^3=64$), compared to in total of only 18 links. Such dimension reduction would be more significant with the increase of sequence length, and could intensify the solution computations.

We now introduce the notations for matching sequences that are needed for the formulations later. Let $\omega$ denote an OD pair of the ridesharing drivers and passengers, and $\omega:(o, d)$, each matching sequence $n$ is defined by two incident matrices, where $s_1$ represents pickup tasks and $s_{-1}$ represents drop-off tasks. If $s_1(n,l,\omega)=1$, it represents the $l^{th}$ task of a matching sequence $n$ is to pick up a passenger at the origin of $\omega$ (i.e., $o$), otherwise $s_1(n,l,\omega)=0$. Similarly, $s_{-1}(n,l,\omega)=1$ represents to drop off a passenger at the destination of $\omega$ (i.e., $d$) as the $l^{th}$ task of a matching sequence $n$. Under such definition, let $L$ \DIFdelbegin \DIFdel{denotes }\DIFdelend \DIFaddbegin \DIFadd{denote }\DIFaddend the last task, each matching sequence $n$ is associated with a driver departing from an origin $o$ (where $s_1(n,0,\omega)=1$), and arriving at a destination (where $s_{-1}(n,L,\omega)=1$). Moreover, each ridesharing traveler (RD/RP) can have more than one matching sequence associated with him/her.

Note that, the proposed matching sequence generalizes the link-based approach of \textcite{xu2015complementarity} and \textcite{di2018link} where passenger transfers are needed. The matching sequence representation collapses to the link-based approach, if assuming drop-off nodes in matching sequences can be any node in the network.

\subsection{Construction of the hyper-network}
\label{subsec:hypernet_construct}
As mentioned in \cref{sec:Setup}, we consider a ridesharing system that travelers can choose among 4 alternative modes: drive alone (DA), ridesharing driver (RD), ridesharing passenger (RP), and public transport passenger (PT). In the proposed ridesharing service, passengers need not to transfer, and both ridesharing drivers and passengers can choose their preferred matching sequences. Hence, the original network is first extended for multi-modal in \cref{subsec:multimodal_extension}, and further expanded to account for matching sequences in \cref{subsec:matching_sequence_extension}.

\subsubsection{Multi-modal extension}
\label{subsec:multimodal_extension}
We consider a ridesharing system setting that:
\begin{itemize}
    \item Drive-alone (DA) and public transport (PT) do not switch to ridesharing modes (RD/RP) during their trips;
    \item Ridesharing drivers may switch to DA before they depart from their origins. This means that there is no \textit{en-route} mode switching for ridesharing drivers;
    \item Similarly, ridesharing passengers may switch to PT before they are picked up. 
\end{itemize}

Under this setting, the original networks are duplicated for each mode with additional links representing traveler mode switching and RD pickup/drop-off, as well as virtual nodes allowing RD and RP leaving ridesharing. We illustrate the multi-modal extended network in \cref{fig:4.multimodal_network}. 

\begin{figure}[H]
    \centering
    \includegraphics[width=0.8\textwidth]{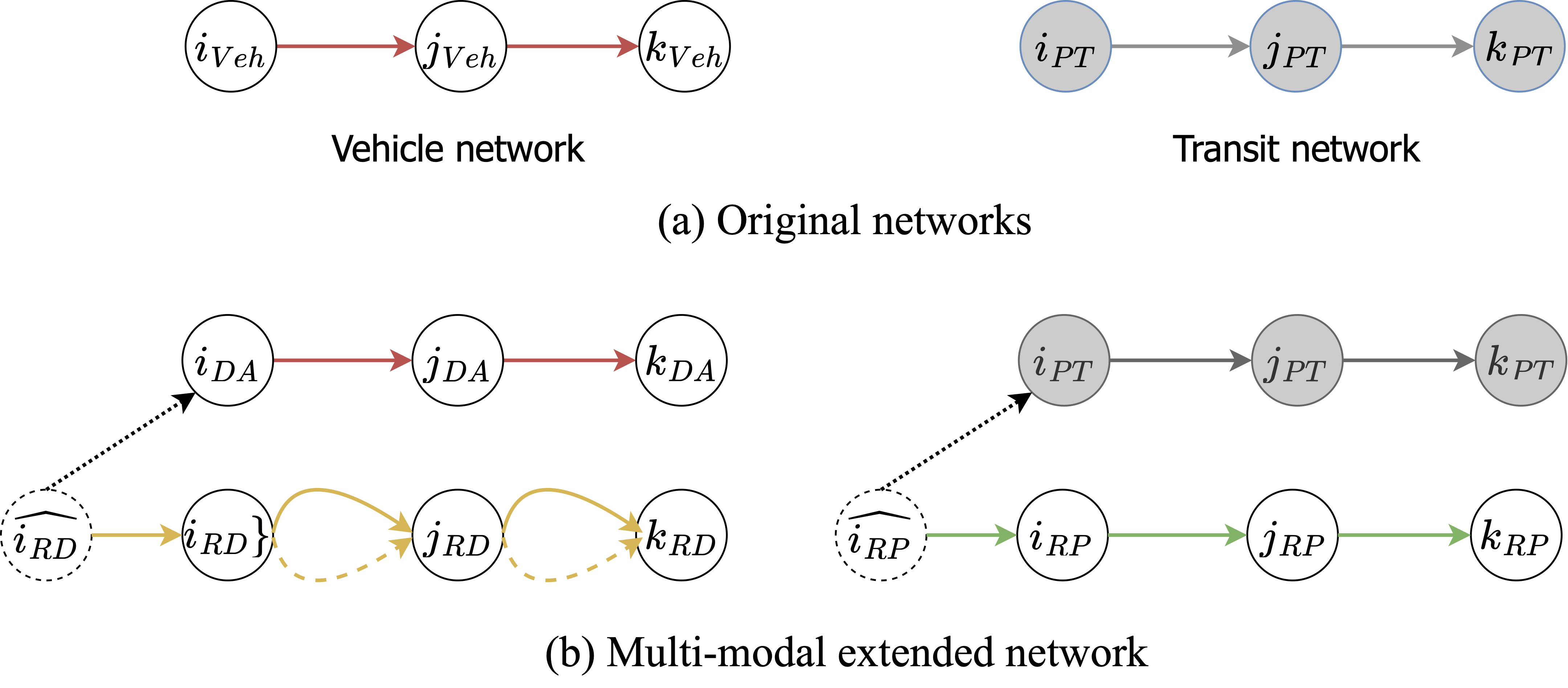}
    \caption{Multi-modal network extension}
    \label{fig:4.multimodal_network}
\end{figure}

In \cref{fig:4.multimodal_network}, virtual nodes $\widehat{i_{RD}}$ and $\widehat{i_{RP}}$ are created for all origins to represent travelers willing to join ridesharing as RD and RP, respectively. However, if they find ridesharing unattractive (e.g., due to long travel time, or even no matching), they can opt out of ridesharing through virtual links $(\widehat{i_{RD}}, i_{DA})$ and $(\widehat{i_{RP}}, i_{PT})$. RD departing from origin $i_{RD}$ starts his/her ridesharing trip to pickup/drop-off RP, during which they traverse on solid links if traveling with passenger, and dashed links if driving alone (for pickup or to reach RD's destination). While RP departs from his/her origin $i_{RP}$ only if picked up by a RD, and continues his/her trip with the matched RD.

Note that, although we do not explicitly connect the mode-specific origin nodes for demand conservations, such conservation is handled within the mode choice model shown in \cref{fig:2.process}.

\subsubsection{Matching sequence expansion}
\label{subsec:matching_sequence_extension}

In this subsection, we propose a network-based approach to model transfer-free ridesharing by incorporating matching sequence into the multi-modal extended network. The proposed matching-sequence expanded network is a compact link-based representation of \textit{all} feasible without-transfer paths belonging to matching sequences, without the need to enumerate the full path set (as in the case of path-based approach). 

Based on the observation made in \textcite{yao2021dynamic}, where matching sequences are shown to be acyclic and indicated by the ascending task index ($0, 1, ..., L$), the matching-sequence expanded ridesharing network is constructed by stacking the basic RD (or RP) subnetworks for $L$ levels. As illustrated in \cref{fig:5.hypernetwork}, each level $l$ in the expanded network represents one routing segment in the matching sequence, with virtual links connecting two consecutive levels and from virtual RD/RP origin nodes to the task level(s).
\begin{figure}[H]
    \centering
    \includegraphics[width=0.85\textwidth]{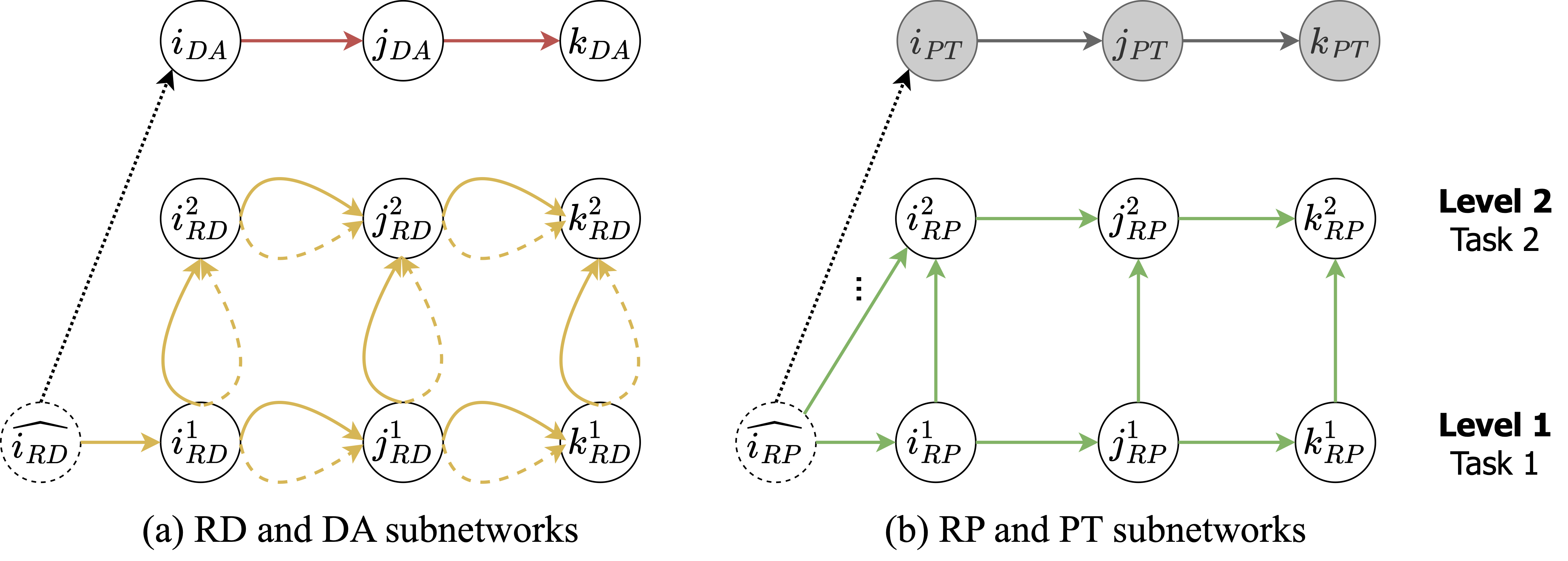}
    \caption{Hyper-network (by matching-sequence expansion)}
    \label{fig:5.hypernetwork}
\end{figure}
In \cref{fig:5.hypernetwork}(a), RD start their ridesharing trips by traversing link $(\widehat{{i}_{RD}}, i_{RD}^1)$, which means RD departs for the 1st task ($l=1$) in the matching sequence. At each level, RD traverse dashed links only if driving without passenger, otherwise, they travel with passengers on the solid links. After picking up/dropping off \DIFaddbegin \DIFadd{a }\DIFaddend passenger at some node $j_{RD}^l$, RD are required to traverse virtual link $(j_{RD}^l, j_{RD}^{l+1})$ representing finishing task $l$ and heading for the next task ($l+1$).

Differently, as shown in \cref{fig:5.hypernetwork}(b), there are multiple virtual links $(\widehat{i_{RP}}, i_{RP}^l), \forall l \in \{1,...,L-1\}$ from virtual RP origin node to multiple levels, which represent passengers from the same origin could be picked at different times (i.e., $l \neq l'$).  After being picked up at some node $j_{RP}^l$, RP traverse links $(j_{RP}^l, j_{RP}^{l+1})$ with their RD to pickup/drop-off other passengers, i.e., passengers and drivers will be coupled (detailed in \cref{sec:RidesharingNetworkModel}). Note that, in the simple case where all pickups are from the same location, the RP virtual links $(\widehat{i_{RP}}, i_{RP}^l), \forall l \in \{2,...,L-1\}$ could be redundant. However, these additional virtual links provide the possibility to impose different pickup waiting times for passengers from the same origin due to matchings.

As described above, RP transfer can be avoided by both RD and RP following transfer-free matching sequences. The matching-sequence expanded network provides a tool that allows imposing additional constraints in the network (as detailed in \cref{sec:RidesharingNetworkModel}) to ensure RD and RP comply with their chosen matching sequences. Moreover, the link-based hyper-network not only avoids path set enumeration, but also suggests a decomposable structure of the problem (hyper-network) where each level could be solved separately if cross-level constraints are satisfied. This observation motivates the development of a sequence-bush algorithm in \cref{sec:algorithm}.  

We \DIFdelbegin \DIFdel{now define the matching-sequence expanded }\DIFdelend \DIFaddbegin \DIFadd{summarize the components of the proposed }\DIFaddend hyper-network \DIFdelbegin \DIFdel{$\mathcal{G}$. Let $\mathcal{G}_{Veh}$ and $\mathcal{G}_{PT}$ denote the original vehicle and public transport networks, respectively; where $\mathcal{G}_{Veh} = (\mathcal{N}_{Veh}, \mathcal{E}_{Veh})$ and $\mathcal{G}_{PT} = (\mathcal{N}_{PT}, \mathcal{E}_{PT})$ with $\mathcal{N}_{Veh}, ~\mathcal{N}_{PT}$ and $\mathcal{E}_{Veh}, ~\mathcal{E}_{PT}$ denoting nodes and edges in $\mathcal{G}_{Veh}$ and $\mathcal{G}_{PT}$; and the set of traveler origins and destinations are denoted as $\mathcal{N}_O$  and $\mathcal{N}_D$, respectively.  The hyper-network $\mathcal{G}$ is composed of four subnetworks $\mathcal{G}_{DA}, ~\mathcal{G}_{RD}, ~\mathcal{G}_{RP}, ~\mathcal{G}_{PT}$ as summarized in \mbox{
\cref{tab:hypernet}}\hskip0pt
:
}\DIFdelend \DIFaddbegin \DIFadd{in Appendix~\ref{appendix.hyper_net_def}.
}\DIFaddend 

\DIFdel{Note that }\DIFdelend \DIFaddbegin \DIFadd{Note that several studies have adapted hyper-network for modeling ridesharing. For example, }\DIFaddend \textcite{yao2021dynamic} applied a similar \DIFdelbegin \DIFdel{matching-sequences }\DIFdelend \DIFaddbegin \DIFadd{matching-sequence }\DIFaddend approach for ridesharing matching\DIFdelbegin \DIFdel{. This previous study focused on real-time operation with the assumption of fixed travel time and correspondingly, only one fixed shortest path is taken by travelers between tasks. }\DIFdelend \DIFaddbegin \DIFadd{, but assumed fixed shortest paths. \mbox{
\textcite{mahmoudi2016finding} }\hskip0pt
and \mbox{
\textcite{liu2020integrated} }\hskip0pt
proposed to enumerate all transfer-free pickup and drop-off combinations for the }\textit{\DIFadd{state}} \DIFadd{dimension in a hyper-network, whose size grows exponentially with sequence length. 
}\DIFaddend In this study, we generalize the matching sequences for capturing \DIFaddbegin \DIFadd{the }\DIFaddend long-term impacts of traveler mode choice and matching on network congestion\DIFdelbegin \DIFdel{in terms of flow-dependent travel times and route choice}\DIFdelend \DIFaddbegin \DIFadd{, where hyper-network size grows linearly with respect to sequence length}\DIFaddend .   

\DIFdelbegin 

\DIFdelend \section{Traveler mode choice model}
\label{sec:ModeChoice}
As mentioned in the previous section, traveler flows in the hyper-network are determined by the mode choice model. Specifically, travelers are assumed to choose among 4 alternative modes prior to their departures: driving alone (DA); ridesharing as a driver (RD); ridesharing as a passenger (RP); or taking public transport (PT). Let $\mathcal{M}=\{DA, RD, RP, PT \}$ \DIFdelbegin \DIFdel{denotes }\DIFdelend \DIFaddbegin \DIFadd{denote }\DIFaddend the set of alternative modes. We consider that total travel demands for each OD pair, $\omega \coloneqq (o, d) \in \mathcal{W}$, is given as $q_{\omega} > 0$, and the mode choice model splits the total demands $q_{\omega}$ into $q_{\omega}^{DA}, q_{\omega}^{RD}, q_{\omega}^{RP}$ and $q_{\omega}^{PT}$ depending on the modal costs.

\subsection{Modal costs}
\label{sec:modal_costs}
The modal costs (i.e., travel disutility for each mode) are defined for each OD pair $\omega$ as follows:
\begin{subequations}
Drive alone (DA):
\begin{align} 
\label{eq:DA_modal_costs}
    C_{\omega}^{DA} = 
    \underbrace{\alpha^{DA} \cdot t_{\omega}^{DA}}_{\text{Travel time cost}} 
    + \underbrace{\beta \cdot d_{\omega}^{DA}}_{\text{Car operational cost}}
\end{align}
Ridesharing driver (RD):
\begin{align} 
\label{eq:RD_modal_costs}
    C_{\omega}^{RD} =  
    \underbrace{\alpha^{RD} \cdot t_{\omega}^{RD}}_{\text{Travel time cost}} 
    + \underbrace{\beta \cdot d_{\omega}^{RD}}_{\text{Car operational cost}} 
    + \underbrace{g_{\omega}^{RD}}_{\text{Inconvenience cost}}
    - \underbrace{p_{\omega}^{RD}}_{\text{Compensation to driver}}
\end{align}
Ridesharing passenger (RP):
\begin{align} 
\label{eq:RP_modal_costs}
    C_{\omega}^{RP} =  
      \underbrace{\alpha^{RP} \cdot t_{\omega}^{RP}}_{\text{Travel time cost}} 
    + \underbrace{g_{\omega}^{RP}}_{\text{Inconvenience cost}}
    + \underbrace{p_{\omega}^{RP}}_{\text{Passenger payment}}
\end{align}
Public transport passenger (PT):
\begin{align} 
\label{eq:PT_modal_costs}
    C_{\omega}^{PT} =  
      \underbrace{\alpha^{PT} \cdot t_{\omega}^{PT}}_{\text{Travel time cost}} 
    + \underbrace{g_{\omega}^{PT}}_{\text{Inconvenience cost}}
    + \underbrace{p_{\omega}^{PT}}_{\text{Passenger payment}}
\end{align}
\end{subequations}
where, $t_{\omega}^{m} (m \in \mathcal{M})$ is the travel time between OD $\omega$ using mode $m$, and $\alpha^{m}$ represents the value of time (VOT) of travelers using mode $m$; car operational cost of the drivers (DA and RD) is assumed proportional to their travel distances $d_{\omega}^{DA/RD}$ with factor $\beta$; ridesharing drivers experience inconvenience cost $g_{\omega}^{RD}$ for taking passengers, but  are compensated with $p_{\omega}^{RD}$; both ridesharing and public transport passengers experience  inconvenience $g_{\omega}^{RP/PT}$ and pay fares $p_{\omega}^{RP/PT}$ for their trips, without the need to operate a vehicle.

Modal costs defined by \cref{eq:DA_modal_costs}-\eqref{eq:PT_modal_costs} are general. In this paper, we specify the inconvenience costs as linear functions of the shared trip duration and shared distance (similar to \cite{li2020path}), while the shared portion is related to the matching sequence. We consider \DIFdelbegin \DIFdel{the whole RP/PT trip is shared (at least with the driver), and }\DIFdelend \DIFaddbegin \DIFadd{that }\DIFaddend only a portion of the RD trip is shared, where $\gamma \leq 1$ \DIFdelbegin \DIFdel{denotes the portion of shared RD trip }\DIFdelend \DIFaddbegin \DIFadd{denote the shared portion }\DIFaddend in a matching sequence. The inconvenience costs \DIFaddbegin \DIFadd{and compensations for each mode }\DIFaddend are formulated in \cref{eq:RD_inconvenience}-\eqref{eq:PT_inconvenience}:
\begin{subequations}
\begin{align}
g_{\omega}^{RD} &=  
      \gamma \cdot  (\tau_{t}^{RD} \cdot t_{\omega}^{RD}
    + \tau_{d}^{RD} \cdot d_{\omega}^{RD} ) \DIFaddbegin \quad &  
\DIFadd{p_{\omega}^{RD} }&\DIFadd{=  
    \gamma \cdot  (\nu_{t}^{RD} \cdot t_{\omega}^{RD}
    + \nu_{d}^{RD} \cdot d_{\omega}^{RD} ) }\DIFaddend \label{eq:RD_inconvenience}\\
g_{\omega}^{RP} &=  
      \tau_{t}^{RP} \cdot t_{\omega}^{RP}
    + \tau_{d}^{RP} \cdot d_{\omega}^{RP} \DIFaddbegin \quad & 
\DIFadd{p_{\omega}^{RP} }&\DIFadd{=  
      \nu_{t}^{RP} \cdot t_{\omega}^{RP}
    + \nu_{d}^{RP} \cdot d_{\omega}^{RP}
    }\DIFaddend \label{eq:RP_inconvenience} \\
g_{\omega}^{PT} &=  
      \tau_{t}^{PT} \cdot t_{\omega}^{PT}
    + \tau_{d}^{PT} \cdot d_{\omega}^{PT} \DIFaddbegin \quad & 
\DIFadd{p_{\omega}^{PT} }&\DIFadd{=  
      \nu_{t}^{PT} \cdot t_{\omega}^{PT}
    + \nu_{d}^{PT} \cdot d_{\omega}^{PT}
    }\DIFaddend \label{eq:PT_inconvenience}
\end{align}
\end{subequations}
where, \DIFaddbegin \DIFadd{for $m' \in \{RD, RP, PT\}$, }\DIFaddend $\tau_{t}^{m'}$ and $\tau_{d}^{m'}$  \DIFdelbegin \DIFdel{$(m' \in \{RD, RP, PT\})$ }\DIFdelend represent the unit inconvenience cost per unit time and unit distance, respectively\DIFaddbegin \DIFadd{; and }\DIFaddend $\nu_{t}^{m'}$ and $\nu_{d}^{m'}$ \DIFdelbegin \DIFdel{$(m' \in \{RD, RP, PT\})$ }\DIFdelend represent the unit ridesharing price per unit time and unit distance, respectively. 
\DIFaddbegin 

\DIFaddend Note that, there exist different functional forms for inconvenience costs and ridesharing prices in the literature. For example, \textcite{di2019unified} assume non-linear RD inconvenience (link) cost and compensation functions to account for different occupancies. However, as shown in \textcite{ban2019general}, such non-linearity could bring further challenges for showing solution existence. Our proposed hyper-network model embeds the vehicle occupancy in the matching sequence, where occupancy-specific costs can be accommodated by augmenting links with different occupancies and associating different costs (instead of only with/without passenger). The linear inconvenience costs and ridesharing prices are considered here for simplicity.

\DIFdelbegin 

\DIFdelend \subsection{Mode choice}
\label{sec:mode_choice}
As shown in \cref{fig:2.process}, the mode choice model considers the travel times, distances, and matching sequences as input to the model. The proposed modal costs \cref{eq:DA_modal_costs}-\eqref{eq:PT_modal_costs} resemble such setting, in which all the variables in the modal costs are determined by the ridesharing platform's matching decisions (e.g., $\gamma$) and route choice decisions in the network model ($t_{\omega}^{m}, d_{\omega}^{m}$).

Given the modal costs, the traveler's mode choice problem can be formulated as the following mixed complementarity conditions (MCP) (\cite{facchinei2003finite}, \cite{Nagurney2009}):

\DIFdelbegin 

\DIFdelend \noindent
[MCP - Mode choice model]
\begin{subequations}
\label{eq:mode_choice_opt}
\begin{align}
    0 \leq \left[C_{\omega}^{m} - \pi_{\omega}\right] \perp q_{\omega}^{m} \geq 0, \forall \omega, m \in \mathcal{M} \label{eq:mcp_mode_choice} \\
    0 \leq \left[\sum_{m \in \mathcal{M}} {q_{\omega}^{m}} - q_{\omega}\right] \perp \pi_{\omega} \geq 0, \forall \omega \label{eq:mcp_mode_conservation}
\end{align}
\end{subequations}
where, $\perp$ represents the inner product between two vectors, and $\pi_{\omega}$ is the multiplier for the demand conservation constraint on OD pair $\omega$.

\DIFdelbegin \DIFdel{We consider }\DIFdelend \DIFaddbegin \DIFadd{By interpreting the multiplier $\pi_{\omega}$ as the minimum modal cost, MCP \eqref{eq:mcp_mode_choice} and \eqref{eq:mcp_mode_conservation} state }\DIFaddend that, at equilibrium, traveler chooses mode $m$ for traveling between OD pair $\omega$ (i.e., $q_{\omega}^m>0$) only if its modal cost $C_{\omega}^m$ is the minimum \DIFaddbegin \DIFadd{(i.e., $C_{\omega}^{m} - \pi_{\omega} = 0$) }\DIFaddend among four alternative modes. 
\section{Ridesharing matching model}
\label{sec:Matching}
In a ridesharing system, \DIFdelbegin \DIFdel{platforms match a set of travelers who decided to rideshare as drivers (RD ) and passengers (RP )}\DIFdelend \DIFaddbegin \DIFadd{platform matches RD with RP to optimize its objective}\DIFaddend . We assume the \DIFdelbegin \DIFdel{platforms determine }\DIFdelend \DIFaddbegin \DIFadd{platform determines }\DIFaddend the matching in terms of matching sequences (\cref{sec:Matching sequence})\DIFdelbegin \DIFdel{. In a matching sequence $n$, the matched driver first departs from his/her origin as the initial task (}\textit{\DIFdel{Task 0}}
\DIFdel{) and finishes at his/her destination as the last task (}\textit{\DIFdel{Task}} 
\DIFdel{$L$); a matched passenger is first picked up at his/her origin (}\textit{\DIFdel{Task}} 
\DIFdel{$l$), and later dropped off at his/her destination (}\textit{\DIFdel{Task}} 
\DIFdel{$l', l'>l$). Under such setting, a ridesharing platform }\DIFdelend \DIFaddbegin \DIFadd{, which }\DIFaddend needs to first compute \DIFdelbegin \DIFdel{these }\DIFdelend \DIFaddbegin \DIFadd{candidate }\DIFaddend matching sequences with respect to some feasibility constraints, and assign RD and RP to \DIFdelbegin \DIFdel{these }\DIFdelend \DIFaddbegin \DIFadd{the }\DIFaddend matching sequences (\cite{agatz2012optimization}).

One of the challenges in ridesharing matching is to efficiently find \DIFdelbegin \DIFdel{these }\DIFdelend \DIFaddbegin \DIFadd{candidate }\DIFaddend matching sequences for \DIFdelbegin \DIFdel{a set of RD and RP}\DIFdelend \DIFaddbegin \DIFadd{RDs and RPs}\DIFaddend , which is known as the NP-hard dial-a-ride problem (DARP) \DIFdelbegin \DIFdel{in the literature }\DIFdelend (\cite{savelsbergh1995general}). We adapt the dynamic tree algorithm (\cite{yao2021dynamic}) subroutine for solving the DARP, which avoids enumeration of candidate matching sequences by casting the problem into a graph-theoretic framework~\citep[][]{santi2014quantifying,alonso2017demand}. This subroutine resembles a bilevel structure: at the \DIFdelbegin \DIFdel{upper-level}\DIFdelend \DIFaddbegin \DIFadd{upper level}\DIFaddend , it finds feasible RD-RP groups incrementally in size as clique problems; at the \DIFdelbegin \DIFdel{lower-level}\DIFdelend \DIFaddbegin \DIFadd{lower level}\DIFaddend , given a RD-RP group, it adapts a tree structure to generate feasible matching sequences\DIFdelbegin \DIFdel{based on RD vehicle capacity and platform level-of-service guarantee, as well as to determine matching sequence objective values from the platform's perspective}\DIFdelend . Only if there is at least one feasible matching sequence, the corresponding RD-RP group is considered feasible. \DIFdelbegin \DIFdel{A more detailed description of the subroutine is provided in Appendix~\ref{appendix.dynamic_tree}}\DIFdelend \DIFaddbegin \DIFadd{We refer to \mbox{
\textcite{yao2021dynamic} }\hskip0pt
for implementation details of the DARP subroutine}\DIFaddend .

Given the matching sequences, $\{n\}$, and the platform's objective \DIFdelbegin \DIFdel{values, $R_n$, }\DIFdelend \DIFaddbegin \DIFadd{value }\DIFaddend associated with each matching sequence, \DIFaddbegin \DIFadd{$R_n$, }\DIFaddend the ridesharing matching problem is reduced to a matching-sequence assignment problem with decision variable $Z_n$ denoting the number of matched $n$ as follows:
\begin{subequations}
\label{eq:matching_opt}
    \begin{alignat}{2}
    &\max_{Z} 			& \qquad 	& \sum_{n} {R_n \cdot Z_n} \label{eq:matching_obj} \\
    &\textrm{subject to}& 			& \nonumber \\
    & 					& 			& \sum_{n} {s_1{(n,0,\omega)} \cdot Z_n} \leq q_{\omega}^{RD}, \forall \omega \label{eq:RD_demand_constraint} \\
    & 					& 			& \sum_{n} {\sum_{1 \leq l \leq L-1} {s_1{(n,l,\omega)} \cdot Z_n} \leq q_{\omega}^{RP}}, \forall \omega \label{eq:RP_demand_constraint}
    \end{alignat}
\end{subequations}
where, \DIFdelbegin \DIFdel{matching-sequence assignment }\DIFdelend $Z_n$ is mapped to RD and RP trips using the pickup incident matrix $s_1$ \DIFdelbegin \DIFdel{. Remind that, }\DIFdelend \DIFaddbegin \DIFadd{(}\DIFaddend as defined in \cref{sec:Matching sequence}\DIFdelbegin \DIFdel{, the pickup incident $s_1(n, 0, \omega)=1$ only if the driver in matching sequence $n$ travels between OD pair $\omega$. 
Similarly, only if the passenger with OD pair $\omega^{'}$ is picked up as the $l^{th}$ task in matching sequence $n$, $s_1(n, l, \omega^{'})=1$.
}\DIFdelend \DIFaddbegin \DIFadd{). 
}\DIFaddend 

\DIFdelbegin \DIFdel{Following this notion, the }\DIFdelend \DIFaddbegin \DIFadd{The }\DIFaddend left-hand side of constraint \eqref{eq:RD_demand_constraint} represents the total number of \textit{matched} RD \DIFdelbegin \DIFdel{of OD pair }\DIFdelend $\omega$, who might be matched to multiple matching sequences (\DIFaddbegin \DIFadd{as indicated }\DIFaddend by $\sum_{n}$), and it is constrained by the total RD demands \DIFdelbegin \DIFdel{of OD pair $\omega$ (}\DIFdelend $q_{\omega}^{RD}$\DIFdelbegin \DIFdel{)}\DIFdelend . Similarly, the total number of \textit{matched} RP \DIFdelbegin \DIFdel{of OD pair }\DIFdelend $\omega$ is captured in the left-hand side of constraint \eqref{eq:RP_demand_constraint}, and constrained by the total RP demands \DIFdelbegin \DIFdel{of OD pair $\omega$ (}\DIFdelend $q_{\omega}^{RP}$\DIFdelbegin \DIFdel{)}\DIFdelend , in which multi-passenger ridesharing is explicitly considered (by including $\sum_{l}$ to represent multiple pickups).

Note that, the proposed ridesharing matching problem \eqref{eq:matching_opt} includes only ridesharing (RD/RP) demand constraints \eqref{eq:RD_demand_constraint}-\eqref{eq:RP_demand_constraint}, while matching sequence feasibility constraints are omitted. This is because the \DIFdelbegin \DIFdel{dynamic tree subroutine , which accounts for }\DIFdelend \DIFaddbegin \DIFadd{DARP subroutine handles the }\DIFaddend matching-sequence feasibility constraints \DIFdelbegin \DIFdel{, is performed before solving problem~\eqref{eq:matching_opt}}\DIFdelend \DIFaddbegin \DIFadd{implicitly}\DIFaddend . For example, \DIFdelbegin \DIFdel{preceding constraints in the subroutine ensure passenger will be dropped off only after being picked up }\DIFdelend \DIFaddbegin \DIFadd{the DARP subroutine ensures that RPs are picked up and dropped off without any transfer }\DIFaddend in a matching sequence\DIFdelbegin \DIFdel{, which also avoids passenger transfers}\DIFdelend ; and vehicle capacity \DIFdelbegin \DIFdel{constraints guarantee that the maximum number of passengers for a single driver will be }\DIFdelend \DIFaddbegin \DIFadd{constraint is }\DIFaddend satisfied anytime in \DIFdelbegin \DIFdel{the }\DIFdelend \DIFaddbegin \DIFadd{a }\DIFaddend matching sequence, which is crucial in a multi-passenger ridesharing setting. \DIFdelbegin \DIFdel{Note }\DIFdelend \DIFaddbegin \DIFadd{It is worth noting }\DIFaddend that, in contrast to the average vehicle-occupancy constraint approach (\cite{xu2015complementarity}, \cite{di2019unified}), the proposed matching-sequence approach considers \DIFdelbegin \DIFdel{explicitly the }\DIFdelend \DIFaddbegin \DIFadd{explicit }\DIFaddend vehicle capacity constraints \DIFaddbegin \DIFadd{in the DARP subroutine}\DIFaddend . The ability to impose vehicle capacity constraints in our model provides new possibilities to study the impacts of heterogeneous vehicle capacities. 

The ridesharing platform matching model \eqref{eq:matching_opt} is general in the sense that, it can accommodate different platform objective functions in \cref{eq:matching_obj}\DIFdelbegin \DIFdel{by computing accordingly the objective value $R_n$ for each matching sequence}\DIFdelend . For example, commercial ridesharing platforms \DIFdelbegin \DIFdel{might consider to maximize their revenues , such that }\DIFdelend \DIFaddbegin \DIFadd{maximizing revenues can substitute commission fees as their }\DIFaddend $R_n$\DIFdelbegin \DIFdel{represents the platform earnings (e.g., in terms of commission fees )}\DIFdelend . In this paper, we assume a public ridesharing platform \DIFdelbegin \DIFdel{which }\DIFdelend \DIFaddbegin \DIFadd{that }\DIFaddend aims at maximizing social benefits in terms of vehicle-kilometer traveled (VKT) savings, where $R_n$ is computed as the difference between ridesharing trip distances and total trip distances (in case all participants traveling alone).

Furthermore, as mentioned in \cref{sec:contributions}, one key feature of the proposed generalized Nash equilibrium model is that, decision variables of other players are considered exogenous to current player\DIFdelbegin \DIFdel{, where }\DIFdelend \DIFaddbegin \DIFadd{. Consequently, }\DIFaddend the objectives and \textbf{constraints} of this player might be dependent on the decision variables of other players. \DIFdelbegin \DIFdel{As this }\DIFdelend \DIFaddbegin \DIFadd{This }\DIFaddend is the case in the proposed ridesharing matching problem \eqref{eq:matching_opt}, \DIFaddbegin \DIFadd{where }\DIFaddend the traveler's mode choice \DIFdelbegin \DIFdel{decision variables }\DIFdelend \DIFaddbegin \DIFadd{decisions }\DIFaddend $q_{\omega}^{RD/RP}$ \DIFdelbegin \DIFdel{appear in }\DIFdelend \DIFaddbegin \DIFadd{interact with platform matching decisions $Z_n$ in the form of }\DIFaddend constraints \eqref{eq:RD_demand_constraint}-\eqref{eq:RP_demand_constraint}\DIFdelbegin \DIFdel{of the ridesharing platform's matching optimization problem (i.
e., coupling constraints), which represent the interactions between travelers and ridesharing platforms.
}\DIFdelend \DIFaddbegin \DIFadd{.
}\DIFaddend 

Let $\mu_{1}^{\omega}$ denote the dual variable of the RD demand constraint \eqref{eq:RD_demand_constraint}, and $\mu_{2}^{\omega}$ \DIFdelbegin \DIFdel{denotes }\DIFdelend \DIFaddbegin \DIFadd{denote }\DIFaddend the dual variable of the ridesharing passenger (RP) demand constraint \eqref{eq:RP_demand_constraint}. The optimality conditions of ridesharing matching problem \eqref{eq:matching_opt} are given by the following complementarity conditions:

\noindent
[MCP - Ridesharing matching model]
\begin{flalign} 
\label{eq:mcp_matching_obj}
    0 \leq \left[
    -R_n 
    + \underbrace{\sum_{\omega} {\mu_{1}^{\omega} \cdot s_1{(n,0,\omega)}}}_{\text{Multiplier of the matched driver}} 
    + \underbrace{\sum_{\omega} {\mu_{2}^{\omega} \cdot \sum_{1 \leq l \leq L-1} {s_1{(n,l,\omega)}}}}_{\text{Multipliers of the matched passengers}} \right] \perp Z_n \geq 0, \forall n
\end{flalign}
\begin{flalign} 
\label{eq:mcp_RD_constraint}
    0 \leq \left[q_{\omega}^{RD} - \sum_{n} {s_1{(n,0,\omega)} \cdot Z_n} \right] \perp \mu_{1}^{\omega} \geq 0, \forall \omega
\end{flalign}
\begin{flalign} 
\label{eq:mcp_RP_constraint}
    0 \leq \left[q_{\omega}^{RP} - \sum_{n} {\sum_{1 \leq l \leq L-1} {s_1{(n,l,\omega)} \cdot Z_n}} \right] \perp \mu_{2}^{\omega} \geq 0, \forall \omega
\end{flalign}

The dual variables $\mu_{1}^{\omega}$ and $\mu_{2}^{\omega}$ of the complementarity constraints \eqref{eq:mcp_RD_constraint} and \eqref{eq:mcp_RP_constraint} can be interpreted as the incentives that the ridesharing platform pays to the ridesharing drivers and passengers. When there \DIFdelbegin \DIFdel{are more }\DIFdelend \DIFaddbegin \DIFadd{is a surplus in }\DIFaddend RD/RP \DIFdelbegin \DIFdel{than needed, ridesharing platforms need not to provide incentive . }\DIFdelend \DIFaddbegin \DIFadd{demands (i.e. demand constraints~\eqref{eq:RD_demand_constraint}, \eqref{eq:RP_demand_constraint} not binding), no incentive is needed (i.e., $\mu_{1}^{\omega}=0, \mu_{2}^{\omega}=0$). }\DIFaddend Conversely, if there are few RD/RP \DIFdelbegin \DIFdel{, platforms need }\DIFdelend \DIFaddbegin \DIFadd{(that is demand constraints are already binding), platform needs }\DIFaddend to promote ridesharing services by providing incentives \DIFaddbegin \DIFadd{($\mu_{1}^{\omega}>0, \mu_{2}^{\omega}>0$).  
}\DIFaddend 

\section{Network model for joint route choice and stable matching}
\label{sec:RidesharingNetworkModel}
As described in \cref{sec:Hypernet}, the hyper-network is composed of four mode-specific subnetworks, where RD and RP subnetworks are expanded for matching sequences. In addition, virtual links are constructed to accommodate \DIFdelbegin \DIFdel{both }\DIFdelend multi-passenger ridesharing without transfer \DIFdelbegin \DIFdel{, }\DIFdelend and matching stability. In this section, a link-based network model is formulated over the hyper-network to jointly handle the route choice and stable matching problems, where travelers' routes are determined endogenously. 

\begin{figure}[H]
    \centering
    \includegraphics[width=0.9\textwidth]{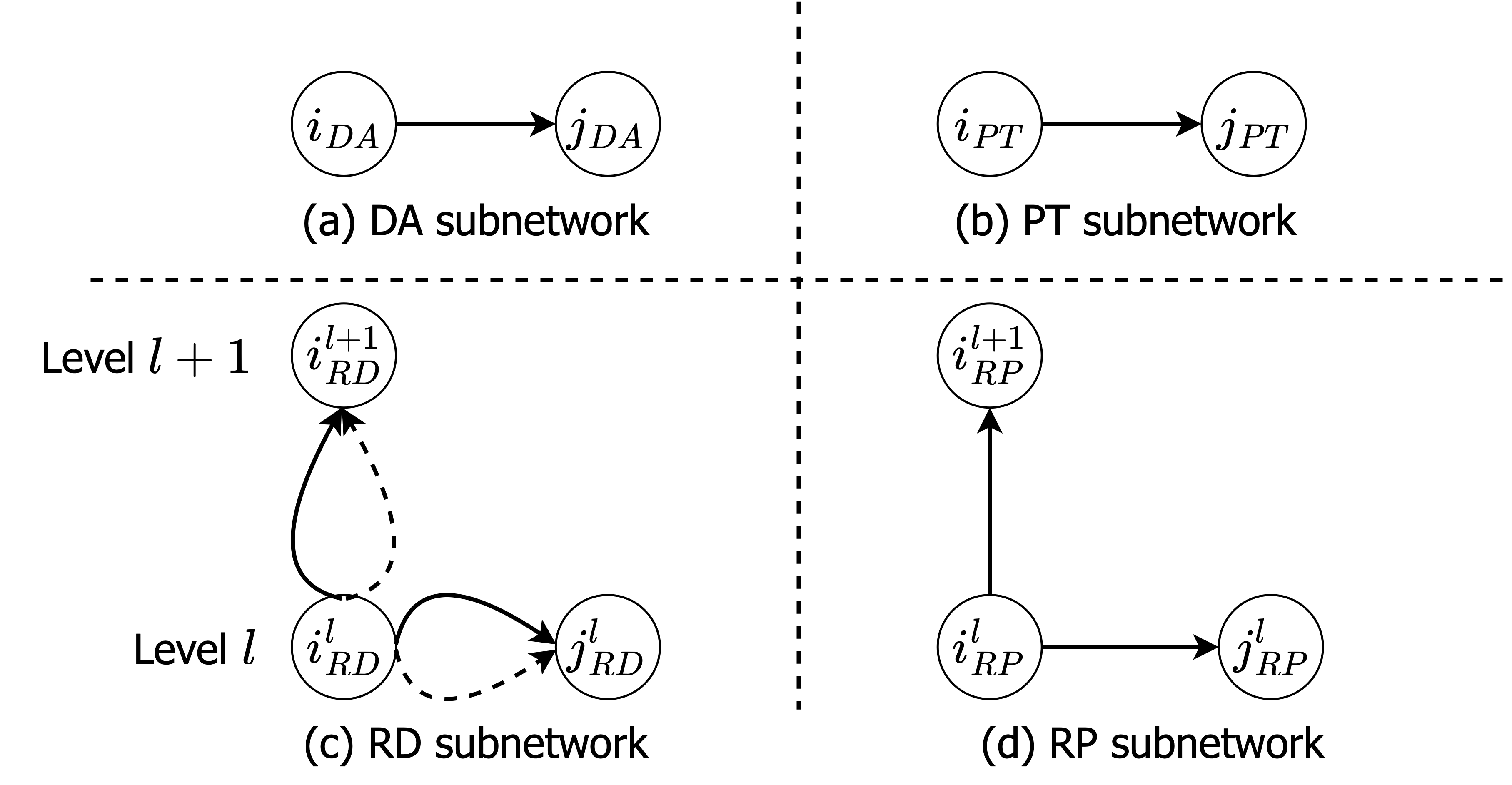}
    \caption{Basic graph components of the hyper-network}
    \label{fig:7.basci_hypernet}
\end{figure}

Before formulating the network model, node-potential and link-flow variables are introduced for the hyper-network. We summarize the index for defining these variables in \cref{tab:hypernet_link_flow}, with respect to the basic graph components in the hyper-network (\cref{fig:7.basci_hypernet}). 

\DIFdelbegin \DIFdel{In network equilibrium models, node-potential variables correspond to the dual variables of the flow conservation constraints at each node. Let $\mathcal{D}$ denote the set of traveler destinations, we consider the case that flow conservations are specified for each mode ($m \in \mathcal{M}$) with respect to traveler destinations $e \in \mathcal{D}$ (i.e., where traveler flows leave the network). To account for the setting that RD and RP follow the matching sequences to travel in the hyper-network, their node-potential variables are augmented with additional index $n$ for matching sequences. Note that, RD's destination index is omitted since it is already embedded in the matching sequence $n$.
}

\DIFdel{Similarly, link flow variables are specified for each mode $m$ with respect to traveler destinations $e \in \mathcal{D}$ and matching sequences $n$ (in the case of RD and RP). Whereas link flows on virtual links $(\widehat{i_{RD}}, i_{DA})$ and $(\widehat{i_{RP}}, i_{PT})$ are specified only for destinations ($e$) such that they are consistent with the DA and PT subnetworks.  
}


\DIFdelend \begin{table}[h]
    \caption{Summary of index for node-potential and link-flow variables.}
    \label{tab:hypernet_link_flow}
    \centering
    \small
    \setlength\tabcolsep{1.5pt}
    \begin{threeparttable}
        \begin{tabular}{c| c| c P{0.11\linewidth} P{0.10\linewidth}| ccc| ccc P{0.11\linewidth} P{0.10\linewidth}
                       }
        \toprule
        \multicolumn{1}{l|}{\textbf{\makecell{Network}}}
           & \multicolumn{1}{l|}{\textbf{\makecell{Node}}}
           & \multicolumn{3}{l|}{\textbf{\makecell[l]{Node potential$^{*}$}}}
              & \multicolumn{3}{l|}{\textbf{\makecell[l]{Link}}}
              & \multicolumn{5}{l}{\textbf{\makecell[l]{Link flow}}} \\

        	&  
        	& 
        Node &
        Destination &
        Matching sequence &
        Tail &
        Head &
        Status &
        Link &
         	&
        	&
        Destination &
        Matching sequence \\

        \midrule
        Original network 									& $i, j$	& -	&	&	& $(i,$	& $j)$	& 	& -	& 	& 	& 	&  		\\
        \midrule
        DA subnetwork 											& $i, j$& $i$	&	$e$&				& $(i,$	& $j)$	& 	& $(i,$	& $j)$	& 	& 	$e$&  		\\
        \midrule
        PT subnetwork 											& $i, j$	& $i$	&	$e$&			& $(i,$	& $j)$	& 	& $(i,$	& $j)$	& 	& 	$e$&  		\\
        \midrule
        {\multirow{2}{*}{\makecell{RD subnetwork$^{**}$}}} 			& 
        	$i^{l_1}:(i, l_1),$							& 
        	{\multirow{2}{*}{\makecell{$i^{l_1}$}}}	&	
        		&	
        		{\multirow{2}{*}{\makecell{$n$}}}&
        		$(i^{l_1},$						& 
        			$j^{l_2},$					& 	
        				$0)$					& 
        				$(i^{l_1},$				& 
        					$j^{l_2},$			& 	
        							$0)$		& 	
        										&  		
        									{\multirow{2}{*}{\makecell{$n$}}}\\

        													& 
        	$j^{l_2}:(j, l_2)$							& 
        		&	
        		&	
        		&
        		$(i^{l_1},$						& 
        			$j^{l_2},$					& 	
        				$1)$					& 
        				$(i^{l_1},$				& 
        					$j^{l_2},$			& 	
        							$1)$		& 	
        										&  		
        									\\
		\midrule
        {\multirow{2}{*}{\makecell{RP subnetwork$^{**}$}}} 			& 
        	$i^{l_1}:(i, l_1),$							&
        	{\multirow{2}{*}{\makecell{$i^{l_1}$}}}	&	
        		{\multirow{2}{*}{\makecell{$e$}}}&	
        		{\multirow{2}{*}{\makecell{$n$}}}& 
        		{\multirow{2}{*}{\makecell{$(i^{l_1},$}}}						& 
        			{\multirow{2}{*}{\makecell{$j^{l_2})$}}}					& 	
        								& 
        				{\multirow{2}{*}{\makecell{$(i^{l_1},$}}}			& 
        					{\multirow{2}{*}{\makecell{$j^{l_2})$}}}			& 	
        									& 	
        								{\multirow{2}{*}{\makecell{$e$}}}		&  		
        									{\multirow{2}{*}{\makecell{$n$}}}\\

        													& 
        	$j^{l_2}:(j, l_2)$							& 
        	&
        	&
        	&
        								& 
        								& 	
        									& 
        								& 
        								& 	
        									& 	
        										&  		
        									\\

        \bottomrule
        \end{tabular}
    	\begin{tablenotes}
        \item[*]The same set of indexes for node $i$ applies to node $j$.
        \item[**] As illustrated in \cref{fig:7.basci_hypernet}, $j^{l_2}$can be either $j^l$(i.e., $l_2 = l_1$) or $i^{l+1}$ (i.e., $l_2 = l_1 +1$).
    \end{tablenotes}
    \end{threeparttable}
\end{table}

In network equilibrium models, node-potential variables correspond to the dual variables of the flow conservation constraints at each node. Specifically, let $\mathcal{D}$ denote the set of traveler destinations, node-potential variables are specified for each mode ($m \in \mathcal{M}$) and each destination $e \in \mathcal{D}$. 
To account for the setting that RD and RP follow the matching sequences to travel in the hyper-network, their node-potential variables are augmented with additional index $n$ for matching sequences. Note that, RD's destination index is omitted since it is already embedded in the matching sequence $n$.

Similarly, link flow variables are specified for each mode $m$ with respect to traveler destinations $e \in \mathcal{D}$ and matching sequences $n$ (in the case of RD and RP). Whereas link flow variables on virtual links $(\widehat{i_{RD}}, i_{DA})$ and $(\widehat{i_{RP}}, i_{PT})$ are specified only for destinations ($e$) such that they are consistent with the DA and PT subnetworks.  

One important feature of the proposed link-based formulation is that, path enumeration can be avoided for multi-passenger ridesharing. With the proposed hyper-network (\cref{fig:5.hypernetwork}), ridesharing routes are determined endogenously by \DIFdelbegin \DIFdel{drivers }\DIFdelend \DIFaddbegin \DIFadd{RD }\DIFaddend traversing links in the hyper-network with minimum node potentials (explained in the following subsections). Note also that, RP node potentials and link flows of a matching sequence $n$ are specified with respect to destinations, but not OD pairs (as in \cite{li2020path}, \cite{chen2022unified}), which could improve efficiency in large-scale problems.

In our network model, \DIFdelbegin \DIFdel{travelers driving alone (DA ) and taking public transport (PT ) }\DIFdelend \DIFaddbegin \DIFadd{DA and PT travelers }\DIFaddend enter directly into the \DIFdelbegin \DIFdel{DA and PT }\DIFdelend \DIFaddbegin \DIFadd{corresponding }\DIFaddend subnetworks, while \DIFdelbegin \DIFdel{travelers ridesharing as driver (RD ) or passenger (PT) first }\DIFdelend \DIFaddbegin \DIFadd{RD (RP) }\DIFaddend enter the virtual origin nodes $\widehat{i_{RD}}$ (\DIFdelbegin \DIFdel{or }\DIFdelend $\widehat{i_{RP}}$) to choose one of the matching sequences or leave ridesharing. After entering the subnetworks, travelers select the routes that minimize their generalized travel costs, in which RD and RP travel together. In the following subsections, stable matching is first formulated as a route choice problem in the hyper-network (\cref{sec:stable_matching}), and the travelers' route choice decisions, as well as ridesharing constraints, are formulated in \cref{sec:route_choice}.

\subsection{Stable matching as a route choice problem}
\label{sec:stable_matching}
Given the matching sequences, $\{n\}$, and the number of matched RD and RP on each matching sequence, $Z_n$, matching stability is reached when no \DIFdelbegin \DIFdel{ridesharing driver (RD ) and passenger (RP ) }\DIFdelend \DIFaddbegin \DIFadd{RD and RP }\DIFaddend prefer to be matched together other than their current matching (\cite{wang2018stable}). 
\DIFdelbegin \DIFdel{Specifically, in }\DIFdelend \DIFaddbegin \DIFadd{In }\DIFaddend the context of multi-passenger ridesharing, stable matching typically \DIFdelbegin \DIFdel{describes the mutual agreements on a driver-passengers matching }\textit{\DIFdel{group}} 
\DIFdel{between the matched }\DIFdelend \DIFaddbegin \DIFadd{refers to the mutual agreement between the }\DIFaddend driver and several passengers \DIFaddbegin \DIFadd{on a RD-RP matching }\textit{\DIFadd{group}} \DIFaddend (\cite{peng2022many}). 
However, \DIFdelbegin \DIFdel{ridesharing drivers and passengers might experience different ridesharing disutilities for the same driver-passengers matching group, due to possible different matching sequences}\DIFdelend \DIFaddbegin \DIFadd{different RD-RP }\textit{\DIFadd{matching sequences}} \DIFadd{can lead to different ridesharing experiences for RD and RP in the same RD-RP matching }\textit{\DIFadd{group}}\DIFadd{, particularly in a multi-OD multi-passenger scenario where RP may detour to pick up peer RPs}\DIFaddend . 
To account for such differences, we propose to extend the multi-passenger stable matching in terms of \textit{matching sequence}, \DIFdelbegin \DIFdel{instead of only driver-passengers matching group. 
Such extension also allows modeling the }\DIFdelend \DIFaddbegin \DIFadd{rather than just RD-RP matching }\textit{\DIFadd{group}}\DIFadd{. 
This expansion enables the modeling of }\DIFaddend multi-passenger stable matching problem in the spatial context, where ridesharing disutilities are jointly determined by the travelers' choice of \textbf{matching sequence} and \textbf{route} in the network.

Formally, we consider a stable matching setting that, each ridesharing traveler $\omega$ (RD/RP) has a preference list for a set matching sequences he/she is matched with $\{n | n:s_{1}(n,\cdot,\omega)=1\}$, ranked based on the ridesharing disutilities of each matching sequence, denoted as $\pi^{n, w}$. It is assumed that, a ridesharing traveler $\omega$ strictly prefers matching sequence $n'$ to $n$
, if its ridesharing disutility is lower, $\pi^{n',w} < \pi^{n,w}$. 
Let $\{w\}_n$ denote the set of participants (the RD and several RPs) in matching sequence $n$, and $\mathcal{S}_w = \{n' | n':s_{1}(n',\cdot,w)=1\}$ \DIFdelbegin \DIFdel{denotes }\DIFdelend \DIFaddbegin \DIFadd{denote }\DIFaddend the set of matching sequences for $w$, the stable multi-passenger ridesharing matching is formally defined for a set of matchings $\{n\}$ as follows:
\begin{definition}[Stable multi-passenger ridesharing matching]

If for every matched $n$, we have:
\begin{align}
\pi^{n,w} \leq \pi^{n',w}, \forall w \in \{w\}_n, n' \in \mathcal{S}_w \label{eq:def_stability_condition}
\end{align}
then, the following stability condition \citep[][]{sotomayor1992multiple} holds for every matched $n$:
\begin{align}
\sum_{w \in \{w\}_n } \pi^{n, w} \leq \sum_{w \in \{w\}_{n} } \inf_{\substack{n' \in \mathcal{S}_w\\n' \neq n} } {\pi^{n', w}},  \nonumber
\end{align}
and $\{n\}$ corresponds to a stable multi-passenger ridesharing matching.
\end{definition}
The above stability condition states that in a stable multi-passenger ridesharing matching there is no matching sequence $n$ such that all participants (the driver and several passengers) in matching sequence $n$ strictly prefer other matching sequences $n'$.

By introducing virtual ridesharing origin nodes $\widehat{i_{RD}}$ and $\widehat{i_{RP}}$, and virtual links (with corresponding flow variables), the multi-passenger ridesharing stable matching problem is cast as a route choice problem in the hyper-network with ridesharing disutilities determined endogenously. This subsection \DIFdelbegin \DIFdel{focus }\DIFdelend \DIFaddbegin \DIFadd{focuses }\DIFaddend on the virtual ridesharing nodes and links in the hyper-network. 

As explained in the previous section, we denote the node potentials in the RD and RP subnetworks as $\pi_{i^1}^{n, RD}$ and $\pi_{i^l}^{n, e, RP}$, whereas $\pi_{\widehat{i_{RD}}}^{e, RD}$
and $\pi_{\widehat{i_{RP}}}^{e, RP}$ are denoted for virtual origin $\widehat{i_{RD}}$ and $\widehat{i_{RP}}$. Link flow variables on links $(\widehat{i_{RD}}, i^1)$ and $(\widehat{i_{RP}}, i^l)$ are augmented to represent ridesharing travelers' choices on matching sequences, denoted by $x_{\widehat{i_{RD}}, i^1}^{n, RD}$ and $x_{\widehat{i_{RP}}, i^l}^{n, e, RP}$, while link flows representing travelers leaving ridesharing are augmented for their destinations $x_{\widehat{i_{RD}}, i_{DA}}^{e, RD}$ and $x_{\widehat{i_{RP}}, i_{PT}}^{e, RP}$. We denote the \textit{generalized} link costs (summation of the link cost and multipliers, as in \cite{larsson1999side}) for the above link flow variables as $\tilde{c}_{\widehat{i_{RD}}, i^1}^{n, RD}$, $\tilde{c}_{\widehat{i_{RP}}, i^l}^{n, e, RP}$, $\tilde{c}_{\widehat{i_{RD}}, i_{DA}}^{e, RD}$ and $\tilde{c}_{\widehat{i_{RP}}, i_{PT}}^{e, {RP}}$. The multi-passenger ridesharing stable matching problem is formulated as the following MCP:
\\\textbf{{Ridesharing driver matching sequence preference as route choice:}}
\begin{subequations}
\label{eq:RD_stable_route}
    \begin{alignat}{2}
    & 0 \leq \left[  
    				\underbrace{\pi_{i^1}^{n, RD} 
    				+ \tilde{c}_{\widehat{i_{RD}}, i^1}^{n, RD}}_{\text{Ridesharing disutility of }n}
    				- \pi_{\widehat{i_{RD}}}^{e, RD}  
    	\right] 
    \perp 
    				x_{\widehat{i_{RD}}, i^1}^{n, RD}
    \geq 0, 
    				 \forall i, n \label{eq:RD_seq_choice} & \\
   	& 0 \leq \left[  
    				\underbrace{\pi_{i}^{e, DA} 
    				+ \tilde{c}_{\widehat{i_{RD}}, i_{DA}}^{e, RD}}_{\text{Disutility of leaving ridesharing}}
    				- \pi_{\widehat{i_{RD}}}^{e, RD}  
    	\right] 
    \perp 
    				x_{\widehat{i_{RD}}, i_{DA}}^{e, RD} 
    \geq 0, 
    				\forall i, e \in \mathcal{D} &  \label{eq:RD_to_DA} \\
    & 0 \leq \left[  
    				\underbrace{
    				\sum_{n} {x_{\widehat{i_{RD}}, i^1}^{n, RD}} 
    				+ x_{\widehat{i_{RD}}, i_{DA}}^{e, RD}
    				- q_{(i,e)}^{RD}}_{\text{RD demand conservation}}
    	\right] 
    \perp 
    				\pi_{\widehat{i_{RD}}}^{e, RD}
    \geq 0, 
    				\forall i, e \in \mathcal{D} &\label{eq:RD_stable_matching_conservation}
    \end{alignat}
\end{subequations}
\textbf{{Ridesharing passenger matching sequence preference as route choice:}}
\begin{subequations}
\label{eq:RP_stable_route}
    \begin{alignat}{2}
    & 0 \leq \left[  
    				\underbrace{\pi_{i^l}^{n, e, RP} 
    				+ \tilde{c}_{\widehat{i_{RP}}, i^l}^{n, e, RP} }_{\text{Ridesharing disutility of }n}
    				- \pi_{\widehat{i_{RP}}}^{e, RP}
    	\right] 
    \perp 
    				x_{\widehat{i_{RP}}, i^l}^{n, e, RP}
    \geq 0, 
    				 \forall i, e, n, l \label{eq:RP_seq_choice} & \\
   	& 0 \leq \left[  
    				\underbrace{\pi_{i}^{e, PT} 
    				+ \tilde{c}_{\widehat{i_{RP}}, i_{PT}}^{e, {RP}}}_{\text{Disutility of leaving ridesharing}}
    				- \pi_{\widehat{i_{RP}}}^{e, RP}  
    	\right] 
    \perp 
    				x_{\widehat{i_{RP}}, i_{PT}}^{e, RP} 
    \geq 0, 
    				\forall i, e \in \mathcal{D} &  \label{eq:RP_to_PT} \\
    & 0 \leq \left[  
    				\underbrace{
    				\sum_{n} {\sum_{1 \leq l \leq L-1} {x_{\widehat{i_{RP}}, i^l}^{n, e, RP}}} 
    				+ x_{\widehat{i_{RP}}, i_{PT}}^{e, RP}
    				- q_{(i,e)}^{RP}}_{\text{RP demand conservation}}
    	\right] 
    \perp 
    				\pi_{\widehat{i_{RP}}}^{e, RP}
    \geq 0, 
    				\forall i, e \in \mathcal{D} &\label{eq:RP_stable_matching_conservation}
    \end{alignat}
\end{subequations}
Subject to:
\\\textbf{{Stable matching constraint:}}
\begin{flalign}
\label{eq:stable_matching_constraint}
	x_{\widehat{i_{RD}}, i^1}^{n, RD} = x_{\widehat{j_{RP}}, j^l}^{n, e, RP}, & \forall n, \\
    	& \underbrace{i:s_1(n, 0, (i, \cdot))=1}_{\text{Origin of the matched driver in }n}, \nonumber \\
    	& (j, l, e) \in \underbrace{\{(j, l, e)|s_1(n, l, (j, e)) = 1, 1 \leq l \leq L - 1\}}_{\text{Set of matched passengers in }n} \nonumber
\end{flalign}
\textbf{{Platform matching constraint:}}
\begin{gather}
	x_{\widehat{i_{RD}}, i^1}^{n, RD} \leq s_1(n, 0, (i, \cdot)) \cdot Z_n, \forall n, i \in O \label{eq:RD_matching_cap} \\
	x_{\widehat{i_{RP}}, i^l}^{n, e, RP} \leq s_1(n, l, (i, e)) \cdot Z_n, \forall n, (i, e): \omega, 1 \leq l \leq L-1 \label{eq:RP_matching_cap}
\end{gather}

\cref{eq:RD_stable_route}-\eqref{eq:stable_matching_constraint} define the multi-passenger ridesharing stable matching, \DIFdelbegin \DIFdel{where }\DIFdelend \DIFaddbegin \DIFadd{in which }\DIFaddend ridesharing traveler's (RD/RP) preferences to matching sequences with lowest ridesharing disutility (stability condition Eq.~\ref{eq:def_stability_condition}) is captured in the route choice problem \eqref{eq:RD_stable_route} and \eqref{eq:RP_stable_route}. \DIFdelbegin \DIFdel{Whereas the }\DIFdelend \DIFaddbegin \DIFadd{The }\DIFaddend maximum number of RD and RP on each matching sequence are constrained by the platform's decision variables $Z_n$ in \eqref{eq:RD_matching_cap}-\eqref{eq:RP_matching_cap}. 

The stable matching constraint \eqref{eq:stable_matching_constraint} can be interpreted as the mutual agreement on a matching sequence $n$ between all the participants in $n$. Intuitively, if all the ridesharing travelers reach mutual agreements \eqref{eq:stable_matching_constraint} that their matching sequences are the best option \textit{they can get} \eqref{eq:RD_stable_route}-\eqref{eq:RP_stable_route}, we have a stable matching. We show the equivalence of the proposed formulation to the stable matching problem in Appendix \cref{prop:equivalence_stable_matching}. 

Given that vehicle capacity constraints are guaranteed in the matching sequences and RD follows these matching sequences, stable matching constraint \eqref{eq:stable_matching_constraint} also ensures that vehicle capacity constraints are satisfied in the network model. In essence, \cref{eq:stable_matching_constraint} states that, for a matching sequence, the number of passengers being picked up (and dropped off) at each pickup (drop-off) node $i^l$, $x_{\widehat{i_{RP}}, i^l}^{n, e, RP}$, equals to the number of drivers, $x_{\widehat{i_{RD}}, i^1}^{n, RD}$. This is equivalently saying that, there are $x_{\widehat{i_{RD}}, i^1}^{n, RD}$ number of drivers traveling on matching sequence $n$, in which each one of them picks up or drops off one passenger (and in total $1\cdot x_{\widehat{i_{RD}}, i^1}^{n, RD}$ RPs being served) at each task node\DIFaddbegin \DIFadd{, }\DIFaddend exactly as specified in the matching sequence\DIFdelbegin \DIFdel{such that his/her }\DIFdelend \DIFaddbegin \DIFadd{. As a result, the }\DIFaddend vehicle capacity constraint is always satisfied.

The stable matching constraint \eqref{eq:stable_matching_constraint} is a key constraint in the formulation. It represents one of the shared constraints that couples the decision variables of \DIFdelbegin \DIFdel{ridesharing drivers and passengers}\DIFdelend \DIFaddbegin \DIFadd{RD and RP}\DIFaddend . From the driver perspective, $x_{\widehat{i_{RD}}, i^1}^{n, RD}$ are the decision variables and $x_{\widehat{i_{RP}}, i^l}^{n, e, RP}$ are the anticipated variables, and reversed for the passengers. Similar to \textcite{ban2019general}, the dual variables of constraint \eqref{eq:stable_matching_constraint} \DIFdelbegin \DIFdel{(denoted as $\phi_{5}^{n, l}$ for RD, and $\widehat{\phi}_{5}^{n, l}$ for RP) need not to }\DIFdelend \DIFaddbegin \DIFadd{need not }\DIFaddend be equivalent, as these dual variables represent the marginal price of the constraint for RD and RP, respectively. \DIFdelbegin \DIFdel{To prove the solution existence of the overall equilibrium model, a }\DIFdelend \DIFaddbegin \DIFadd{A }\DIFaddend normalized relationship between these multipliers will be introduced \DIFdelbegin \DIFdel{in \mbox{
\cref{sec:Existence}}\hskip0pt
.
}\DIFdelend \DIFaddbegin \DIFadd{to prove the solution existence of the overall equilibrium model in Appendix~\mbox{
\cref{sec:Existence}}\hskip0pt
. 
}\DIFaddend 
Note that, as indicated in the introduction section, the explicit consideration of stable matching was missing in the literature, apart from \textcite{li2020path}. The addition of \DIFdelbegin \DIFdel{the }\DIFdelend stable matching constraints in the model formulation resembles a more realistic setting, in which drivers and passengers may reject the matching sequences. \DIFaddbegin \DIFadd{We also summarize the MCP of the stable matching problem with its corresponding multipliers in Appendix~\ref{appendix.MCP_network_model}.
}\DIFaddend 
\subsection{Traveler route choice model}
\label{sec:route_choice}
After ridesharing travelers (RD/RP) choose their preferred matching sequence, or switch their modes into DA or PT at the virtual nodes, all travelers depart from their origins in the four subnetworks. In these subnetworks, DA and PT travelers are assumed to make route choice decisions to minimize their generalized travel costs, while ridesharing routes are jointly determined by RD and RP. One of the key features of the proposed link-based model is that, RD can decide their routes between two consecutive tasks, as long as they follow their chosen matching sequences. Consequently, each RP is served by one RD, and needs not to make transfer. 

\subsubsection{Link costs}
\label{sec:link_cost}
To ensure consistency between the mode choice model and network model, we assume link cost components are the same as in modal costs functions \eqref{eq:DA_modal_costs}-\eqref{eq:PT_modal_costs}, such that the travel times and distances on the actual traveled route can be retrieved for the mode choice model. The travel costs for each type of flow on link $(i, j)$ are defined as follows:
\begin{table}[H]
\caption{Link cost components}
\label{tab:link_cost_final}
\centering
\setlength\tabcolsep{20pt}
\begin{tabular}{cc|cc}
  \toprule
   \multicolumn{2}{l}{\makecell{Link cost}} & $t_{i,j}$ & $d_{i,j}$  \\
   \midrule
DA &$c_{i, j}^{DA}$ & $\alpha^{DA}$          & $\beta$                    \\[0.5em]
\midrule
RD &$c_{i^{l_1},j^{l_2}, 0}^{RD}$ & $\alpha^{RD}$          & $\beta$                \\[0.5em]
 &$c_{i^{l_1},j^{l_2}, 1}^{RD}$ & $\alpha^{RD} + \tau_{t}^{RD} - \nu_{t}^{RD}$          & $\beta  + \tau_{d}^{RD} - \nu_{d}^{RD}$                \\[0.5em]
 \midrule
RP &$c_{i^{l_1}, j^{l_2}}^{RP}$ & $\alpha^{RP} + \tau_{t}^{RP} + \nu_{t}^{RP}$          & $\tau_{d}^{RP} + \nu_{d}^{RP}$             \\[0.5em]
\midrule
PT &$c_{i,j}^{PT}$ & $\alpha^{PT} + \tau_{t}^{PT} + \nu_{t}^{PT}$          & $\tau_{d}^{PT} + \nu_{d}^{PT}$               \\
\bottomrule   
\end{tabular}
\end{table}
\noindent
where, $t_{i, j}$ is the flow-dependent travel time on link $(i, j)$ for DA, RD and RP, and segregated transit network link travel time for PT; $d_{i, j}$ are the link length in the vehicle network and transit network, respectively. It is also assumed that, link costs on virtual links $(i^{l}, i^{l+1})$ are zero (i.e., $c_{i^{l},i^{l+1}, 0/1}^{RD} = 0, c_{i^{l}, i^{l+1}}^{RP} = 0$), representing no additional link time/distance to pick up/drop off a passenger at the same location.

Note that, compared to the RD link cost functions in \cref{tab:link_cost_final}, shared portion variables $\gamma$ are included in RD modal cost components Eq.~\eqref{eq:RD_inconvenience} \DIFdelbegin \DIFdel{and Eq.~\eqref{eq:RD_price} }\DIFdelend to represent the average ridesharing disutility. Using the proposed network model, the ridesharing disutility of a matching sequence can be computed endogeneously, by accumulating separately for the shared and non-shared portion using the with-passenger and without-passenger link costs.

\subsubsection{Route choice}
\label{sec:network_route_choice}
Following previous studies, travelers in a ridesharing system are assumed to depart their origins and select minimum-generalized-cost routes to reach their destinations. Let $\mathbf{A}$ \DIFdelbegin \DIFdel{denotes }\DIFdelend \DIFaddbegin \DIFadd{denote }\DIFaddend the node-link incident matrix of the hyper-network, in which, for node $j$, $a_{j}^{ij} = 1$ for all incoming links $(i,j)$, and $a_{j}^{jk} = -1$ for all outgoing links $(j,k)$. Let $\mathbf{x}$ and $\bm{\pi}$ denote the vectors of link flows \DIFdelbegin \DIFdel{node potentials and  }\DIFdelend \DIFaddbegin \DIFadd{and node potentials }\DIFaddend with respect to mode $m$, destinations $e$ and matching sequences $n$ (as defined in \cref{tab:hypernet_link_flow}), $\mathbf{q}$ \DIFdelbegin \DIFdel{denotes }\DIFdelend \DIFaddbegin \DIFadd{denote }\DIFaddend the vector of demands for each mode at each origin\DIFdelbegin \DIFdel{, }\DIFdelend \DIFaddbegin \DIFadd{. The }\DIFaddend flow conservation constraints can be represented as follows:
\begin{subequations}
\begin{align} 
\label{eq:flow_conservation}
(\mathbf{Ax} - \mathbf{q}) \odot \bm{\pi} = \mathbf{0} \\
\mathbf{Ax} - \mathbf{q} \geq \mathbf{0}\DIFdelbegin 
\DIFdelend \DIFaddbegin \DIFadd{, }\quad  \DIFaddend \bm{\pi} \geq \mathbf{0}
\end{align}
\end{subequations}
where $\odot$ represents the component-wise multiplication. Correspondingly, let $\tilde{\mathbf{c}}$ denote the vector of \textit{generalized} link costs, the route choice problem is represented as follows:
\begin{subequations}
\begin{align} 
\label{eq:network_route_choice_matrix}
(\mathbf{A}^\intercal \bm{\pi} + \tilde{\mathbf{c}}) \odot \mathbf{x} = \mathbf{0} \\
\mathbf{A}^\intercal \bm{\pi} + \tilde{\mathbf{c}} \geq \mathbf{0}\DIFdelbegin 
\DIFdelend \DIFaddbegin \DIFadd{, }\quad \DIFaddend \mathbf{x} \geq \mathbf{0}
\end{align}
\end{subequations}
We refer to Appendix~\DIFdelbegin \DIFdel{\ref{subsec:MCP_route_choice_descriptive} }\DIFdelend \DIFaddbegin \DIFadd{\ref{appendix.MCP_network_model} }\DIFaddend for a more descriptive MCP formulation for the flow conservation and route choice problem.

Note that, the above route choice problem is similar to a classic traffic assignment problem (TAP). However, in a multi-passenger ridesharing problem, the route choice behaviors of RD and RP are expected to influence each other, i.e., there exists coupling between \DIFdelbegin \DIFdel{ridesharing driver and the passenger }\DIFdelend \DIFaddbegin \DIFadd{RD and RP }\DIFaddend link flow variables. In the following subsection, we introduce a set of constraints that captures the interactions between RD and RP in a multi-passenger ridesharing problem.

\DIFdelbegin 

\DIFdelend \subsubsection{Multi-passenger ridesharing constraints}
\label{sec:network_ridesharing_constraint}
This subsection deals with another key component of the proposed model. \DIFdelbegin \DIFdel{In this subsection, multi-passenger }\DIFdelend \DIFaddbegin \DIFadd{Multi-passenger }\DIFaddend ridesharing constraints are introduced for the hyper-network, such that ridesharing passengers will travel together with drivers, and passenger transfer is avoided\DIFaddbegin \DIFadd{. This is achieved }\DIFaddend by constraining ridesharing drivers to follow their selected matching sequence, whereas routes between two tasks are endogenously determined in the route choice model. Note that, relying on the total flow conservation constraints (see Eq.~\DIFdelbegin \DIFdel{\eqref{eq:RD_network_conservation}-\eqref{eq:RP_network_conservation}}\DIFdelend \DIFaddbegin \DIFadd{\eqref{eq:flow_conservation}}\DIFaddend ), the multi-passenger ridesharing constraints are specified for the with-passenger RD flows, such that without-passenger flows also correspond to the proposed multi-passenger ridesharing setting.
\begin{figure}[H]
    \centering
    \includegraphics[width=0.65\textwidth]{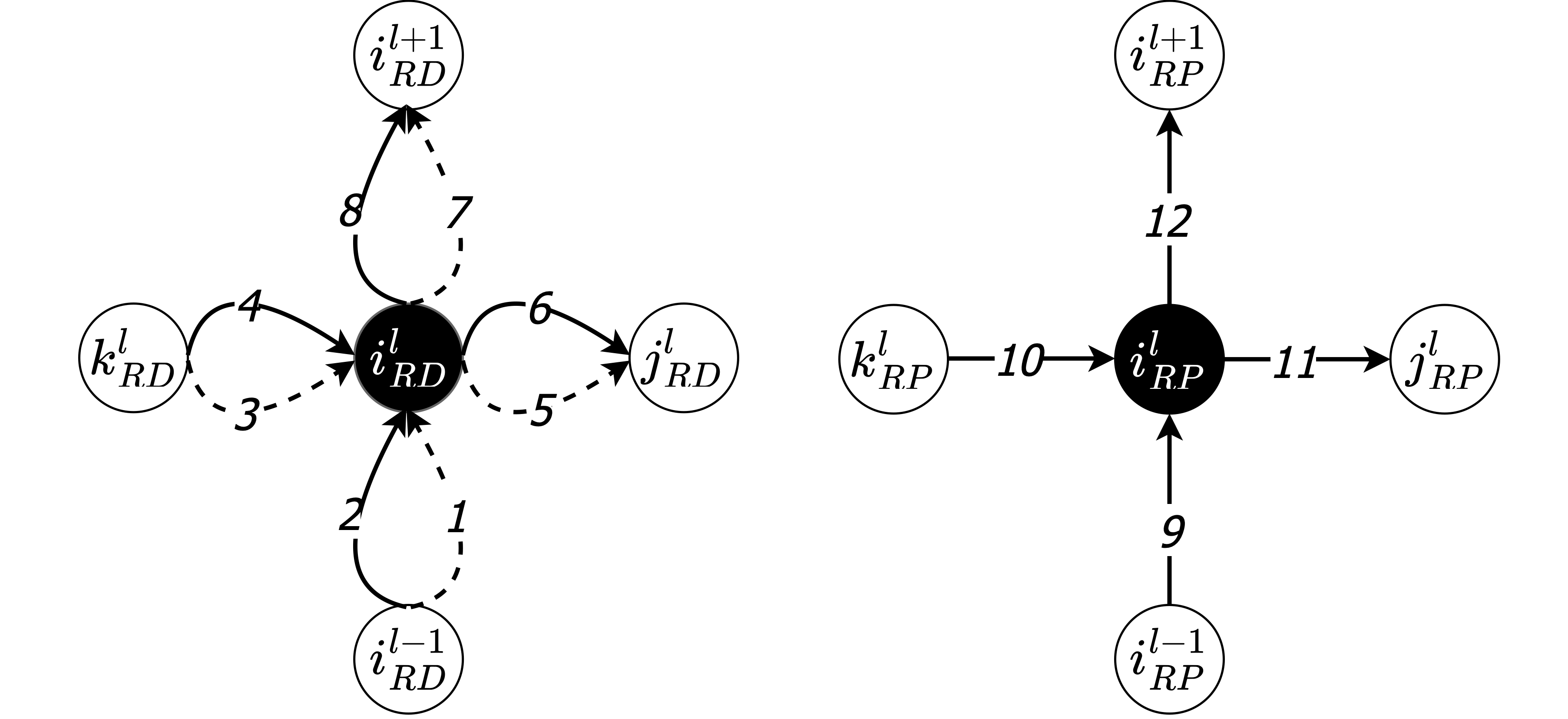}
    \caption{Example of multi-passenger coupling constraints (with each link indexed)}
    \label{fig:8.constraint_illustration}
\end{figure}
\textbf{Ridesharing driver-passenger coupling constraints:}
\begin{flalign} 
\label{eq:RD_RP_coupling_upper}
		\underbrace{\sum_{e \in \mathcal{D}}{x_{i^{l}, j^{l'}}^{n, e, RP}}}_{\text{On-board passengers}} 
		\leq 
		\underbrace{\left[ \sum_{\omega}  {\sum_{l'' \leq l' - 1} \left(s_1(n, l'', \omega) - s_{-1}(n, l'', \omega)\right) } - 1 \right]}_{\text{Vehicle occupancy defined in matching sequence }n}
		\cdot x_{i^{l}, j^{l'}, 1}^{n, RD}
        \quad, \forall i, n, l, j^{l'}:(i^{l}, j^{l'}) \in \mathcal{E}_{RD, 1}
\end{flalign} 
\begin{flalign} 
\label{eq:RD_RP_coupling_lower}
		\sum_{e \in \mathcal{D}}{x_{i^{l}, j^{l'}}^{n, e, RP}} 
		\geq 
		x_{i^{l}, j^{l'}, 1}^{n, RD}
        \quad, \forall i, n, l, j^{l'}:(i^{l}, j^{l'}) \in \mathcal{E}_{RD, 1}
\end{flalign}
Constraint \eqref{eq:RD_RP_coupling_upper}-\eqref{eq:RD_RP_coupling_lower} couple the with-passenger RD link flow with RP link flow, and allow driver to traverse with-passenger links only if the vehicle occupancy defined in the matching sequence is positive. Equation \eqref{eq:RD_RP_coupling_upper} states that, if the driver is traversing link $(i, j)$, the matched passengers can travel on link $(i, j)$, i.e., driver and passengers are traveling together. Moreover, only if the vehicle occupancy at task $l$ is positive (the driver is with some passengers), RP link flow on $(i, j)$ can be positive. Consequently, \cref{eq:RD_RP_coupling_lower} ensures that only drivers with positive vehicle occupancy can travel on with-passenger links (otherwise bounded by zero RP link flow), such that vehicle capacity constraints are also satisfied.

\DIFdelbegin \DIFdel{We illustrate these constraints in \mbox{
\cref{fig:8.constraint_illustration}}\hskip0pt
. Constraint \eqref{eq:RD_RP_coupling_upper} states that ridesharing passengers travel together with drivers. If a ridesharing driver does not travel on link 2, it implies that passengers are not on link 9 in the RP subnetwork; and if passengers have already been traveling on link 9, there must be drivers traveling with them on link 2. Moreover, if the driver travels from node $k$ to pickup at $i$ without any on-board passenger, there is no passenger on link 9 or 10. Consequently, \mbox{
\cref{eq:RD_RP_coupling_lower} }\hskip0pt
indicates that the driver should not use links 2 and 4, which represent traveling with passengers; and only take link 1 or link 3 (due to total flow conservation). 
}\DIFdelend 

Note that, constraints \eqref{eq:RD_RP_coupling_upper}-\eqref{eq:RD_RP_coupling_lower} generalize the capacity constraints in \textcite{xu2015complementarity} and \textcite{di2018link}, by replacing their homogeneous vehicle capacity with matching sequence vehicle occupancy. As discussed in \cref{sec:Matching}, vehicle capacity constraints are handled within the DARP subroutine, which allows heterogeneous vehicle capacities. Moreover, the proposed formulation in this paper explicitly considers the impacts of ridesharing platform matching decisions on route choice behavior of ridesharing drivers and passengers, through pickup/drop-off incidents $s_1$ and $s_{-1}$ in constraint \eqref{eq:RD_RP_coupling_upper}.
\\\textbf{Ridesharing passenger-transfer avoidance constraints:}

One of the unique features of the proposed multi-passenger ridesharing network model is that, passengers need not to make transfers in the ridesharing services. This is accomplished by the following set of constraints that: 1) avoid passengers being dropped off at any intermediate node between two tasks; 2) pickup and drop-off tasks are executed, and 3) sequentially, as defined in the matching sequences. 

\textit{Conservation of with-passenger RD flow at intermediate nodes:}   
\begin{flalign} 
\label{eq:intermediate_conservation_dropoff}
	\DIFdelbegin 
\DIFdelend \DIFaddbegin \begin{split}
		\underbrace{\sum_{k^{l''}:(k^{l''}, i^{l}) \in \mathcal{E}_{RD, 1}} {x_{k^{l''}, i^{l}, 1}^{n, RD}}}_{\text{Incoming with-passenger RD flows}} 
		&- \underbrace{\sum_{j^{l'}:(i^{l}, j^{l'}) \in \mathcal{E}_{RD, 1}} {x_{i^{l}, j^{l'}, 1}^{n, RD}}}_{\text{Outgoing with-passenger RD flows}}
		\leq 
		\underbrace{\left(\sum_{o \in \mathcal{O}}{s_{-1}(n, l, (o, i))}  \right)}_{\text{Drop-off at }i} \\
		{}&\cdot 
		\underbrace{\left(
		\sum_{k^{l''}:(k^{l''}, i^{l}) \in \mathcal{E}_{RD, 0}} {x_{k^{l''}, i^{l}, 0}^{n, RD}}
		+ \sum_{k^{l''}:(k^{l''}, i^{l}) \in \mathcal{E}_{RD, 1}} {x_{k^{l''}, i^{l}, 1}^{n, RD}}
		\right)}_{\text{Total incoming flows}},\; \forall i, n, l \geq 2
   	\end{split}\DIFaddend 
\end{flalign} 
\begin{flalign} 
\label{eq:intermediate_conservation_pickup}
	\DIFdelbegin 
\DIFdelend \DIFaddbegin \begin{split}
		\underbrace{\sum_{j^{l'}:(i^{l}, j^{l'}) \in \mathcal{E}_{RD, 1}} {x_{i^{l}, j^{l'}, 1}^{n, RD}}}_{\text{Outgoing with-passenger RD flows}}
		& - \underbrace{\sum_{k^{l''}:(k^{l''}, i^{l}) \in \mathcal{E}_{RD, 1}} {x_{k^{l''}, i^{l}, 1}^{n, RD}}}_{\text{Incoming with-passenger RD flows}} 	
		\leq 
		\underbrace{\left(\sum_{o \in \mathcal{O}}{s_{1}(n, l, (o, i))}  \right)}_{\text{Pickup at }i} \\
		{}&\cdot 
		\underbrace{\left(
		\sum_{k^{l''}:(k^{l''}, i^{l}) \in \mathcal{E}_{RD, 0}} {x_{k^{l''}, i^{l}, 0}^{n, RD}}
		+ \sum_{k^{l''}:(k^{l''}, i^{l}) \in \mathcal{E}_{RD, 1}} {x_{k^{l''}, i^{l}, 1}^{n, RD}}
		\right)}_{\text{Total incoming flows}},\; \forall i, n, l \geq 2
   	\end{split}\DIFaddend 
\end{flalign}
Constraints \eqref{eq:intermediate_conservation_dropoff}-\eqref{eq:intermediate_conservation_pickup} state that, if there is no pickup or drop-off at node $i^{l}$ in matching sequence $n$, the with-passenger RD link flow is conserved (i.e., \eqref{eq:intermediate_conservation_dropoff}-\eqref{eq:intermediate_conservation_pickup} represent the equality constraint). In case of any pickup/drop-off at this node, the right-hand side of the constraints \eqref{eq:intermediate_conservation_dropoff}-\eqref{eq:intermediate_conservation_pickup} will not be active due to total conservation \DIFdelbegin \DIFdel{\eqref{eq:RD_network_conservation}}\DIFdelend \DIFaddbegin \DIFadd{\eqref{eq:flow_conservation}}\DIFaddend . Note that, constraints \eqref{eq:intermediate_conservation_dropoff}-\eqref{eq:intermediate_conservation_pickup} need not to be defined for level $1$ due to coupling constraints \eqref{eq:RD_RP_coupling_upper}-\eqref{eq:RD_RP_coupling_lower}. When constraints \eqref{eq:intermediate_conservation_dropoff}-\eqref{eq:intermediate_conservation_pickup} are satisfied, it can be interpreted that passengers are not being dropped off at any intermediate node between two tasks.

Using the example in \cref{fig:8.constraint_illustration}, constraints \eqref{eq:intermediate_conservation_dropoff}-\eqref{eq:intermediate_conservation_pickup} state that, if there is no drop-off or pickup at node $i^l$, drivers on link 2 and 4 will only continue with link 6 or 8. Correspondingly, passengers on link 9 and 10 will continue to travel with the driver on link 11 or 12 and not dropped off at this intermediate node.

\textit{Matching sequence intra-task constraint:}

Remind that, each level $l$ in the hyper-network represents that ridesharing drivers are heading to either pick up/drop off at some node according to the matching sequence. After finishing this task, drivers move on to the next level in the hyper-network for performing the next task. This means that executing the corresponding pickup/drop-off task is the only way to exit current level: if there is no pickup or drop-off at node $i^l$, the driver should not cross to the next level. Formally,
\begin{flalign} 
\label{eq:RD_no_pickup_remain}
	\begin{split}
		\sum_{j^{l+1}:(i^{l}, j^{l+1}) \in \mathcal{E}_{RD, 0}^{l, l+1}} {x_{i^{l}, j^{l+1}, 0}^{n, RD}}
    	+& \sum_{j^{l+1}:(i^{l}, j^{l+1}) \in \mathcal{E}_{RD, 1}^{l, l+1}} {x_{i^{l}, j^{l+1}, 1}^{n, RD}}    	 \\
    	 \leq & \left(
    	 	\underbrace{\sum_{e\in \mathcal{D}} {s_1(n,l,(i, e))}}_{\text{Pickups at }i^l} 
    	 	+ \underbrace{\sum_{o\in \mathcal{O}} {s_{-1}(n,l,(o, i))}}_{\text{Drop-offs at }i^l}
    	 \right) \\
    	&\cdot \underbrace{\left( 
    		x_{\widehat{i_{RD}}, i^{l}}^{n, RD}
    		+ \sum_{k^{l''}:(k^{l''}, i^{l}) \in \mathcal{E}_{RD, 0}} {x_{k^{l''}, i^{l}, 0}^{n, RD}}
    		+ \sum_{k^{l''}:(k^{l''}, i^{l}) \in \mathcal{E}_{RD, 1}} {x_{k^{l''}, i^{l}, 1}^{n, RD}} 
    	\right)}_{\text{Total incoming flows}}
    	,  \forall i, n, l
	\end{split}
\end{flalign}
Constraints \eqref{eq:RD_no_pickup_remain} states that, if there is no pickup or drop-off at node $i^l$ (i.e., this is not a task node), the left-hand side of the constraint, total outgoing flows to the next level, should be zero. When the driver performs any pickup/drop-off at this node, these constraints will not be active due to total flow conservations. Constraints \eqref{eq:RD_no_pickup_remain} ensures that pickup/drop-off task at each level is executed, such that passengers are picked up from their origins \textit{or} dropped off at their destinations. These constraints will be completed by additional constraints shown later, to ensure passengers are first picked up at their destinations \textit{and} later dropped off at their destinations.

With the example in \cref{fig:8.constraint_illustration}, constraints \eqref{eq:RD_no_pickup_remain} states, if there is no pickup/drop-off task at node $i^l$, drivers cannot enter \DIFdelbegin \DIFdel{to }\DIFdelend the next level through links 7 or 8, and \DIFdelbegin \DIFdel{consequently }\DIFdelend they remain at level $l$ until their tasks are executed. As a result, \DIFdelbegin \DIFdel{on-board }\DIFdelend \DIFaddbegin \DIFadd{onboard }\DIFaddend passengers continue on link 11 with their drivers.

\textit{Matching sequence inter-task constraint:}

The above intra-task constraint \eqref{eq:RD_no_pickup_remain} ensures each task is being executed. Here, we complete these constraints by enforcing tasks being executed in sequence according to the matching sequence, such that passengers are first picked up from their origins and later dropped off at their destinations, and vehicle capacity constraints are satisfied endogenously. Specifically, this is achieved by enforcing ridesharing drivers to move to the next level (i.e., heading to the next task), after the pickup/drop-off task is finished. Formally, 
\begin{flalign} 
\label{eq:RD_pickup_cross}
		\underbrace{\left(
			\sum_{j^{l}:(i^{l}, j^{l}) \in \mathcal{E}_{RD, 0}^{l, l}} {x_{i^{l}, j^{l}, 0}^{n, RD}}
    		+ \sum_{j^{l}:(i^{l}, j^{l}) \in \mathcal{E}_{RD, 1}^{l, l}} {x_{i^{l}, j^{l}, 1}^{n, RD}}
    	\right)}_{\text{Outgoing flows remaining at level }l}
    	\cdot
       	\left(
    	 	\underbrace{\sum_{e\in \mathcal{D}} {s_1(n,l,(i, e))}}_{\text{Pickups at }i^l} 
    	 	+ \underbrace{\sum_{o\in \mathcal{O}} {s_{-1}(n,l,(o, i))}}_{\text{Drop-offs at }i^l}
    	 \right)
    	 = 0
    	,  \forall i, n, l
\end{flalign}
Constraints \eqref{eq:RD_pickup_cross} states that, if there is pickup or drop-off at node $i^l$ (i.e., this is a task node), the total outgoing flows remaining at current level $l$, should be zero. By total flow conservation, these constraints ensure drivers move to the next level to execute the next task in a matching sequence. Together with intra-task constraint \eqref{eq:RD_no_pickup_remain}, drivers in the proposed network model follow their chosen matching sequences to pick up and drop off passengers.

Constraints \eqref{eq:RD_pickup_cross} can be illustrated in \cref{fig:8.constraint_illustration} that, if driers pick up or drop off passengers at node $i^l$, they will not continue traveling on link 5 or 6. Instead, drivers only move to the next level (together with constraint \eqref{eq:RD_no_pickup_remain}), and their \DIFdelbegin \DIFdel{on-board }\DIFdelend \DIFaddbegin \DIFadd{onboard }\DIFaddend passengers travel together with them on link 12. 

\DIFdelbegin \DIFdel{In the following paragraphs, the complementarity conditions }\DIFdelend \DIFaddbegin \DIFadd{The corresponding MCP formulation }\DIFaddend for the multi-passenger ridesharing constraints \DIFdelbegin \DIFdel{\eqref{eq:RD_RP_coupling_upper}-\eqref{eq:RD_pickup_cross} are formulated, and their dual variables will be included into the route choice model. As previously discussed, the multipliers for the shared constraints could be interpreted differently for RD and RP, respectively. We denotes the dual variables for the coupling constraint \eqref{eq:RD_RP_coupling_upper} as $\lambda_{6}^{i^l, j^{l'}, n}$, and $\widehat{\lambda}_{6}^{i^l, j^{l'},n}$; and $\lambda_{7}^{i^l, j^{l'}, n}$, and $\widehat{\lambda}_{7}^{i^l, j^{l'}, n}$ for the coupling constraint \eqref{eq:RD_RP_coupling_lower}, for ridesharing drivers and passengers, respectively. Let $\lambda_{8}^{i^l, n}$, $\lambda_{9}^{i^l, n}$, $\lambda_{10}^{i^l, n}$, $\lambda_{11}^{i^l, n}$ denote the dual variables for constraints \eqref{eq:intermediate_conservation_dropoff},\eqref{eq:intermediate_conservation_pickup}, \eqref{eq:RD_no_pickup_remain}, and \eqref{eq:RD_pickup_cross}, respectively; and let $\lambda_{8}^{i^1, n}=0$, and $\lambda_{9}^{i^1, n}=0$ for completeness. The complementarity conditions for the joint stable matching and route choice model are summarized as follows:
}\DIFdelend \DIFaddbegin \DIFadd{is provided in Appendix~\ref{appendix.MCP_network_model}. 
}\DIFaddend

\section{The overall general equilibrium model and solution existence}
\label{sec:Overview}
\noindent
As previously shown in \cref{fig:1.general_scheme}, the general equilibrium model is composed of travelers, ridesharing platforms, and transportation networks. Travelers are assumed to make mode choice decisions $q$ with modal costs \eqref{eq:DA_modal_costs}-\eqref{eq:PT_modal_costs} dependent on the ridesharing matching decisions $Z$ and generalized travel costs in the network $\pi$. Ridesharing platforms make matching decisions $Z$ based on ridesharing demands ($q$) and network congestion ($\pi$). The network model determines the link flows $x$ and node potentials $\pi$ given travel demands $q$ and matching sequences ($Z$). Similar to \textcite{ban2019general}, the proposed general equilibrium model is defined formally as:

\begin{definition}[GEM-mpr: General equilibrium for multi-passenger ridesharing systems]
  \label{def:general_equilibrium}
  By definition, this is a tuple $\left(q, Z, x, \pi \right)$ such that, 
  the tuple $\left\{q_{\omega}^{m}: (\omega, m) \in \mathcal{W} \times \mathcal{M} \right\}$ is an optimal solution of mode choice model \eqref{eq:mode_choice_opt} with modal costs $C_{\omega}^{m}$ given by \eqref{eq:DA_modal_costs}-\eqref{eq:PT_modal_costs}; the tuple $\left\{Z_n \right\}$ is an optimal solution of the ridesharing matching problem \eqref{eq:matching_opt} with respect to $R_{n}$; and the tuple $\{(x, \pi)\}$ satisfies the network model \eqref{eq:mcp_group_coupling_constraints}-\eqref{eq:mcP_RP_stable_route_final}.
\end{definition}

\subsection{A preview of the proof of solution existence}
\label{sec:proof_preview}
Solution existence for the equilibrium models are typically cast as a variational inequality (VI) or complementarity problem (CP), and applied the theory of these problems (\cite{Nagurney2009}; \cite{facchinei2003finite}). However, certain realistic features of the proposed general equilibrium model make the demonstration of the solution existence challenging. For example, unlike standard (symmetric) traffic assignment problems (\cite{dafermos1969traffic}), in which constraints contain only flow variables from the same OD (viewed as one \textit{player}). The proposed equilibrium model includes traveler's mode choice decision variables $q$ in the constraints of the ridesharing platform's matching optimization problem \eqref{eq:RD_demand_constraint}-\eqref{eq:RP_demand_constraint}. Similarly, in the stable matching problem, constraint \eqref{eq:stable_matching_constraint} couples the RD player with the RP player. These features make the proposed model a generalized type of Nash equilibrium (\cite{facchinei2007generalized}), which should be carefully addressed before applying general theories in VI/CP and a standard fixed-point approach to our problem.

To summarize, the proof of the existence of an equilibrium solution for the proposed general equilibrium model is challenging from the following aspects:
\begin{itemize}
  \item \textbf{Shared constraints}. Coupling of variables between various players within one player's own optimization problem (e.g., constraints \eqref{eq:stable_matching_constraint}, \eqref{eq:RD_RP_coupling_upper}, and \eqref{eq:RD_RP_coupling_lower}). This brings challenges to the proof of solution existence, since the dual variables of these shared constraints represent the player-specific marginal prices, which need not be the same for different players.
  \item \textbf{Unboundedness}. The lack of explicit bounds on some model variables makes the prerequisite of being a compact convex set in applying a fixed point theorem difficult to satisfy (\cite{ban2019general}). 
  \item \textbf{Asymmetric multipliers}. Some of the shared constraints appear only in one player's optimization problem, while not in the other player's problem. As a result, these asymmetric coupling constraints cannot be simply added to the VI formulation, otherwise, the multipliers will be added to both players. 
\end{itemize}  

\DIFdelbegin \DIFdel{In the following section, we }\DIFdelend \DIFaddbegin We extend the work of \textcite{ban2019general} \DIFdelbegin \DIFdel{to show }\DIFdelend \DIFaddbegin \DIFadd{for showing }\DIFaddend that the proposed general equilibrium model for multi-passenger ridesharing admits at least one solution under mild assumptions. \DIFdelbegin \DIFdel{Readers may skip the detailed proof and proceed directly to the numerical result section. To tackle the challenges mentioned above, several main steps are taken for the proof :
}\DIFdelend \DIFaddbegin \DIFadd{The main steps of the proof is summarized as follows:
}

\DIFaddend \begin{itemize}
  \item \textbf{Normalized equilibrium} (\cite{rosen1965existence}): a \textit{normalized} equilibrium assumes that the multipliers corresponding to a shared constraint are proportional to the player's own proportionality constant. Specifically, this means there exist positive constants $\eta_{5}^{n}, \eta_{67}^{i^l,j^{l'},n}$ such that $\widehat{\phi}_{5}^{n,l}=\eta_{5}^{n}{\phi}_{5}^{n,l}, \widehat{\lambda}_{6}^{i^l,j^{l'},n}=\eta_{67}^{i^l,j^{l'},n}{\lambda}_{6}^{i^l,j^{l'},n}$ , and $\widehat{\lambda}_{7}^{i^l,j^{l'},n}=\eta_{67}^{i^l,j^{l'},n}{\lambda}_{7}^{i^l,j^{l'},n}$.
  \item \textbf{VI existence theorem for unbounded set and linear complementarity lemma}: The existence of the primary decision variables is shown alternatively using existence theorem for unbounded feasible set (Theorem 2 in \cite{Nagurney2009}). The existence of unbounded dual variables is completed using linear complementarity lemma (\cite{cottle2009linear}).
  \item \textbf{Relaxation of the asymmetric shared constraints through penalization}: The asymmetric multipliers are relaxed with a penalty scheme, and shown to recover the proposed general equilibrium model 
  by taking the limit of the penalties at infinity (\cite{ban2019general}).
\end{itemize}

\begin{theorem}
The proposed GEM-mpr admits a normalized equilibrium solution, under mild assumptions.
\end{theorem}
\begin{proof}

See Appendix~\ref{sec:Existence}.
\end{proof}

\subsection{Discussion of uniqueness}
For a highly complex general equilibrium model like the GEM-mpr proposed in this paper, it is difficult to expect a general solution uniqueness. One of the typical prerequisites for solution uniqueness in VI/CP is that, the cost functions are strictly monotone (\cite{facchinei2003finite}, \cite{Nagurney2009}). However, this is often difficult to satisfy in the proposed general model for multi-passenger ridesharing, as only vehicle flows directly contribute to the traffic congestion, the increase of passenger flows need not \DIFdelbegin \DIFdel{to }\DIFdelend increase the link travel time. To a certain extent, the proposed general model resembles some of the properties of asymmetric traffic assignment problems, where strictly monotone link travel times are generally not satisfied as well. Solution uniqueness could be established under restrictive assumptions on the link cost functions (\cite{li2020restricted}, \cite{li2020path}). For the generality of the proposed model, we leave such specifications for future research.

\section{Sequence-bush assignment method}
\label{sec:algorithm}

The interactions between the players in the \DIFdelbegin \DIFdel{proposed }\DIFdelend GEM-mpr model indicate a non-separable \DIFdelbegin \DIFdel{structure of the problem , which is difficult }\DIFdelend \DIFaddbegin \DIFadd{problem structure, which makes it challenging }\DIFaddend to solve. The \DIFdelbegin \DIFdel{motivations of the }\DIFdelend proposed solution method \DIFdelbegin \DIFdel{is to decouple the player's decisions through augmented Lagrangian method}\DIFdelend \DIFaddbegin \DIFadd{tackles this challenge by: a) decomposing the complex GEM-mpr problem into simpler subproblems, and b) solving each subproblem efficiently.
}

\DIFadd{We propose using the Augmented Lagrangian method~}\DIFaddend \citep{kanzow2016augmented} \DIFdelbegin \DIFdel{, and }\DIFdelend to decompose the complex \DIFdelbegin \DIFdel{network model (\mbox{
\cref{sec:RidesharingNetworkModel}}\hskip0pt
) into (relatively) simpler subproblems. Yet, the stable matchingconstraints \eqref{eq:stable_matching_constraint} and the multi-passenger ridesharing constraints \eqref{eq:RD_RP_coupling_upper}-\eqref{eq:RD_pickup_cross} pose challenges for solving the subproblems efficiently}\DIFdelend \DIFaddbegin \DIFadd{GEM-mpr problem into three subproblems: traveler mode choice, ridesharing matching, and network assignment. Since the network assignment subproblem is the most challenging to solve, we exploit the problem structure to further decompose it for each RD-RP group who are matched together. We summarize the proposed AL solution scheme in the flowchart (Appendix~\mbox{
\cref{fig:algorithm_flowchart}}\hskip0pt
)}\DIFaddend . 

To \DIFdelbegin \DIFdel{tackle the above problem, we propose a sequence-bush assignment method which embeds these constraints into the route-flow variables. As will be shown, }\DIFdelend \DIFaddbegin \DIFadd{efficiently solve the challenging network assignment subproblem, we develop an assignment algorithm that integrates matching sequence with bush-based algorithm~}\citep{dial2006path,nie2010class}\DIFadd{. Several steps are taken in the sequence-bush algorithm to handle the complex ridesharing constraints, such that }\DIFaddend solving a subproblem is similar to a classic \DIFdelbegin \DIFdel{traffic assignment problem. 
The following subsections correspond to steps taken to the sequence-bush assignment method: 1) auxiliary variables }\DIFdelend \DIFaddbegin \DIFadd{TAP. 
Specifically, the stable matching constraints \eqref{eq:stable_matching_constraint} and transfer avoidance constraints \eqref{eq:intermediate_conservation_dropoff}-\eqref{eq:RD_pickup_cross} are embedded into the auxiliary sequence-flow variables }\DIFaddend $F$\DIFdelbegin \DIFdel{representing flows on the matching sequence are introduced into the network model to embed the complex constraints; 2) the network model is }\DIFdelend \DIFaddbegin \DIFadd{. The network model is also }\DIFaddend rewritten in terms of route-flow variables \DIFdelbegin \DIFdel{with }\DIFdelend \DIFaddbegin \DIFadd{$f$ with an }\DIFaddend unrestricted path set\DIFdelbegin \DIFdel{to bridge the }\DIFdelend \DIFaddbegin \DIFadd{. This approach not only bridges }\DIFaddend sequence-flow \DIFdelbegin \DIFdel{variables }\DIFdelend $F$ \DIFdelbegin \DIFdel{with }\DIFdelend \DIFaddbegin \DIFadd{and }\DIFaddend link-flow variables $x$\DIFdelbegin \DIFdel{; 3)the }\DIFdelend \DIFaddbegin \DIFadd{, but also enables the coupling of RD and RP flows (constraints~\ref{eq:RD_RP_coupling_upper}-\ref{eq:RD_RP_coupling_lower}).
}

\DIFadd{The block proximal algorithm~}\citep{bolte2014proximal} \DIFadd{is adapted to successively solve the subproblems at each iteration, whereas the multipliers are updated once the inner loop converges~}\citep{kanzow2016augmented}\DIFadd{. In the following subsections, we first detail the derivations of a sequence-bush assignment problem, and present the }\DIFaddend overall solution framework\DIFdelbegin \DIFdel{is presented}\DIFdelend .

\subsection{An equivalent reformulation with auxiliary sequence-flow variables}
\label{subsec:alg_step1}
\noindent
In this subsection, special attention is given to the RD and RP demands in the network model, whereas DA and PT demands are considered as special cases of RD and RP demands without any pickup or drop-off. To unify the notations, we first define the set of link-flow classes for each mode:
\begin{itemize}
  \item \DIFdelbegin \DIFdel{matching sequences $n$ associated with RDbetween OD $w$, indexed by level $l$ and occupancy (-1}\DIFdelend \DIFaddbegin \DIFadd{RD}\DIFaddend : \DIFdelbegin \DIFdel{without-passenger, 0: with-passenger):
}\DIFdelend \DIFaddbegin \DIFadd{$\Psi^{w, RD}=\{(n, l, -1), (n, l, 0)|n:s_1(n,0,w)=1, 1\leq l \leq L \}$
}\DIFaddend 

  \DIFdelbegin \DIFdel{$\Psi^{w, RD}=\{(n, l, -1), (n, l, 0)|n:s_1(n,0,w)=1, 1\leq l \leq L \}$, }\DIFdelend \DIFaddbegin \DIFadd{where, RD traveling without passenger is represented by -1, and 0 otherwise.
  }\DIFaddend \item \DIFdelbegin \DIFdel{matching sequences associated with RP between OD $w$, indexed by level $l$, pickup index $u$, and matched RD $w'$. Let $(n^{w, PT}, 1, 1)$ denotes RP leaving for PT, we have
}

\DIFdelend \DIFaddbegin \DIFadd{RP: }\DIFaddend $\Psi^{w, RP}= \{(n^{w, PT}, 1, 1)\} \cup_{w'} \Psi^{w, RP}_{w'}$

  where, \DIFaddbegin \DIFadd{$(n^{w, PT}, 1, 1)$ denote RP leaving for PT, and $\Psi^{w, RP}_{w'}$ denote the set of matching sequences served by RD $w'$. Specifically, }\DIFaddend by assuming RP are dropped off \textit{as-soon-as-possible}, \DIFaddbegin \DIFadd{we have:
}

  \DIFaddend $\Psi^{w, RP}_{w'}=\{(n, l, u)|n:s_1(n,u,w)=1,s_1(n,0,w')=1, u\leq l \leq l':\min\{l':s_{-1}(n,l',w)=1, l'>u \} \}$ 
  \DIFdelbegin \DIFdel{representing the set of matching subsequences of RP $w$ served by RD $w'$, 
  }\DIFdelend \item \DIFdelbegin \DIFdel{for completeness, }\DIFdelend \DIFaddbegin \DIFadd{DA and RP: }\DIFaddend $\Psi^{w, DA}=\{(n^{w, DA}, 1, 0)\}$ and $\Psi^{w, PT}=\{(n^{w, PT}, 1, 0)\}$,
  \item \DIFdelbegin \DIFdel{and the set of all }\DIFdelend \DIFaddbegin \DIFadd{All }\DIFaddend link-flow classes\DIFdelbegin \DIFdel{of OD $w$, }\DIFdelend \DIFaddbegin \DIFadd{: }\DIFaddend $\Psi^w = \cup_{m \in \mathcal{M}} \Psi^{w, m}$.
\end{itemize}

Recall that, the stable matching (Eq.~\ref{eq:stable_matching_constraint}) and passenger pickup drop-off (Eqs.~\ref{eq:RD_no_pickup_remain}-\ref{eq:RD_pickup_cross}) are handled through link flow coupling in the network model. These constraints essentially state that 1) \DIFdelbegin \DIFdel{the }\DIFdelend \DIFaddbegin \DIFadd{Equal }\DIFaddend RD and RP demands on the same matching sequence \DIFdelbegin \DIFdel{are equal for a }\DIFdelend \DIFaddbegin \DIFadd{under }\DIFaddend stable matching; and 2) \DIFaddbegin \DIFadd{Conservation of }\DIFaddend RD and RP demands between consecutive levels \DIFdelbegin \DIFdel{should be conserved }\DIFdelend with occupancy updated according to the matching sequences. Motivated by this observation, we introduce \DIFdelbegin \DIFdel{an }\DIFdelend auxiliary sequence-flow \DIFdelbegin \DIFdel{variable }\DIFdelend \DIFaddbegin \DIFadd{variables }\DIFaddend $F_{n}$ \DIFdelbegin \DIFdel{indexed by its matching sequence $n$ }\DIFdelend to represent these constraints at the sequence-flow space.

Following the gap function formulation \citep{lo2000reformulating}, the \DIFaddbegin \DIFadd{link-based }\DIFaddend network model presented in \cref{sec:RidesharingNetworkModel} can be rewritten as follows:
\begin{subequations}\label{eq:reformulation_sequence_flow}
\begin{align}
    \min_{x, \pi, F}\quad & \sum_{w\in\mathcal{W}} 
    \sum_{\psi \in \Psi^w} \sum_{ij}  x_{ij}^{\psi, w} (\pi_j^{\psi, w} + c_{ij}^{\psi}(x_{ij}) - \pi_i^{\psi, w}) & \label{eq:reformulation_sequence_obj}\\
    s.t. \quad & \sum_{k:(j,k)}x_{jk}^{\psi, w} -  \sum_{i:(i, j)}x_{ij}^{\psi, w} = p_j^{\psi, w},\; \forall j \in \mathcal{N}, w \in \mathcal{W}, \psi \in \Psi^w &\label{eq:reformulation_link_conservation} \\
    & x_{ij}^{RD, w} = \sum_{\psi \in \Psi^{w, RD}} x_{ij}^{\psi, w} ,\; \forall (i,j) \in \mathcal{E}_{Veh}, w \in \mathcal{W} \label{eq:reformulation_def_constrant_first}\\
    & x_{ij}^{RP, w} = \sum_{\psi \in \Psi^{w, RP}} x_{ij}^{\psi, w} & \nonumber \\
    & = \underbrace{\sum_{\psi \in \Psi_{w'}^{w, RP}} x_{ij}^{\psi, w}}_{\text{RP traveling with RD }w'}+ \sum_{w'' \neq w'} \sum_{\psi \in \Psi^{w, RD}_{w''}} x_{ij}^{\psi, w} +  x_{ij}^{(n^{w, PT},1, 1), w},\; \forall (i,j) \in \mathcal{E}_{Veh}, w \in \mathcal{W}  & \\
    & x_{ij} = \begin{cases}
    \sum_{w\in\mathcal{W}} \left(x_{ij}^{RD, w} + x_{ij}^{(n^{w, DA}, 1, 0), w} \right), \; \forall (i,j) \in \mathcal{E}_{Veh} &\\
    \sum_{w\in\mathcal{W}} \left(x_{ij}^{(n^{w, PT}, 1, 1), w} + x_{ij}^{(n^{w, PT}, 1, 0), w} \right), \; \forall (i,j) \in \mathcal{E}_{PT} &
    \end{cases}  \label{eq:reformulation_def_constrant_last}\\
    & \underbrace{\sum_{n \in \{n: (n, \cdot, \cdot) \in \Psi^{w, RD}\}}  F_{n}}_{\text{All sequence flows associated with RD }w} = q_w^{RD}, \; \forall w \in \mathcal{W}& \label{eq:reformulation_RD_seq_flow}\\
    & \underbrace{\sum_{n \in \{n:(n,\cdot,\cdot) \in \Psi^{w, RP}\}} F_{n}}_{\text{All sequence flows associated with RP }w} = q_w^{RP} , \; \forall w \in \mathcal{W}& \label{eq:reformulation_RP_seq_flow}\\
    & F_n \leq Z_n, \; \forall n  & \label{eq:reformulation_cap}\\
    & x_{ij}^{\psi, w} \geq 0,\; \pi_j^{\psi, w}\geq0, \;\pi_j^{\psi, w} + c_{ij}^{\psi}(x_{ij}) - \pi_i^{\psi, w} \geq 0, \forall (i,j) \in \mathcal{E}, w \in \mathcal{W}, \psi \in \Psi^w \label{eq:reformulation_nonnegative}\\
    & \textbf{Ridesharing driver-passenger coupling constraints} \quad\quad\quad \text{\eqref{eq:RD_RP_coupling_upper}-\eqref{eq:RD_RP_coupling_lower}} \nonumber &
\end{align}
\end{subequations}
\DIFdelbegin \DIFdel{Let $B_{n, l} \in \{-1,0\}$ denotes the RD status of sequence $n$ at level $l$, and $B_{n, l}=-1$ if traveling alone at $l$: $\sum_{l' < l}\sum_{w\in\mathcal{W}} \left(s_1(n, l, w)-s_{-1}(n, l, w)\right) - 1 = 0$, and $B_{n, l}=0$ , otherwise; and $s(n, l) \rightarrow i: \sum_{j \in \mathcal{O} \cup \mathcal{D}} (s_1(n, l, (i, j)) + s_{-1}(n, l, (j, i))) = 1$ maps the pickup/drop-off node of $n$ at level $l$. }\DIFdelend We define $p_j^{\psi, w}$ for the conservation constraint \eqref{eq:reformulation_link_conservation} as follows:
\begin{subnumcases}{p_j^{\psi, w} = }
  q_w^{\psi}, & $\psi \in \{(n^{w, DA}, 0, 0), (n^{w, PT}, 0, 0)\}, j = o(w)$ 
  \\
  - q_w^{\psi}, & $\psi \in \{(n^{w, DA}, 0, 0), (n^{w, PT}, 0, 0)\}, j = d(w)$
  \\
   F_n, & $\psi:(n, l, b)\in \Psi^{w, RD}, j = s(n ,l-1), B_{n,l} = b \label{eq:reformulation_RD_ori}$ \\
  -F_n, & $\psi:(n, l, b)\in \Psi^{w, RD}, j = s(n ,l), B_{n, l} = b \label{eq:reformulation_RD_des}$
  \\
  F_n,  & $\psi:(n, l, u) \in \Psi^{w, RP} , l = u, j = o(w)$ \label{eq:reformulation_RP_ori}\\
  -F_n,  & $\psi:(n, l, u) \in \Psi^{w, RP} , l > u, j = d(w)$ \label{eq:reformulation_RP_des} \\
  0, & \text{otherwise}
  \end{subnumcases}
\DIFaddbegin \DIFadd{where, $B_{n, l} \in \{-1,0\}$ denote the RD status of sequence $n$ at level $l$, and $B_{n, l}=-1$ if traveling alone at $l$: $\sum_{l' < l}\sum_{w\in\mathcal{W}} \left(s_1(n, l, w)-s_{-1}(n, l, w)\right) - 1 = 0$, and $B_{n, l}=0$ , otherwise; and $s(n, l) \rightarrow i: \sum_{j \in \mathcal{O} \cup \mathcal{D}} (s_1(n, l, (i, j)) + s_{-1}(n, l, (j, i))) = 1$ maps the pickup/drop-off node of $n$ at level $l$.
}\DIFaddend 

At optimal, objective function \eqref{eq:reformulation_sequence_obj} is equal to 0, which, together with non-negative constraints \eqref{eq:reformulation_nonnegative}, represents the equilibrium conditions for route choice in \cref{sec:RidesharingNetworkModel}. 
The stable matching constraint \eqref{eq:stable_matching_constraint} is captured by sharing the same variable $F_n$ for both RD and RP in Eqs.~\eqref{eq:reformulation_RD_seq_flow}-\eqref{eq:reformulation_RP_seq_flow} and Eqs.~\eqref{eq:reformulation_RD_ori}-\eqref{eq:reformulation_RP_des}, such that RD and RP demands on the same matching sequence \DIFdelbegin \DIFdel{is }\DIFdelend \DIFaddbegin \DIFadd{are }\DIFaddend equal. 
The \textit{ridesharing passenger-transfer avoidance constraints}~\eqref{eq:intermediate_conservation_dropoff}-\eqref{eq:RD_pickup_cross} are rewritten into conservation conditions~\eqref{eq:reformulation_RD_ori}-\eqref{eq:reformulation_RD_des}, which are specified at pickup and drop-off nodes. Consequently, class \DIFdelbegin \DIFdel{link flows $(n, l, b)$ are only updated at pickup/drop-off nodes, such that class }\DIFdelend (e.g., with/without-passenger) \DIFaddbegin \DIFadd{link }\DIFaddend flows at intermediate nodes are conserved (as in the original constraints~\ref{eq:intermediate_conservation_dropoff}-\ref{eq:intermediate_conservation_pickup}). 
Moreover, through conservation condition~\eqref{eq:reformulation_link_conservation}, the hyper-network can be decomposed into $L$ levels, and associated each level $l$ with a virtual origin $s(n, l-1)$ and a virtual destination $s(n,l)$ with variable demands $F_n$. As a result, conservation between consecutive levels (as described in the original constraints~\eqref{eq:RD_no_pickup_remain}-\eqref{eq:RD_pickup_cross}) is also satisfied by the same sequence flow variable $F_n$. 

The formulation presented above \DIFdelbegin \DIFdel{suggests }\DIFdelend \DIFaddbegin \DIFadd{indicates }\DIFaddend that each matching sequence \DIFdelbegin \DIFdel{could be decomposed into several levels }\DIFdelend \DIFaddbegin \DIFadd{is composed of several levels connected by the sequence-flow variable }\DIFaddend (motivation of a \textit{sequence-bush} method). Whereas the \DIFdelbegin \DIFdel{sequence flow }\DIFdelend \DIFaddbegin \DIFadd{shared }\DIFaddend variable $F_n$ \DIFdelbegin \DIFdel{suggests the link flows associated with RP and their matched RD }\DIFdelend \DIFaddbegin \DIFadd{between RD and RP suggests that the network assignment subproblem }\DIFaddend should be solved collectively \DIFaddbegin \DIFadd{for the RD and their matched RPs }\DIFaddend (motivation of \DIFdelbegin \DIFdel{the decomposed subproblems}\DIFdelend \DIFaddbegin \DIFadd{a RD-RP group decomposition}\DIFaddend ).

Note that, the proposed formulation is one step closer to the classic \DIFdelbegin \DIFdel{traffic assignment problems}\DIFdelend \DIFaddbegin \DIFadd{TAP}\DIFaddend , except the \DIFdelbegin \DIFdel{ridesharing driver-passenger }\DIFdelend \DIFaddbegin \DIFadd{RD-RP }\DIFaddend coupling constraints~\eqref{eq:RD_RP_coupling_upper}-\eqref{eq:RD_RP_coupling_lower}, the additional sequence-flow variables $F_n$ introduced here, and the capacity constraints~\eqref{eq:reformulation_cap}. The following subsections deal with these constraints to allow solving this problem with traffic assignment methods. 

\subsection{A path-based representation with unrestricted path set}
\label{subsec:alg_step2}
\noindent
The auxiliary sequence-flow variables $F_n$ presented in the previous subsection simplify the set of constraints at the cost of additional computation of variables $F_n$. In this subsection, we introduce path-flow variables $f$ with \DIFaddbegin \DIFadd{an }\DIFaddend unrestricted route set to restate link flow variables $x$ and sequence-flow variables $F$, such that the problem requires solving only $f$. 

Let $K_{od}$ denote the set of all acyclic feasible paths between node $o$ and $d$, and $f_k^{\psi, w}$ \DIFdelbegin \DIFdel{denotes }\DIFdelend \DIFaddbegin \DIFadd{denote }\DIFaddend the flow on route $k$ for class $\psi:(n,l,b)$, where route $k \in K_{o^{nl}d^{nl}}$ (i.e., $k$ is one of the paths that originate from previous task location $o^{nl}=s(n, l-1)$ to current task location $d^{nl}=s(n, l)$). We now restate the link-flow variables $x$ and sequence-flow variables $F_n$ in terms of route-flow variables $f_k^{\psi, w}$ as follows:
\begin{align}
x_{ij}^{\psi, w} &= \sum_{k \in K_{o^{nl}d^{nl}}} \delta_{ij}^k f_{k}^{\psi, w}, \; \forall (i,j) \in \mathcal{E}, w \in \mathcal{W}, \psi:(n, l, \cdot) \in \Psi^{w} \label{eq:link_flow_route_flow_map}\\
F_n &=\begin{cases}
 \sum_{k \in K_{o^{nl}d^{nl}}} f_{k}^{\psi, w}, \; \forall  w \in \mathcal{W}, \psi:(n, l, b) \in \Psi^{w, RD}, b = B_{n,l}\\
 \sum_{k \in K_{o^{nl}d^{nl}}} f_{k}^{\psi, w}, \; \forall  w \in \mathcal{W}, \psi:(n, l, u) \in \Psi^{w, RP}
\end{cases} \label{eq:seq_flow_route_flow_map}
\end{align}  
where, $\delta_{ij}^{k}=1$ if path $k$ traverses link $(i, j)$, and $\delta_{ij}^{k} = 0$, otherwise. Note that Eq.~\eqref{eq:seq_flow_route_flow_map} is defined for each level $l$ in matching sequence $n$, with its corresponding occupancy status $b=B_{n,l}$ (in case of RD). This definition implies RD finishing one task will all head to the next task (i.e., conservation between levels), and their occupancy status (with/without passenger) \DIFdelbegin \DIFdel{are }\DIFdelend \DIFaddbegin \DIFadd{is }\DIFaddend explicitly defined according to the matching sequences. Furthermore, stable matching is established by the equal RD and RP demands ($F_n$), which are aggregated over all route flows between two tasks in a matching sequence (Eq.~\ref{eq:seq_flow_route_flow_map}). 

By substituting route-flows variable $f_k^{\psi, w}$ into \DIFdelbegin \DIFdel{sequence-flow }\DIFdelend \DIFaddbegin \DIFadd{the link-based }\DIFaddend formulation~\eqref{eq:reformulation_sequence_flow}, we obtain the following problem with $f_{k}^{\psi, w}$ being the primary decision variable:
\begin{subequations}\label{eq:reformulation_route_flow}
\begin{align}
    \min_{f, \pi}\quad & \sum_{w\in\mathcal{W}} 
    \sum_{\psi:(n,l,\cdot) \in \Psi^w} \sum_{k \in K_{o^{nl}d^{nl}}} f_{k}^{\psi, w} \left(\sum_{ij}\delta_{ij}^{k}c_{ij}^{\psi} - \pi_{o^{nl}}^{\psi, w}\right) & \label{eq:reformulation_route_obj}\\
    s.t. \quad & \underbrace{\sum_{\substack{\psi:(n,l',\cdot) \in \Psi^{w, m},\\ l' = l}} \sum_{k \in K_{o^{nl}d^{nl}}}  f_{k}^{\psi, w} = q_w^{m}}_{\text{Modal demand conservation}}, \; \forall w \in \mathcal{W}, m \in \mathcal{M}, 1\leq l \leq L& \label{eq:reformulation_route_demand_conservation}\\
    & \underbrace{\sum_{k \in K_{o^{nl}d^{nl}}}  f_{k}^{\psi, w} \leq Z_n}_{\text{Platform matched demand constraint}}, \; \forall w \in \mathcal{W}, \psi:(n,l,\cdot) \in \Psi^{w, RD}  & \label{eq:reformulation_route_cap}\\
    & \underbrace{f_{k}^{\psi, w} = f_{k}^{\psi^{'}, w'}}_{\substack{\text{RD-RP coupling} \\{\text{\& stable matching}}}}, \; \forall w, w' \in \mathcal{W}, \psi:(n,l,b) \in \Psi^{w, RD}, b=0, \nonumber\\ 
    & \quad\quad\quad\quad\quad\quad\quad \psi^{'}:(n',l',\cdot) \in \Psi^{w', RP}, n = n', l = l', k \in K_{o^{nl}d^{nl}}  & \label{eq:reformulation_route_RD_RP_coupling}\\
    & \underbrace{\sum_{k \in K_{o^{nl}d^{nl}}} f_{k}^{\psi, w} = \sum_{k' \in K_{o^{nl+1}d^{nl+1}}} f_{k}^{(n,l+1,B_{n,l+1}), w}}_{\text{RD conservation between two consecutive levels}}, \; \forall w \in \mathcal{W}, \psi:(n,l,b) \in \Psi^{w, RD}, b = B_{n,l } \label{eq:reformulation_route_cross_level_conservation}\\ 
    & f_{k}^{\psi, w} \geq 0,\; \pi_j^{\psi, w}\geq0, \;\pi_j^{\psi, w} + c_{ij}^{\psi} - \pi_i^{\psi, w} \geq 0, \nonumber \\
    & \quad\quad\quad\quad\quad\quad\quad \forall (i,j) \in \mathcal{E}, w \in \mathcal{W}, \psi:(n,l,\cdot) \in \Psi^w, k \in K_{o^{nl}d^{nl}} \label{eq:reformulation_route_nonnegative}
\end{align}
\end{subequations}
We include here two new constraints (Eqs.~\ref{eq:reformulation_route_RD_RP_coupling}-\ref{eq:reformulation_route_cross_level_conservation}) that are not presented in the sequence-flow formulation, to handle RP transfer avoidance, RD-RP coupling and stable matching. Specifically, the RD-RP coupling is handled at the route-flow space (Eq.~\ref{eq:reformulation_route_RD_RP_coupling}) to represent that RD and RP \DIFdelbegin \DIFdel{in a matching sequence are traveling together }\DIFdelend \DIFaddbegin \DIFadd{travel }\DIFaddend on the same route \DIFdelbegin \DIFdel{. Together with constraint~\eqref{eq:reformulation_route_cross_level_conservation} , which requires }\DIFdelend \DIFaddbegin \DIFadd{in a matching sequence. Constraint~\eqref{eq:reformulation_route_cross_level_conservation} ensures no RP transfer by requiring }\DIFaddend RD to follow the matching sequence for pickup and drop-off\DIFdelbegin \DIFdel{, RP transfer is avoided}\DIFdelend . Furthermore, coupling constraint~\eqref{eq:reformulation_route_RD_RP_coupling} implies \DIFaddbegin \DIFadd{also }\DIFaddend stable matching, which can be verified by aggregating route flows using Eq.~\eqref{eq:seq_flow_route_flow_map}.

\DIFdelbegin 

\DIFdelend Note that, most of the constraints in formulation~\eqref{eq:reformulation_route_flow} are only defined for RD, while RP are jointly constrained through coupling (Eq.~\ref{eq:reformulation_route_RD_RP_coupling}). This problem structure suggests a decomposition \DIFdelbegin \DIFdel{, in the spirit of Gauss-Seidel method, }\DIFdelend with respect to RD-RP group who \DIFdelbegin \DIFdel{share the same matching sequences}\DIFdelend \DIFaddbegin \DIFadd{are matched together}\DIFaddend , namely, for RD $w'$ and RPs $\{w | \Psi^{w, RP}_{w'} \neq \emptyset \}$. Such decomposition is similar to the bush approach, which decomposes the traffic assignment problem with respect to origins~\citep{bar2002origin,dial2006path}.
\DIFaddbegin 

\DIFaddend We now define the subproblem with respect to a RD-RP group $w$, while treating flows contributed by other RD-RP groups as constant background flows, as follows:
\begin{subequations}\label{eq:reformulation_subprob_route_flow}
\begin{align}
    \min_{f, \pi}\quad &  
    \sum_{\psi:(n,l,\cdot) \in \Psi^{w, RD}} \sum_{k \in K_{o^{nl}d^{nl}}} f_{k}^{\psi, w} \left(\sum_{ij}\delta_{ij}^{k}c_{ij}^{\psi} - \pi_{o^{nl}}^{\psi, w}\right) & \nonumber\\
    & + \sum_{w' \in \mathcal{W}} \sum_{\psi:(n,l,\cdot) \in \Psi^{w', RP}_{w} \cup \{(n^{w', PT}, 1,1)\}} \sum_{k \in K_{o^{nl}d^{nl}}} f_{k}^{\psi, w'} \left(\sum_{ij}\delta_{ij}^{k}c_{ij}^{\psi} - \pi_{o^{nl}}^{\psi, w'}\right)  & \label{eq:reformulation_subprob_route_obj}\\
    s.t. \quad & \sum_{\substack{\psi:(n,l',\cdot) \in \Psi^{w, RD},\\ l' = l}} \sum_{k \in K_{o^{nl}d^{nl}}}  f_{k}^{\psi, w} = q_w^{RD}, \; \forall 1\leq l \leq L& \label{eq:reformulation_subprob_RD_conservation}\\
    &  f_{0}^{w', RD} + \sum_{\substack{\psi:(n,l',\cdot) \in \Psi^{w', RP}_{w} \cup \{(n^{w', PT}, 1, 1)\},\\ l' = l}} \sum_{k \in K_{o^{nl}d^{nl}}}  f_{k}^{\psi, w'} = q_w^{RP} , \; \forall w' \in \{w' | \Psi^{w', RP}_{w} \neq \emptyset \}, 1 \leq l \leq L & \label{eq:reformulation_subprob_RP_conservation}\\
    & \sum_{k \in K_{o^{nl}d^{nl}}}  f_{k}^{\psi, w} \leq Z_n, \; \forall \psi:(n,l,\cdot) \in \Psi^{w, RD}  & \label{eq:reformulation_subprob_route_cap} \\
    & \text{RD-RP route coupling constraint: \quad\quad \eqref{eq:reformulation_route_RD_RP_coupling}} & \nonumber\\
     & \text{RD cross-level flow conservation: \quad\quad \eqref{eq:reformulation_route_cross_level_conservation}} & \nonumber\\
     & \text{Non-negative constraints: \quad\quad \eqref{eq:reformulation_route_nonnegative}} \nonumber
\end{align}
\end{subequations}
where, $f_{0}^{w', RD}$ \DIFdelbegin \DIFdel{denotes the RP }\DIFdelend \DIFaddbegin \DIFadd{denote the RP demand }\DIFaddend served by other \DIFdelbegin \DIFdel{RD. 
}\DIFdelend \DIFaddbegin \DIFadd{RDs, and the coupling of RP demands in different RD-RP groups are captured by the RP quitting option $(n^{w', PT}, 1,1)$ in constraint~\eqref{eq:reformulation_subprob_RP_conservation}, which allows RPs to join other RD-RP groups. Note that, in contrast to \mbox{
\textcite{lo2000reformulating}}\hskip0pt
, our objective function~\eqref{eq:reformulation_subprob_route_obj} is non-convex. Therefore, we propose to adapt the AL solution scheme~}\citep{kanzow2016augmented} \DIFadd{with BPD method~}\citep{bolte2014proximal} \DIFadd{for their convergence guarantees.
}\DIFaddend 

The subproblem formulation~\eqref{eq:reformulation_subprob_route_flow} is similar to the traffic assignment problem with capacity constraints~\citep[e.g.,][]{nie2004models}, except that the \textit{capacity} depends on the platform's decision (as the asymmetric constraint indicated in \cref{sec:Existence}). We \DIFdelbegin \DIFdel{propose to }\DIFdelend apply the augmented Lagrangian method~\citep[][]{kanzow2016augmented} \DIFdelbegin \DIFdel{, such that }\DIFdelend \DIFaddbegin \DIFadd{to decouple }\DIFaddend the interaction between the platform and ridesharing travelers\DIFdelbegin \DIFdel{are decoupled}\DIFdelend . Let $h_n(f, Z) = \sum_{k \in K_{o^{nl}d^{nl}}}  f_{k}^{\psi, w} - Z_n$, the corresponding augmented Lagrangian function \citep[adapted from][]{nie2004models} relaxes capacity constraints~\eqref{eq:reformulation_subprob_route_cap} as follows:
\begin{align}
\label{eq:augmented_lagrangian_problem}
\min \mathcal{L} \quad &  =g(f, \pi) + \frac{1}{2\rho} \sum_{n}\left[\max(0, \tilde{\mu}_n + \rho h_n(f, Z) )^2 -  \tilde{\mu}_n^2 \right] & \\
s.t. \quad & \text{RD and RP demand conservations: \quad\quad \eqref{eq:reformulation_subprob_RD_conservation}-\eqref{eq:reformulation_subprob_RP_conservation}} & \nonumber\\
& \text{RD-RP route coupling constraint: \quad\quad \eqref{eq:reformulation_route_RD_RP_coupling}} & \nonumber\\
& \text{RD cross-level flow conservation: \quad\quad \eqref{eq:reformulation_route_cross_level_conservation}} & \nonumber\\
& \text{Non-negative constraints: \quad\quad \eqref{eq:reformulation_route_nonnegative}} \nonumber
\end{align}
where, $\rho$ and $\tilde{\mu}_n$ are penalty parameters \DIFdelbegin \DIFdel{which }\DIFdelend \DIFaddbegin \DIFadd{that }\DIFaddend are updated at each \DIFaddbegin \DIFadd{outer AL }\DIFaddend iteration.

\DIFdelbegin 

\DIFdelend \subsection{A sequence-bush assignment method}
\label{subsec:sequence_bush}
\noindent
The key challenges in solving the subproblems are twofold: 1) to determine the full path sets $K_{od}$ and 2) to ensure cross-level conservation. Enumerating the full path set is impractical in real-size networks. 
The bush approach~\citep{bar2002origin,dial2006path,nie2010class} proposes a compact link-based representation of the unrestricted path set, which exploits the acyclic property of link flows for equilibrium in the full path set space. The bush approach is suitable for assigning flows at each level (i.e., as individual TAP with OD $(o^{nl}, d^{nl})$), but additional procedure is needed for cross-level conservations.

We proposed to \DIFdelbegin \DIFdel{connect }\DIFdelend \DIFaddbegin \DIFadd{chain }\DIFaddend bushes in sequence for each $n$ (termed as \textit{sequence-bush}) to ensure cross-level conservations. Recall that, the destination of level $l$ is the origin of level $l+1$ (i.e., $d^{nl}=o^{nl+1}$). By connecting bushes according to the sequence of $[(o^{n1}, d^{n1}), (o^{n2}, d^{n2}), ..., (o^{nL}, d^{nL})]$, cross-level conservations are satisfied\DIFdelbegin \DIFdel{and }\DIFdelend \DIFaddbegin \DIFadd{. Moreover, it provides a compact sequence-bush representation for }\DIFaddend the set of full trajectories \DIFdelbegin \DIFdel{for }\DIFdelend \DIFaddbegin \DIFadd{of }\DIFaddend a matching sequence\DIFdelbegin \DIFdel{is represented as a sequence-bush.
}\DIFdelend \DIFaddbegin \DIFadd{.
}

\DIFaddend In the following paragraphs, we first present the overall iterative solution framework for the GEM-mpr, then detail the \textit{bush update} and \textit{Sequence-bush flow pushing} procedures.

\begin{algorithm}[H]
\caption{Sequence-bush solution framework}\label{alg:sequence_bush}
\begin{algorithmic}
\State  \textbf{Step 0}\; \textit{Initialization}
\State  \quad\textit{\textbf{Step 0.1.}}\;- \textit{Bush construction}:\; Construct initial bushes $B$ for all origins $\mathcal{O}$ and destinations $\mathcal{D}$, with respect to DA, with/without-passenger RD, RP and PT link costs
\State  \quad\textit{\textbf{Step 0.2.}}\;-  Set iteration $p=0$, and choose feasible values $(q^{(0)}, Z^{(0)}, f^{(0)}, \tilde{\mu}^{(0)}, \rho^{(0)})$.\\
\State  \textbf{Step 1}\; \textit{Inner loop}
\State  \quad\textit{\textbf{Step 1.1.}}\;- \textit{Modal split update} $(q)$:\; Shift modal demands to the cheapest mode while conserving the total demands

\State  \quad\textit{\textbf{Step 1.2.}}\;- \textit{Matching update} $(z)$:\; Solve the gradient projection problem with respect to RD and RP demand constraints
\State  \quad\textit{\textbf{Step 1.3.}}\;- \textit{Bush update}:\; Update bushes by removing unused links and adding links that have the potentials to reduce OD travel costs
\State  \quad\textit{\textbf{Step 1.4.}}\;- \textit{Sequence-bush flow pushing} $(f)$:\; For each RD-RP group, flow shifting for the sequence-bush
\State  \quad\textit{\textbf{Step 1.5.}}\;- \textit{Inner convergence check}:\; If converged, go to \textbf{Step 2}; otherwise, go to \textit{\textbf{Step 1.1}}. \\
\State  \textbf{Step 2}\; \textit{Update augmented Lagrangian parameters} $(\tilde{\mu}, \rho)$\\
\State  \textbf{Step 3}\; \textit{Overall convergence check}:\; If converged, stop; otherwise, go to \textbf{Step 1}.
\end{algorithmic}
\end{algorithm}

At initialization, bushes are constructed as minimum-cost path trees rooted at the origins and destinations ($\mathcal{O} \cup \mathcal{D}$).  
The total demand at each OD is evenly split into 4 modes ($q_w^{m, (0)}$), in which RD and RP demands are matched (\DIFdelbegin \DIFdel{$z^{(0)}$}\DIFdelend \DIFaddbegin \DIFadd{$Z^{(0)}$}\DIFaddend ). 
The augmented Lagrangian parameters $\tilde{\mu}^{(0)}$ and  $\rho^{(0)}$ are initialized as 0 (as suggested in \textcite{kanzow2016augmented}). All the modal demands are assigned into the minimum-cost path trees (\DIFdelbegin \DIFdel{$f$}\DIFdelend \DIFaddbegin \DIFadd{$f^{(0)}$}\DIFaddend ), in which RD and RP are initialized respectively as \textit{leaving for DA} and \textit{PT} to ensure feasibility.  

For the given $\tilde{\mu}^{(p)}$ and  $\rho^{(p)}$, the inner loop (Steps 1.1-1.5) is performed until convergence of the inner problem is reached. Based on inner-loop solutions, the augmented Lagrangian parameters are updated. If overall convergence is reached, we stop the algorithm, otherwise, start the inner loop with the updated $\tilde{\mu}^{(p+1)}$ and  $\rho^{(p+1)}$. In the following paragraphs, we explain the inner-loop steps.

For the mode choice problem (\textit{\textbf{Step 1.1}}), let $\underbar{m}$ denote the cheapest mode \DIFdelbegin \DIFdel{for each OD }\DIFdelend \DIFaddbegin \DIFadd{with cost $C_w^{\underbar{m}, (p)}$ for OD $w$, }\DIFaddend and $\theta_1^{p}$ \DIFdelbegin \DIFdel{denotes }\DIFdelend \DIFaddbegin \DIFadd{denote }\DIFaddend the step size, the \DIFdelbegin \DIFdel{modal split is updated according to the }\DIFdelend \DIFaddbegin \DIFadd{block proximal descent method~}\citep{bolte2014proximal} \DIFadd{is integrated to update the modal split for each OD $w$ as follows:
}\begin{subequations}
\begin{align}
\DIFadd{\label{eq:mode_choicee_BPD}
\bm{q}_w^{(p+1)} = \arg \min_{\bm{q}_w \in \iota_q} \sum_m \left[q_w^m \cdot \left(C_w^{m, (p)} - C_w^{\underbar{m}, (p)}\right) + \underbrace{\frac{1}{2\theta_1^{p}} ||q_w^m - q_w^{m, (p)}||^2}_{\text{Proximal regularization}}\right]
}\end{align}
\DIFadd{where, $\iota_q$ denote the feasible modal splits. By setting the subgradient of the right-hand side in Eq.~\eqref{eq:mode_choicee_BPD} as 0, the }\DIFaddend following closed-form formulas \DIFdelbegin \DIFdel{at each inner iteration:
}
\DIFdelend \DIFaddbegin \DIFadd{can be derived:
}\DIFaddend \begin{align}
\Delta q_w^m =& \min \left(q_w^{m, (p)}, \theta_1^{p}(C_w^{m, (p)} - C_w^{\underbar{m}, (p)})\right) &\\
q_w^{m,(p+1)} =& q_w^{m,(p)} - \Delta q_w^m&\\
q_w^{\underbar{m},(p+1)} =& q_w^{\underbar{m},(p)} + \sum_m\Delta q_w^m&
\end{align}
\end{subequations}
\DIFaddbegin 

\DIFaddend Given the updated modal split, the set of feasible matching is denoted as \DIFdelbegin \DIFdel{$\iota_z(q_w^{RD,(p+1)},q_w^{RP,(p+1)})$}\DIFdelend \DIFaddbegin \DIFadd{$\iota_Z(q_w^{RD,(p+1)},q_w^{RP,(p+1)})$}\DIFaddend . To update the matching (\textit{\textbf{Step 1.2}}), we solve the following \DIFdelbegin \DIFdel{gradient projection }\DIFdelend \DIFaddbegin \DIFadd{BPD }\DIFaddend problem with step size $\theta_2^{(p)}$ \DIFdelbegin \DIFdel{:
}\DIFdelend \DIFaddbegin \DIFadd{for the platform's maximization objective:
}\DIFaddend \begin{align}
\DIFdelbegin \DIFdel{z}\DIFdelend \DIFaddbegin \bm{Z}\DIFaddend ^{(p+1)} = \arg \DIFdelbegin \DIFdel{\min_{z \in \iota_z} }
\DIFdel{z^{(p+1)} - z}\DIFdelend \DIFaddbegin \DIFadd{\max_{\bm{Z} \in \iota_Z} \sum_n }\left[\DIFadd{R_n}\DIFaddend ^{(p)} \DIFaddbegin \DIFadd{\cdot Z_n }\DIFaddend -\DIFdelbegin \DIFdel{\theta_2}\DIFdelend \DIFaddbegin \DIFadd{\frac{1}{2\theta_2^{p}} ||Z_n- Z_n}\DIFaddend ^{(p)}\DIFdelbegin \DIFdel{R^{(p)}}
\DIFdelend \DIFaddbegin \DIFadd{||}\DIFaddend ^2\DIFaddbegin \right]
\DIFaddend \end{align}  
\DIFdelbegin \DIFdel{For the given $q_w^{m,(p+1)}$ and $z^{(p+1)}$}\DIFdelend \DIFaddbegin 

\DIFadd{In the following paragraphs, we detail the algorithm steps for solving the network assignment subproblem. Given modal splits $\bm{q}_w^{(p+1)}$ and matching $\bm{Z}^{(p+1)}$}\DIFaddend , the \textit{\textbf{Step 1.3} bush update} procedure \citep[adapted from][]{bar2002origin,dial2006path,nie2010class} is summarized as follows:
\begin{algorithm}[H]
\begin{algorithmic}
\Procedure{Bush update}{}
\State \textbf{\textit{Step 1.3.1}}\;- \textit{Unused link removal}:\;

Remove all links with zero flow in each bush (indexed by bush root $o$), while maintaining the connectivity from the bush origin to all other nodes.

\State \textbf{\textit{Step 1.3.2}}\;- \textit{Topology ordering and label setting}:\; 
\begin{itemize}
\item Relying on the acyclicity of \DIFdelbegin \DIFdel{the }\DIFdelend bush $o$, the topology of a bush can be obtained by running depth-first searches originating from nodes in the bush. Let $\mathcal{T}_i^o$ \DIFdelbegin \DIFdel{denotes }\DIFdelend \DIFaddbegin \DIFadd{denote }\DIFaddend the index of node $i$ in the topology order.
\item During traversing the bush $o$, the minimum-cost and maximum-cost labels from the bush origin to node $i$: $\tilde{\pi}_i^o$ and $\tilde{U}_i^{n,o}$ (for each incident matching sequence $n$) can be set simultaneously, with the corresponding cheapest route $\underbar{k}_i^o$ and costliest \DIFaddbegin \DIFadd{route }\DIFaddend $\bar{k}_i^{n,o}$ recorded. 
\end{itemize}

\State \textbf{\textit{Step 1.3.3}}\;- \textit{Bush expansion}:\; 

The bush is expanded by adding links with the potential to reduce OD cost (i.e., result in smaller $\tilde{\pi}_i^o$), while ensuring the acyclicity of the bush.\\
\begin{itemize}
\item Let $P_1$ denote the set of links with reduced cost, namely, $P_1 = \{ij|\tilde{\pi}_i^o + \tilde{c}_{ij}(x_{ij}) < \tilde{\pi}_j^o  \}$, where $x_{ij}$ is the link flow prior to the bush update. 
\item Let $P_2$ \DIFdelbegin \DIFdel{denotes }\DIFdelend \DIFaddbegin \DIFadd{denote }\DIFaddend the set of links guarantee acyclicity, $P_2 = \{ij|\mathcal{T}_i^o<\mathcal{T}_j^o \}$
\item The set of links added to the bush $P = P_1 \cap P_2$. In case $P_1 \cap P_2 = \emptyset$ and $P_2 \neq \emptyset$, $P=P_2$ to avoid breakdown.
\end{itemize}

\EndProcedure
\end{algorithmic}
\end{algorithm}

Note that, \textbf{\textit{Step 1.3.2}} implies that, in each bush $o$, path enumeration and storage can be avoided if flows are only shifted from the maximum-cost route $\bar{k}_i^{n,o}$ to the minimum-cost route $\underbar{k}_i^o$ \citep{dial2006path}. The acyclicity property of a bush allows max-cost route to be found efficiently using depth-first search, which would be difficult in general networks. 
We also propose to share a bush for matching sequences with the same bush origin and only update $\tilde{U}_i^{n,o}$ for each matching sequence, such that redundant bushes are avoided. 

To also avoid enumeration and storage of full trajectories, we propose to \DIFdelbegin \DIFdel{consider only solving }\DIFdelend \DIFaddbegin \DIFadd{solve only }\DIFaddend for the maximum-cost and minimum-cost \textit{sequence routes} \DIFaddbegin \DIFadd{at each inner iteration}\DIFaddend . Let $\bar{\kappa}_n$ and $\underbar{\kappa}_n$ denote the maximum-cost and minimum-cost \textit{sequence routes} for $n$, \DIFaddbegin \DIFadd{the }\DIFaddend $\bar{\kappa}_n$ is constructed by connecting maximum-cost routes $\bar{k}_{d^{nl}}^{n,o^{nl}}$ in sequence (for $l=1,...,L$), whereas $\underbar{\kappa}_n$ similarly connects minimum-cost routes $\underbar{k}_{d^{nl}}^{o^{nl}}$. \DIFdelbegin \DIFdel{Furthermore, we consider for each RD-RP group $w$ a flow shifting problem }\DIFdelend \DIFaddbegin \DIFadd{Consequently, shifting flows between }\textit{\DIFadd{sequence routes}} \DIFadd{guarantees cross-level flow conservation.
}

\DIFadd{Based on }\textit{\DIFadd{sequence routes}}\DIFadd{, the sequence-bush flow pushing problem shifts flows }\DIFaddend from $\bar{\kappa}_{\bar{n}}$ of the costliest matching sequence $\bar{n}$ to $\underbar{\kappa}_{\underbar{n}}$ of the cheapest matching sequence $\underbar{n}$ \DIFdelbegin \DIFdel{, such that RD cross-level flow conservation is implicitly satisfied}\DIFdelend \DIFaddbegin \DIFadd{for each RD-RP group}\DIFaddend . The \textit{\textbf{Step 1.4} sequence-bush flow pushing} procedure for RD-RP group $w$ is detailed as follows:

\begin{algorithm}[H]
\begin{algorithmic}
\Procedure{Sequence-bush flow pushing}{}
\State \textbf{\textit{Step 1.4.1}}\;- \textit{Label updating}:\;
\begin{itemize}
\item Perform \textit{label setting} \DIFaddbegin \DIFadd{$(\pi^{(p+1)})$ }\DIFaddend in \textbf{\textit{Step 1.3.2}} for all relevant bushes in case travel costs have been updated
\item For each matching sequence $n$ associated with RD $w$, calculate the minimum cost $\underbar{C}_n^w$ and maximum cost $\bar{C}_n^w$ according to the bush labels, record also the minimum-cost and maximum-cost RD \DIFdelbegin \DIFdel{route sequences }\DIFdelend \DIFaddbegin \textit{\DIFadd{sequences routes}} \DIFaddend $\underbar{\kappa}_n$ and $\bar{\kappa}_n$. 
\end{itemize}

\State \textbf{\textit{Step 1.4.2}}\;- \textit{Sequence-bush flow shift}:\;
\begin{itemize}
\item Let $\underbar{n}$ \DIFdelbegin \DIFdel{denotes }\DIFdelend \DIFaddbegin \DIFadd{denote }\DIFaddend the cheapest matching sequence, $\underbar{n}=\arg\min_n \{ \underbar{C}_n^w \}$, and $\bar{n}$ \DIFdelbegin \DIFdel{denotes }\DIFdelend \DIFaddbegin \DIFadd{denote }\DIFaddend the costliest matching sequence, $\bar{n}=\arg\max_n \{ \bar{C}_n^w \}$. 
\item For simplicity, let $\{f_{\bar{k}} \}$ and $\{f_{\underbar{k}} \}$ denote the set of RD and RP route flows \DIFaddbegin \DIFadd{(which are equal due to coupling constraints) }\DIFaddend on $\bar{k} \in \bar{\kappa}_{\bar{n}}$ and $\underbar{k} \in \underbar{\kappa}_{\underbar{n}}$, respectively. \DIFaddbegin \DIFadd{RD and the matched RP flows are updated as follows:
}\DIFaddend \end{itemize}
\begin{subequations}
\DIFaddbegin \label{eq:gradient_projection_problem}
\DIFaddend \begin{align}
\DIFdelbegin 
\DIFdelend \Delta f_{\bar{k}}^{(p+1)} = \quad& \arg \min\DIFaddbegin \DIFadd{_{f_{\bar{k}}} }\left[\DIFadd{\mathcal{L} }\DIFaddend \left(f_{\bar{k}}\DIFaddbegin \DIFadd{, }\bm{f}\DIFadd{_{k \neq {\bar{k}}}^{(p)}, \pi}\DIFaddend ^{(p+1)}\DIFaddbegin \DIFadd{, Z^{(p+1)}}\right) \DIFadd{+ \frac{1}{2\theta^{(p)}} ||f_{\bar{k}} }\DIFaddend - f_{\bar{k}}^{(p)} \DIFdelbegin \DIFdel{+ \theta^{(p)} \nabla_{f_{\bar{k}}} \mathcal{L} }
\DIFdelend \DIFaddbegin \DIFadd{||}\DIFaddend ^2 \DIFaddbegin \right] \DIFaddend - f_{\bar{k}}^{(p)} , \; \forall \bar{k} \in \bar{\kappa}_{\bar{n}}  & \\
s.t., &\quad \text{constraints } \DIFdelbegin \DIFdel{\eqref{eq:reformulation_subprob_RD_conservation},\eqref{eq:reformulation_subprob_RP_conservation},}\DIFdelend \eqref{eq:reformulation_route_RD_RP_coupling}, \eqref{eq:reformulation_route_nonnegative} & \nonumber \\
\Delta f^{(p+1)} = \quad& \min \{\Delta f_{\bar{k}}^{(p+1)} \} & \\
f_{\bar{k}}^{(p+1)} = \quad& f_{\bar{k}}^{(p)} - \Delta f^{(p+1)}, \forall \bar{k} \in \bar{\kappa}_{\bar{n}} \\
f_{\underbar{k}}^{(p+1)} = \quad& f_{\underbar{k}}^{(p)} + \Delta f^{(p+1)}, \forall \underbar{k} \in \underbar{\kappa}_{\underbar{n}}
\end{align}
\end{subequations}
\EndProcedure
\end{algorithmic}
\end{algorithm}

\DIFdelbegin 

\DIFdelend We apply the self-regulated averaging method proposed by \textcite{liu2009method} for updating step sizes $\theta^{(p)}, \theta_1^{(p)}, \theta_2^{(p)}$, in which step size is regulated if \DIFaddbegin \DIFadd{the }\DIFaddend solution at current iteration greatly diverges, otherwise a relatively larger step size is applied for exploration.

The inner convergence is checked against two conditions: 1) normalized gap in modal costs $G_M \leq \epsilon_M$, and 2) normalized gap in the network model $G_N \leq \epsilon_N$. We consider that when these two conditions are met, the mode choice and network problems reach equilibrium, which implies the ridesharing matching problem is also at equilibrium due to the interactions between three components. Specifically, the normalized gaps are specified as follows:
\begin{align}
G_M = \frac{\sum_w \sum_{m} q_w^m \left(C_w^{m} - C_w^{\underbar{m}} \right)}{\sum_w q_w}, \; G_N = \frac{\sum_{w\in\mathcal{W}} 
    \sum_{\psi:(n,l,\cdot) \in \Psi^w} \left(c_{\bar{k}}^{nl} -c_{\underbar{k}}^{\underbar{n}l} \right) \sum_{k \in K_{o^{nl}d^{nl}}} f_{k}^{\psi, w} }{\sum_w q_w}
\end{align}  
where, $c_{\bar{k}}^{nl}$ is the maximum route cost of $n$ at $l$, and $c_{\underbar{k}}^{\underbar{n}l}$ is the minimum route cost of cheapest sequence $\underbar{n}$ at $l$.
We adopt the augmented Lagrangian parameter update scheme proposed by \textcite{nie2004models} as follows:
\begin{subequations}
\begin{align}
\tilde{\mu}_n^{(p+1)} &= \max \left[0, \tilde{\mu}_n^{(p)} + \rho^{(p)}h_n(f^{(p)}, Z^{(p)}) \right]  \\
\rho^{(p+1)} &= \begin{cases}
\sigma_1 \rho^{(p)}, \; \text{if } \|\mathbf{f}^{(p)} - \mathbf{Z}^{(p)} \| \geq \sigma_2 \|\mathbf{f}^{(p-1)} - \mathbf{Z}^{(p-1)} \|\\
\rho^{(p)}, \; \text{Otherwise}
\end{cases}
\end{align}
\end{subequations}
where, $\sigma_1$ and $\sigma_2$ are additional parameters for the update scheme. As suggested in \textcite{kanzow2016augmented}, after first outer iteration, $\tilde{\mu}_n$ is initialized by solving $\nabla \mathcal{L}=0$ for RD, and $\rho$ is set as the ratio between $g(f, \pi)$ and $\sum_n h_n(f, Z)$. 

We consider the overall problem is converged if constraint violation is relatively small~\citep[similar to][]{nie2004models}. Let $\mathcal{A}$ being the augmented Lagrangian terms in $\mathcal{L}$ (Eq.\ref{eq:augmented_lagrangian_problem}), namely, $\mathcal{A}(f, Z, \tilde{\mu}, \rho) = \mathcal{L} - g(f, \pi)$, the outer convergence condition is defined as:
\begin{align}
\label{eq:outer_convergence}
\frac{\mathcal{A}(f, Z, \tilde{\mu}, \rho)}{\mathcal{A}(f, Z, \tilde{\mu}, \rho) + g(f, \pi)} \leq \epsilon_3
\end{align} 

\DIFdelbegin 
\DIFdelend \DIFaddbegin \section[Illustrative example]{\DIFaddend Illustrative example\DIFaddbegin \footnote{\DIFadd{Codes available at: https://github.com/andyYaoR/SequenceBush}}\DIFaddend }
\label{sec:illustration}
In \DIFdelbegin \DIFdel{\mbox{
\cref{fig:illustrative_example}}\hskip0pt
}\DIFdelend \DIFaddbegin \DIFadd{this section}\DIFaddend , we illustrate \DIFdelbegin \DIFdel{several }\DIFdelend \DIFaddbegin \DIFadd{selected }\DIFaddend features of our model and the solution algorithm \DIFdelbegin \DIFdel{for 1 RD and 2 RPs.
}\DIFdelend \DIFaddbegin \DIFadd{in a simple network (\mbox{
\cref{fig:illustrative_network}}\hskip0pt
), and compare the runtime performances between the proposed algorithm and the GAMS/PATH solver~}\citep{ferris2000complementarity}\DIFadd{. We first present the example setup and equilibrium solution, and discuss the contributions of the paper in the following subsections.
}

\DIFaddend \begin{figure}[H]
    \centering
    \DIFdelbeginFL 
{
\DIFdelFL{Illustrative example (a) matching sequence with 2 RPs (b) Corresponding link flows in the hyper-network (with RD-RP coupling and cross-level flow conservation)}}
\DIFdelendFL \DIFaddbeginFL \includegraphics[width=0.95\textwidth]{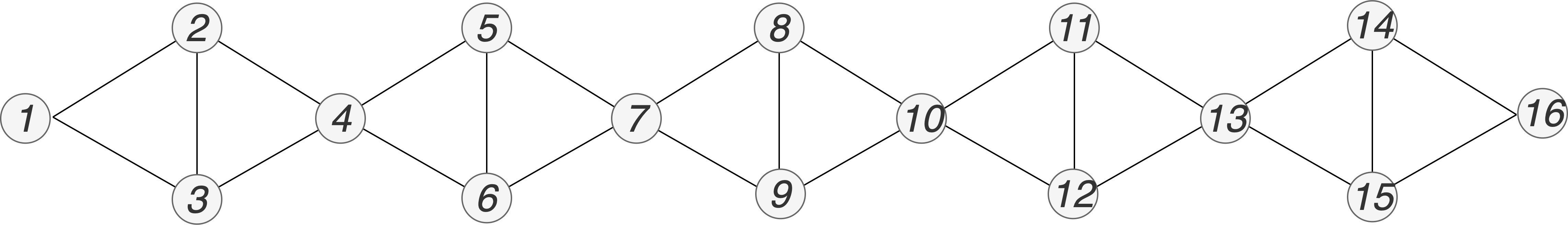}
    \caption{\DIFaddFL{Physical network for the illustrative example}}
    \label{fig:illustrative_network}
\DIFaddendFL \end{figure}
\DIFdelbegin \DIFdel{As illustrated in \mbox{
\cref{fig:illustrative_example}}\hskip0pt
(a ), we consider a matching sequence $n$ for RD $v$ with vehicle capacity of 2 taking 2 RPs, $u_2$ and $u_3$, who have different origins and destinations. The matching sequence ensures the pickup and drop-off of multiple RPs at different locations, as well as the vehicle capacity constraints. In this example , RP $u_2$ is first picked up at his/her origin, and travel with RD $v$ to pickup RP $u_3$ at a different location (as indicated in $v$ and $u_2$ traversing network links at level $l=2$).
After picking up $u_3$, }\DIFdelend \DIFaddbegin 

\DIFadd{We assume a single OD pair for }\DIFaddend RD $v$\DIFdelbegin \DIFdel{reaches maximum occupancy }\DIFdelend \DIFaddbegin \DIFadd{, with vehicle capacity }\DIFaddend of 2 \DIFdelbegin \DIFdel{at $l=3$, during which RD }\DIFdelend \DIFaddbegin \DIFadd{for each }\DIFaddend $v$\DIFdelbegin \DIFdel{and RPs $u_2, u_3$ travel together to the drop-off location of }\DIFdelend \DIFaddbegin \DIFadd{. There are two OD pairs for RPs }\DIFaddend $u_2$ \DIFdelbegin \DIFdel{. Following the drop-off of $u_3$ at level $l=4$, RD $v$ travels to his/her own destination at $l = 5$. This example illustrates that matching sequence implicitly accounts for vehicle capacity constraints and multi-passenger ridesharing pickup and drop-off at different locations. Specifically, routing between two consecutive tasks $l-1$ }\DIFdelend and \DIFdelbegin \DIFdel{$l$ is represented as }\textit{\DIFdel{level}} 
\DIFdel{$l$ in the hyper-network, in which the pickup/drop-off location of previous task $s(n, l-1)$ and current task $s(n,l)$ are considered as the origin and the destination in $l$, respectively. 
}

\DIFdel{The matching sequence $n$ is an abstraction of the actual routing (paths) inthe road network, which are explicitly represented in the corresponding hyper-network \mbox{
\cref{fig:illustrative_example}}\hskip0pt
(b). For simplicity, we assume in this example each level consists of 5 links with equilibrium link costs for each class specified in \mbox{
\cref{tab:illustrative_example_link_costs}}\hskip0pt
, and the corresponding paths in \mbox{
\cref{tab:illustrative_example_path_costs}}\hskip0pt
:
}\DIFdelend \DIFaddbegin \DIFadd{$u_3$. Assume that there are 12 matching sequences and additional options to quit ridesharing for RD $v$ and RPs $u_2, u_3$. The example setup is summarized in~\mbox{
\cref{tab:illustrative_example_setup}}\hskip0pt
:
}

\begin{table}[H]
\caption{Illustrative example setup}
\label{tab:illustrative_example_setup}
\centering
\setlength\tabcolsep{10pt}
\begin{tabular}{lrcll}
  \toprule
   \textbf{Physical network}
   & $(\mathcal{N, E})$ 
   & 
   & \textbf{Link costs}
   & $(c_{ij}^m)$ \\
Number of nodes & 16                        &    & DA          &      $t_{ij} + 20$         \\
Number of links (bidirectional) & 25        &    & RD - Without RP         &  $t_{ij} + 20$     \\
 &                                     &    & RD - With RP          &  $t_{ij} + 10$             \\
\textbf{Modal demands} & $(q_w^m)$                                      &    & RP         & $t_{ij} + 10$              \\
\textit{RD} &                                      &    & PT         & $t_{ij} + 15$               \\
$v:(1, 16)$ &     $40,000$                                 &    &         &                \\
\textit{RP} &                                      &    & \textbf{Travel time function}         & $(t_{ij})$               \\
$u_2:(4, 10)$ &     $20,000$                                 &    & \multicolumn{2}{l}{$t_{ij} = 5 \cdot \left(1 + 0.15 \cdot (\frac{x_{DA} + x_{RD}}{10,000})^{4}\right)$}              \\
$u_3:(7, 13)$ &     $20,000$                                 &    & \textbf{Link length} & $l_{ij} = 5$      \\
\bottomrule   
\end{tabular}
\end{table}

\DIFaddbegin \DIFadd{The solution of this example is reported in \mbox{
\cref{tab:illustrative_stable_results}}\hskip0pt
, which presents the 5 cheapest matching sequences in terms of RD generalized costs.
The corresponding equilibrium ridesharing link flows and aggregated vehicular flows are in \mbox{
\cref{fig:illustrative_example}}\hskip0pt
(b) and \mbox{
\cref{fig:illustrative_example}}\hskip0pt
(c), respectively. 
}

\begin{table}[H]
\caption{Matching sequence equilibrium results (Top 5 cheapest)}
\label{tab:illustrative_stable_results}
\centering
\setlength\tabcolsep{6pt}
\begin{tabular}{c|c|cccccc}
  \toprule
$n$   & Matching sequence                                      & Generalized cost & RD cost$^{\ast}$ & RP cost$^{\ast}$ & RD flow$^{\ast}$ & RP flow$^{\ast}$ \\
   &                                                       &     (RD Cost + Multipliers)    &       $v$    & $u_2 , u_3$  &     $v$        &   $u_2, u_3$       \\
   \midrule
1& $(1, 4, 7, 10, 13, 16)$ & 370          & 310       &   108          &    20,000 & 20,000         \\
2& RD/RP Quit & 370       & 370       & 128         & 20,000 & 0         \\
3& $(1, 4, 10, 16)$ & 390       &   330      &    108      & 0 & 0            \\
4& $(1, 7, 13, 16)$ & 390       &   330      &    108      & 0 & 0           \\
5& $(1, 4, 4, 10, 10, 16)$ & 390       &   330      &    108      & 0 & 0             \\
\midrule
        \multicolumn{2}{c}{Total} &   & & & 40,000 & 20,000  \\
\bottomrule   
\end{tabular}
\begin{tablenotes}
  \item *: \textit{Equilibrium solutions are the same for the sequence-bush algorithm and GAMS.}
\end{tablenotes}
\end{table}

\DIFdelbegin \DIFdel{In this example}\DIFdelend \DIFaddbegin \DIFadd{At equilibrium, all active physical links are highlighted in \mbox{
\cref{fig:illustrative_example}}\hskip0pt
(c), where each active link has vehicular flow of 20}\DIFaddend ,\DIFdelbegin \DIFdel{RP }\DIFdelend \DIFaddbegin \DIFadd{000 and travel time of 17. In this specific example, the equilibrium solution splits the flows between RD and DA trips (due to RD quitting), as shown in \mbox{
\cref{fig:illustrative_example}}\hskip0pt
(b).
}

\begin{figure}[H]
    \centering
    \includegraphics[width=0.95\textwidth]{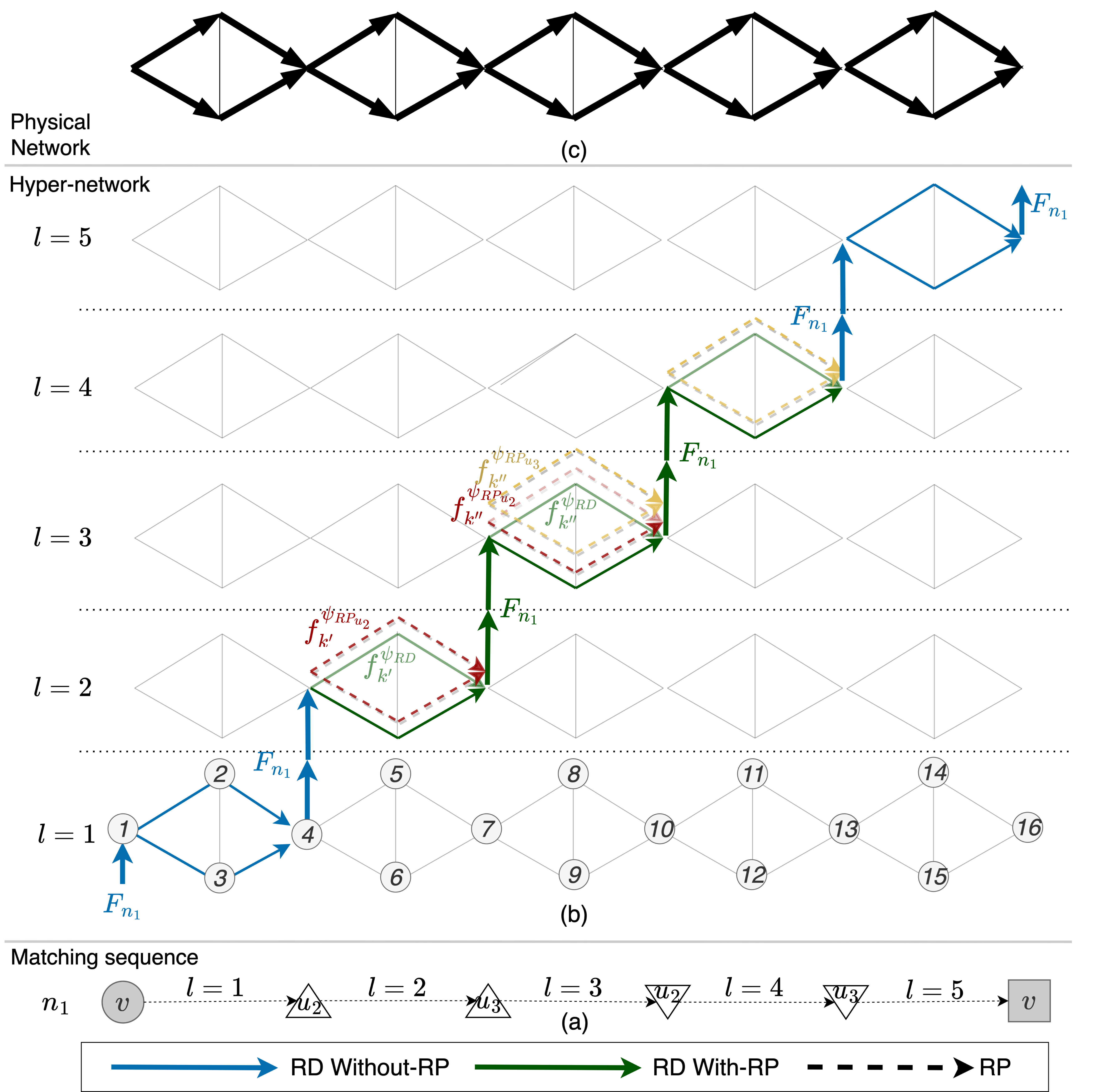}
    \caption{\DIFaddFL{Illustrative example (a) matching sequence $n_1$; (b) Hyper-network ridesharing flows; (c) Physical network flows}}
    \label{fig:illustrative_example}
\end{figure}

\subsection{\DIFadd{Model feature: Matching sequence and hyper-network}}

\DIFadd{One key feature of the proposed model is the integration of matching sequence and hyper-network for multi-passenger ridesharing for a link-based formulation. The matching sequences handle the multiple RP pickup/drop-off and vehicle capacity constraints, while the hyper-network handles the routing.
}

\DIFadd{As illustrated in \mbox{
\cref{fig:illustrative_example}}\hskip0pt
(a) for the optimal matching sequence $n_1:(1, 4, 7, 10, 13, 16)$, the matching sequence ensures both }\DIFaddend $u_2$ \DIFdelbegin \DIFdel{makes detours with RD $v$ to pickup RP $u_3$ (termed as }\textit{\DIFdel{detour}} 
\DIFdel{trip for $u_2$), whereas }\DIFdelend \DIFaddbegin \DIFadd{and }\DIFaddend $u_3$ \DIFdelbegin \DIFdel{make detours to }\DIFdelend \DIFaddbegin \DIFadd{are picked up and dropped off at their origins (node 4 and 7) and destinations (node 10 and 13). Note that this example also ensures RP need not make any transfer, in contrast to previous studies for multi-passenger ridesharing. 
}

\DIFadd{Moreover, the maximum occupancy of 2 is reached at $l=3$ after picking up $u_3$, and the vehicle capacity constraint is satisfied for the whole matching sequence.
}

\DIFadd{The matching sequence is an abstraction of the actual routing (paths) in the road network. The explicit representation of the routing of matching sequence $n_1$ is illustrated in a hyper-network in \mbox{
\cref{fig:illustrative_example}}\hskip0pt
(b). In this example, the hyper-network consists of 5 levels, which corresponds to the 5 routing segments in $n_1$. At each level $l$, RD and RP can choose any path from the pickup/}\DIFaddend drop-off \DIFdelbegin \DIFdel{$u_2$ (termed as }\textit{\DIFdel{drop-off}} 
\DIFdel{trip for }\DIFdelend \DIFaddbegin \DIFadd{location of previous task $s(n, l-1)$ to current task $s(n,l)$.  For example, at $l=2$, the origin is node 4 (pickup }\DIFaddend $u_2$\DIFdelbegin \DIFdel{and }\textit{\DIFdel{detour}} 
\DIFdel{trip for }\DIFdelend \DIFaddbegin \DIFadd{) and the destination is node 7 (pickup }\DIFaddend $u_3$)\DIFdelbegin \DIFdel{. Consequently, RP }\DIFdelend \DIFaddbegin \DIFadd{, where 2 routes (highlighted in  \mbox{
\cref{fig:illustrative_example}}\hskip0pt
(b)) are chosen. After picking up $u_3$, RD $v$ (and on-board RPs $u_2, u_3$) traverse the virtual link $(7^{l=2}, 7^{l=3})$ for dropping off }\DIFaddend $u_2$ \DIFdelbegin \DIFdel{'s cost for taking matching sequence $n$  equals }\DIFdelend \DIFaddbegin \DIFadd{at $l=3$, such that the matching sequence $n_1$ is retained.
}

\DIFadd{Our formulation integrates matching sequence with hyper-network, which endogenously determines the matching sequence costs by traversing RD/RP's trajectory in the hyper-network and summing up the link costs. For example, the cost of matching sequence $n_1$ for RP $u_2$ is equal }\DIFaddend to the sum of detour \DIFdelbegin \DIFdel{(}\DIFdelend \DIFaddbegin \DIFadd{trip cost for picking up peer RP $u_3$ at }\DIFaddend $l=2$\DIFdelbegin \DIFdel{) and drop-off (}\DIFdelend \DIFaddbegin \DIFadd{, and being dropped off at }\DIFaddend $l=3$\DIFdelbegin \DIFdel{) trip costs ($28 = 14 + 14$). 
}\DIFdelend \DIFaddbegin \DIFadd{, which is $\underbrace{2\cdot (17 + 10)}_{l=2} + \underbrace{2\cdot (17 + 10)}_{l=3} = 108$ . 
}

\DIFaddend In case RPs are not on-board at the same time, their matching sequence costs correspond only to their individual \DIFdelbegin \DIFdel{drop-off }\DIFdelend trips, which could be cheaper for RPs but \DIFdelbegin \DIFdel{at the cost of RD making }\DIFdelend \DIFaddbegin \DIFadd{might require RD to make }\DIFaddend more detours.
\DIFdelbegin \DIFdel{This trade-off is captured by the stable matching constraint inour model}\DIFdelend \DIFaddbegin 

\subsection{\DIFadd{Model feature: Stable matching}}
\DIFadd{In principle, RD and RP can choose the matching sequence with their minimum costs. However, this might not lead to a stable matching, as RD and RP can have conflicting objectives. Our formulation considers this issue by stating that a matching is stable if no RD or RP can reduce his/her cost by unilaterally switching to another matching sequence. 
}

\DIFadd{As shown in~\mbox{
\cref{tab:illustrative_stable_results}}\hskip0pt
, RPs $u_2, u_3$ are matched with RD $v$ on the matching sequence $n_1$ with minimum cost of $108$; while half of the RDs cannot be matched (that is, they quit RD and become DA) with higher cost ($370$ instead of $310$). One way to interpret this result is through demand constraints for the RP: modal demands for RP $u_2$ and $u_3$ are only $20,000$, so there is no remaining RP for the unmatched RD. Another interpretation relies on the stability condition: unmatched RD cannot choose another matching sequence, because no RP, who are currently matched, is willing to deviate from their matching $n_1$. As a result, the matching is stable, which is also verified by the same RD generalized cost for the matching sequences with positive flows $n_1$ and $n_2$ in~\mbox{
\cref{tab:illustrative_stable_results}}\hskip0pt
}\DIFaddend .

\DIFdelbegin \DIFdel{This example also shows how the proposed model could be applied for evaluating the effects }\DIFdelend \DIFaddbegin \DIFadd{Apart from capturing the stability between RD and RP, the stable-matching feature also allows evaluating the long-term performance }\DIFaddend of platform operational strategies (e.g., matching objective, pricing scheme)\DIFdelbegin \DIFdel{on ridesharing travelers as well as on its long-term performance. Suppose }\DIFdelend \DIFaddbegin \DIFadd{, under more realistic behavioral assumptions. For example, if }\DIFaddend the platform takes a commission fee based on trip length, the longest route \DIFdelbegin \DIFdel{, $k^3$, }\DIFdelend should maximize the platform's revenue. \DIFdelbegin \DIFdel{From }\DIFdelend \DIFaddbegin \DIFadd{As shown in \mbox{
\cref{tab:illustrative_example_path_costs}}\hskip0pt
, from }\DIFaddend RD and RP's perspectives, both $k^1$ and $k^2$ are obviously cheaper than \DIFaddbegin \DIFadd{the longest route }\DIFaddend $k^3$. If RD and RP are only allowed to follow the platform's optimal route $k^3$, they will leave ridesharing \DIFaddbegin \DIFadd{in a stable matching, }\DIFaddend and the actual platform revenue is zero. \DIFdelbegin \DIFdel{If }\DIFdelend \DIFaddbegin \DIFadd{On the other hand, if }\DIFaddend ridesharing travelers are allowed to deviate from $k^3$, \DIFdelbegin \DIFdel{the platform, however, }\DIFdelend \DIFaddbegin \DIFadd{these matching }\DIFaddend might still be profitable \DIFaddbegin \DIFadd{for the platform}\DIFaddend . 

\DIFdelbegin \DIFdel{We now illustrate the proposed sequence-flow assignment algorithm. 
With the notion of sequence flows $F_n$ and path flows $f_k^{\psi, RD (RP)}$, the ridesharing constraints 1) RP transfer avoidance and stable matching, and 2) RD-RP coupling are illustrated in \mbox{
\cref{fig:illustrative_example}}\hskip0pt
(}\DIFdelend \DIFaddbegin \begin{table}[H]
\caption{\DIFaddFL{Illustration of path costs.}}
\label{tab:illustrative_example_path_costs}
\centering
\setlength\tabcolsep{10pt}
\begin{tabular}{l|lll}
  \toprule
  \multirow{2}{*}{\textbf{\makecell{Path}}}
   & \multirow{2}{*}{\textbf{\makecell{Length}}} 
   & \multicolumn{2}{l}{\textbf{{\makecell[l]{Path costs}}}} \\
    &  & \DIFaddFL{DA, RD (Without RP) }& \DIFaddFL{RP, RD (with RP) }\\
   \midrule
\DIFaddFL{$k^1:(1,2,4)$ }& \DIFaddFL{$10=2\cdot 5$        }& \DIFaddFL{$74=2\cdot\underbrace{(17+20)}_{c_{ij}^{DA}}$        }&   \DIFaddFL{$54=2\cdot(17+10)$                      }\\
\DIFaddFL{$k^2:(1,3,4)$ }& \DIFaddFL{10        }& \DIFaddFL{74      }& \DIFaddFL{54           }\\
\DIFaddFL{$k^3:(1,2,3,4)$ }& \DIFaddFL{15        }&    \DIFaddFL{111     }& \DIFaddFL{81                }\\
\bottomrule   
\end{tabular}
\end{table}

\subsection{\DIFadd{Solution algorithm: Sequence-bush representation}}

\DIFadd{The first two subsections illustrate the contributions related to the link-based formulation of the GEM-mpr problem. This subsection exemplifies the contributions related to the bush-based solution algorithm. 
The bush-based algorithm is efficient for large-scale network problems~}\citep{bar2002origin,nie2004models}\DIFadd{.
Our solution algorithm adapts the bush-based method by suitable chaining the bushes as indicated in this example. 
}

\begin{figure}[H]
    \centering
    \includegraphics[width=0.6\textwidth]{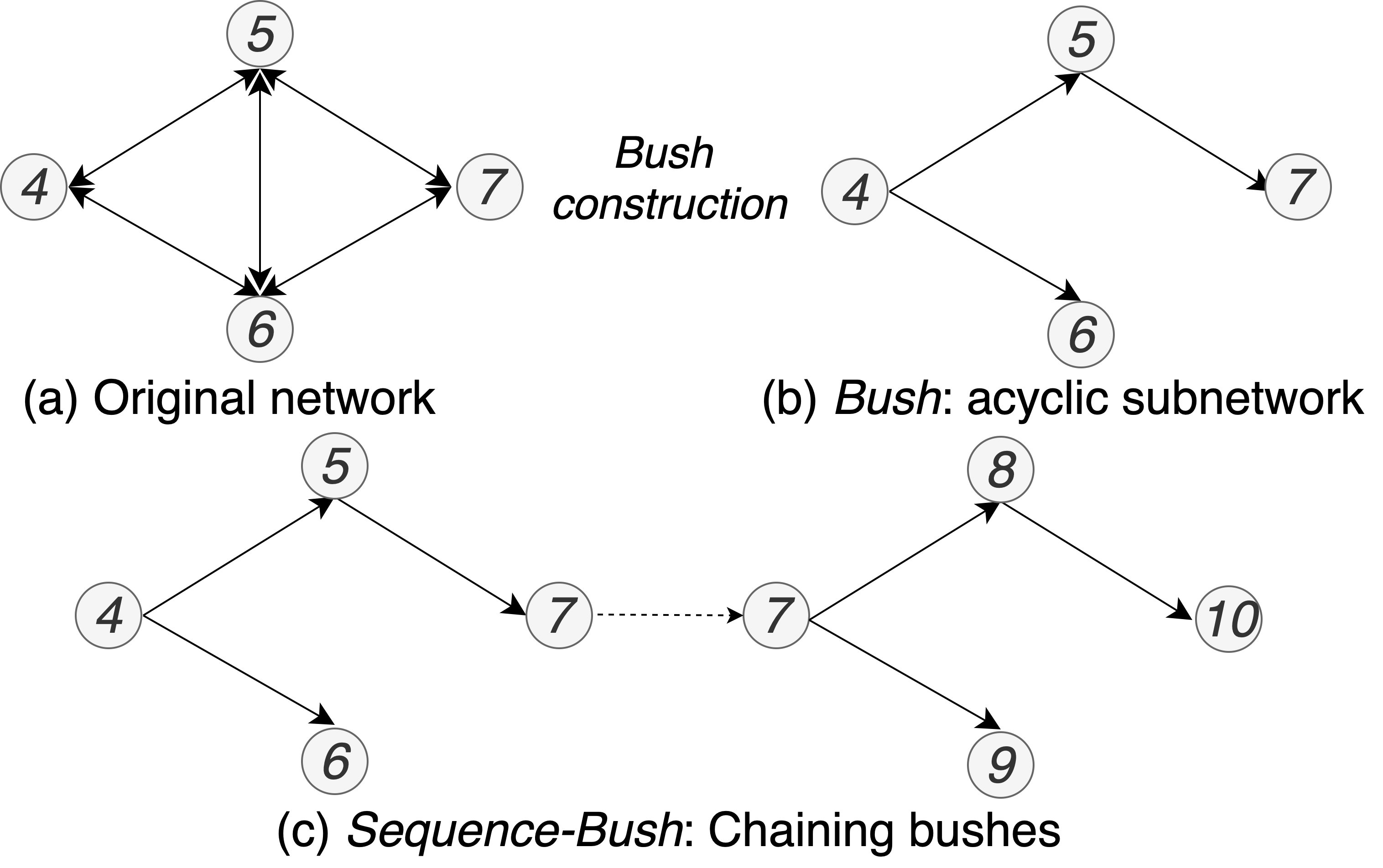}
    \caption{\DIFaddFL{Illustration of bush and sequence-bush}}
    \label{fig:bush_illustration}
\end{figure}

\DIFadd{As illustrated in \mbox{
\cref{fig:bush_illustration}}\hskip0pt
(a)-(}\DIFaddend b), \DIFdelbegin \DIFdel{which corresponds to the equivalent reformulation steps (Subsections~\ref{subsec:alg_step1}-\ref{subsec:alg_step2}). Recall that matching sequences are already transfer-free by construction, which meansRP can avoid transfer if the matched RD follows the matching sequence. Consider the case of a single RD choosing $n$ (i. 
e., $F_n = 1$), he/she follows the matching sequence if and only if RD flows }\DIFdelend \DIFaddbegin \DIFadd{a bush, rooted at RP $u_2$'s origin node 4, is constructed by pruning the original network, such that the resulting subnetwork is acyclic. At each iteration, links with }\textit{\DIFadd{reduced costs}}\citep{nie2010class} \DIFadd{are added to the bush, while other unused links are removed. This means, at initialization, the shortest path tree is used as the bush. For example, the cheapest path between node 4 and 7, $(4, 5, 7)$, is included in this initial bush, which will be assigned with path flows $f_k^{\psi, RD (RP)}$ (detailed in the next subsection). 
}

\DIFadd{To account for matching sequences, sequence-bush is created by chaining bushes. In \mbox{
\cref{fig:bush_illustration}}\hskip0pt
(c), the sequence-bush for RP $u_2$ in $n_1$ is composed of two RP bushes, rooted at node 4 (for $l=2$) and node 7 (for $l=3$), respectively. In \mbox{
\cref{fig:illustrative_example}}\hskip0pt
(b), flows of RP $u_2$ }\DIFaddend are conserved at \DIFdelbegin \DIFdel{each task node }\DIFdelend \DIFaddbegin \DIFadd{node 7 }\DIFaddend (as indicated by the \DIFdelbegin \DIFdel{equal $F_n$). Without loss of generality, we also assume one RD pickup/drop-off one RP at each task (and for $F_n$ RD, the total number of served RP at each task node is $1 \cdot F_n$).
Consequently, RPs are not dropped off at any intermediate task node. }\DIFdelend \DIFaddbegin \DIFadd{same link width and equal sequence-flow $F_{n_1}$). As a result of this conservation, RP $u_2$ follows the matching sequence $n_1$, which that implies $u_2$ is picked up and dropped off accordingly.
}

\subsection{\DIFadd{Solution algorithm: Sequence-bush assignment}}

\begin{figure}[H]
    \centering
    \includegraphics[width=1\textwidth]{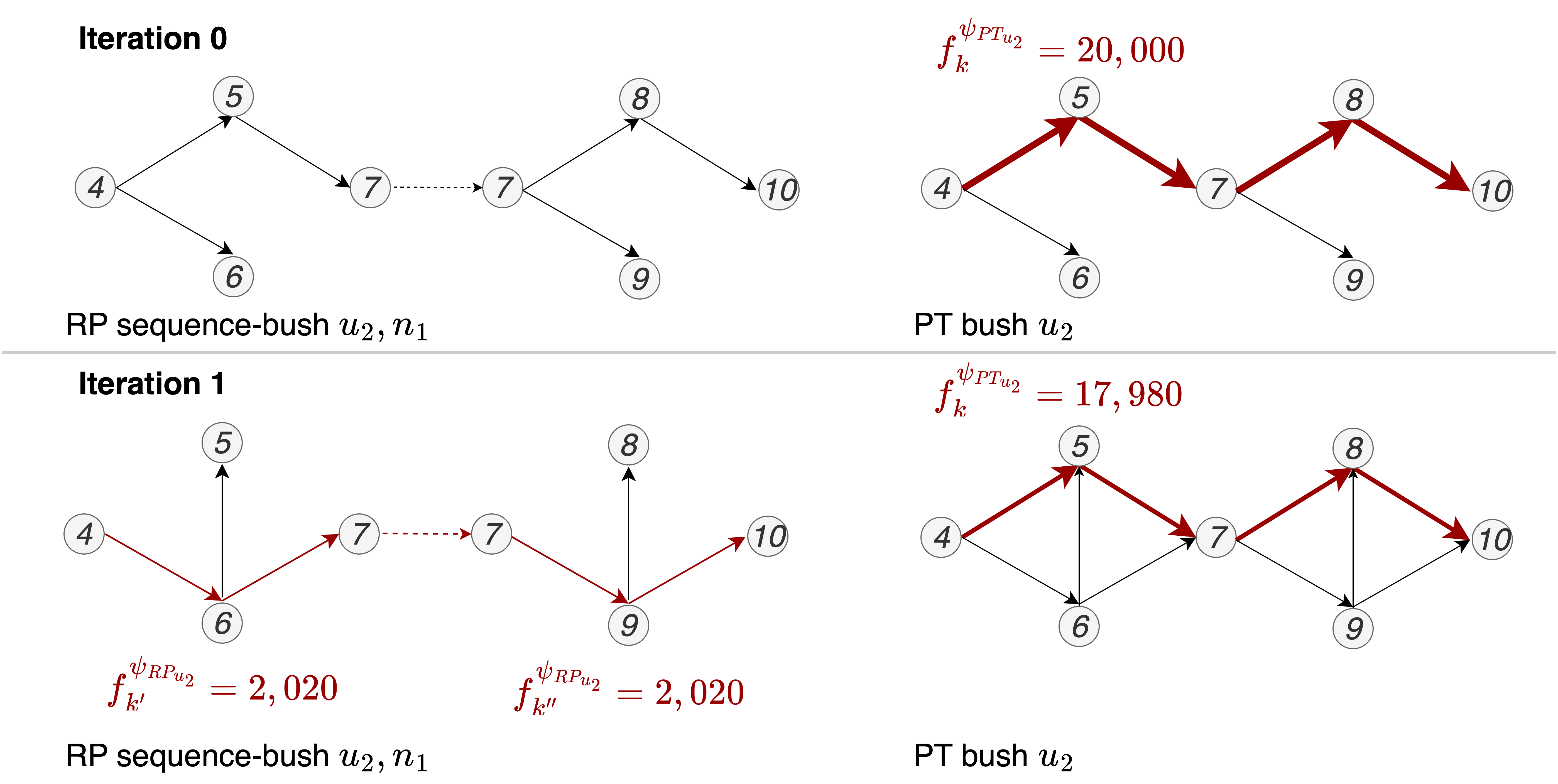}
    \caption{\DIFaddFL{Illustration of sequence-bush assignment iterations}}
    \label{fig:bush_assignment}
\end{figure}

\DIFadd{Given the sequence-bush, this subsection illustrates the sequence-bush assignment algorithm. In \mbox{
\cref{fig:bush_assignment}}\hskip0pt
, all RP demands of $u_2$ are assigned on the cheapest PT route $(4, 5, 7, 8, 10)$ at iteration 0. Similarly, all RD demands are loaded on the cheapest DA route, such that these initial flows are feasible.
}

\DIFadd{After the initial loading, link costs are updated for every mode, and correspondingly the sequence-bushes. This means that cheaper links are added to the bushes and unused links are removed. }\DIFaddend For example, \DIFdelbegin \DIFdel{flows of the on-board RP $u_2$ are conserved at the pickup node of $u_3$ (by the equal incoming and outgoing $1\cdot F_n$), and are only absorbed at the destination of }\DIFdelend \DIFaddbegin \DIFadd{at the end of iteration 0, links $(4, 5), (5, 7)$ become very congested. At the beginning of iteration 1, these links are removed from the RP bush, and new cheapest links $(6, 5), (6, 7)$ are added in both RP and PT bushes. Note that, since there are still PT flows on links $(4, 5), (5, 7)$, they are kept in the PT bush.
}

\DIFadd{The costliest and cheapest sequence-route can be retrieved from these sequence-bushes, namely, the PT route $(4, 5, 7, 8, 10)$ and RP sequence-routes in $n_1$, $k'-k^{''}:(4, 6, 7)-(7, 9, 10)$ in this example. After solving Eq.~\eqref{eq:gradient_projection_problem}, $2,020$ unit of flows are deducted from the PT route $k$, and are loaded on the RP sequence-route $k'-k^{''}$ for RP }\DIFaddend $u_2$. \DIFdelbegin \DIFdel{Furthermore, the equal $F_n$ in RD }\DIFdelend \DIFaddbegin \DIFadd{Moreover, the same $2,020$ unit of flows are simultaneously shifted for the corresponding RD $v$ and peer RP $u_3$ on the same sequence-route $k'-k^{''}$ (from their }\textit{\DIFadd{DA}} \DIFaddend and \DIFdelbegin \DIFdel{RP flows (at each pickup) }\DIFdelend \DIFaddbegin \textit{\DIFadd{PT}} \DIFadd{flows). 
}

\DIFadd{This flow shifting scheme exploits the RD-RP coupling constraints, for which on-board RPs must travel together with their RD on the same route. Furthermore, due to flow conservation at pickup/drop-off nodes, the resulting sequence-flows $F_{n_1}$ for each traveler in $n_1$ are equal, which }\DIFaddend implies stable matching, since either RD or RP leaving \DIFdelbegin \DIFdel{$n$ }\DIFdelend \DIFaddbegin \DIFadd{${n_1}$ }\DIFaddend will break the equality and stability. 
\DIFdelbegin \DIFdel{To handle RD-RP coupling, }\DIFdelend \DIFaddbegin 

\subsection{\DIFadd{Runtime performance and analysis}}

\DIFadd{In \mbox{
\cref{tab:algorithm_performance}}\hskip0pt
, we compare the algorithm performance of the proposed sequence-bush assignment algorithm and }\DIFaddend the \DIFdelbegin \DIFdel{sequence flow is decomposed into path flows $f_k^{\psi, RD (RP)}$ at each level, such that RD and RP traveling together are represented by the equal path flows (as indicated by the same edge width in \mbox{
\cref{fig:illustrative_example}}\hskip0pt
). }\DIFdelend \DIFaddbegin \DIFadd{GAMS/PATH solver for the joint route choice and stable matching problem in \mbox{
\cref{sec:RidesharingNetworkModel}}\hskip0pt
. The sequence-bush algorithm is implemented in Python 3.8 and run on a 6-core machine, while the GAMS is executed on a 32-core server.
}\DIFaddend 

\DIFdelbegin \DIFdel{Based on }\DIFdelend \DIFaddbegin \begin{table}[h]
\caption{\DIFaddFL{Algorithm performance comparison between sequence-bush and GAMS}}
\label{tab:algorithm_performance}
\centering
\setlength\tabcolsep{10pt}
\begin{tabular}{c|cc}
\toprule
      & \textbf{\DIFaddFL{Sequence-bush assignment}} & \textbf{\DIFaddFL{GAMS/PATH solver}} \\
\midrule
\DIFaddFL{Number of iterations }&   \DIFaddFL{82 }& \DIFaddFL{265,827 }\\
\DIFaddFL{Runtime }[\DIFaddFL{s}]   & \DIFaddFL{$3.92$ }& \DIFaddFL{$446.34$ }\\
\DIFaddFL{Number of variables   }& \DIFaddFL{$780$ }& \DIFaddFL{$119,642$}\\
\bottomrule   
\end{tabular}
\end{table}

\DIFadd{Results show that the proposed sequence-bush assignment algorithm significantly outperforms the GAMS solver in terms of runtime and number of iterations for the illustrative example. This can be explained by the number of variables involved at each iteration for the sequence-bush algorithm and GAMS solver. 
}

\DIFadd{We could estimate the order for the number of variables for the two methods: 1) for the sequence-bush method, it requires $O(|\mathcal{O}| + |\mathcal{D}|)$ bushes for each flow type, and 2 routes for each matching sequence $O(|\mathcal{S}|)$ at each level $O(|\mathcal{L}|)$, which results in total of $O\left[(|\mathcal{O}| + |\mathcal{D}|) \cdot |\mathcal{S}| \cdot |\mathcal{L}| \right]$ for the number of path variables; 2) for the link-based formulation used in GAMS, }\DIFaddend the \DIFdelbegin \DIFdel{gap-based path-flow reformulation~\eqref{eq:reformulation_route_flow},}\DIFdelend \DIFaddbegin \DIFadd{highest cardinality of the link-flow sets is the one for RP, which is defined for each link in }\DIFaddend the \DIFaddbegin \DIFadd{hyper-network $O(|\mathcal{E}| \cdot |\mathcal{L}|)$ for each matching sequence $O(|\mathcal{S}|)$ and each destination $O(|\mathcal{D}|)$, and in total of $O\left[|\mathcal{E}| \cdot |\mathcal{L}| \cdot |\mathcal{S}| \cdot |\mathcal{D}|  \right]$. As a result, it is more difficult to solve at each iteration for the link-based formulation in GAMS. On the other hand, the }\DIFaddend sequence-bush \DIFdelbegin \DIFdel{algorithm exploits this stability constraint and solves only the RD path flow $f_k^{\psi, RD}$ (and RP path flow is set equal to RD flow) . To maintain conservation, flows are shifted from the costliest to the cheapest sequence-route, which are created by connecting costliest (cheapest) routes at each level in sequence.  As shown in \mbox{
\cref{tab:illustrative_example_path_costs}}\hskip0pt
, all with-flow RD and RP paths, as well as the corresponding sequence-routes, are the cheapest.
Note that, by the adaption of a }\DIFdelend \DIFaddbegin \DIFadd{algorithm exploits the network property to reduce the solution space and improves the algorithm runtime performance, while the computation overhead for maintaining the bushes is negligible~}\citep{nie2004models}\DIFadd{.
}

\DIFadd{Although the }\DIFaddend sequence-bush \DIFdelbegin \DIFdel{representation, at each iteration, we only need 2 routes (costliest and cheapest)at each level, which requires finding in total $2\cdot 5=10$ routesin this example, compared to enumerating $4^5=1024$ routes}\DIFdelend \DIFaddbegin \DIFadd{method and the column-generation approach both add cheapest path into the solution space dynamically, the key difference is that the sequence-bush can determine simultaneously the cheapest and costliest routes with one simple forward pass in an acyclic subnetwork (bush). Moreover, the sequence-bush avoids path set storage, while path-based approach requires keeping $|k|^{|\mathcal{L}|}$ exponential number of routes}\DIFaddend . 

\section{Numerical results}
\label{sec:Results}
In this section, we present numerical results on the proposed general model for multi-passenger ridesharing systems. It is worth noticing that, solving the proposed model for large networks is challenging. As suggested in \textcite{chen2022unified}, general-purpose MCP solvers, like GAMS, might not be suitable for real-size networks. To the best of our knowledge, the proposed sequence-bush algorithm is one of the first methods that is able to solve real-size problems. We illustrate our model using the full Sioux-Falls network (\cite{bargeraNetwork}), which has 24 nodes, 76 links, 529 OD pairs, and a total demands of 360,600 trips (we do not find previous studies able to solve the full Sioux-Falls network). We set cost gap criteria as $\epsilon_M \leq 10^{-2}, \epsilon_N \leq 10^{-2}$ and flow infeasibility criteria as $\epsilon_3 \leq 5\cdot10^{-3}$. For all the Sioux Falls network experiments, each outer iteration takes about 20 minutes, and the overall convergence is reached within 5 outer iterations. 

Similar to \textcite{li2020path}, we assume the link cost functions are specified as in \cref{tab:link_cost_final} with the \DIFdelbegin \DIFdel{following parameters :
}\DIFdelend \DIFaddbegin \DIFadd{parameters in \mbox{
\cref{tab:parameters}}\hskip0pt
, whereas the volume delay functions are specified as in the dataset~}\citep{bargeraNetwork}\DIFadd{.
}\DIFaddend \begin{table}[H]
\caption{Parameter setting for the numerical examples.}
\label{tab:parameters}
\centering
\setlength\tabcolsep{10pt}
\begin{tabular}{c|cccccc}
  \toprule
   Mode & $\alpha$ & $\beta$ & $\tau_t$ & $\tau_d$ & $\nu_t$ & $\nu_d$ \\
   \midrule
DA & 1.0          & 1.0       &             &             &         &         \\
RD & 1.0        & 1.0       & 0.3         & 0.2         & 0.3       & 0.7     \\
RP & 0.6        &         & 0.3         & 0.1         & 0.1       & 0.4     \\
PT & 0.4        &         & 0.6         & 0.6           & 0       & 0.4     \\
\bottomrule   
\end{tabular}
\end{table}

Moreover, for simplicity, we assume public transport \DIFdelbegin \DIFdel{are }\DIFdelend \DIFaddbegin \DIFadd{is }\DIFaddend operated in a segregated system such that their travel times are flow-independent. A public ridesharing platform is also assumed in this paper, in which the matching objective is to maximize the VKT savings (\cite{yao2021dynamic}). 
\subsection{Example of the stable matching and route choice problem}
One of the key contributions of this paper is the transformation of the stable matching problem into a network problem, such that stable matching is jointly determined by the route choice model. In this subsection, we focus on illustrating the network component of the overall GEM-mpr, by showing the equilibrium flows and generalized costs for the matching sequences, and interpreting matching stability from the results. 

\begin{figure}[H]
    \centering
    \includegraphics[width=0.55\textwidth]{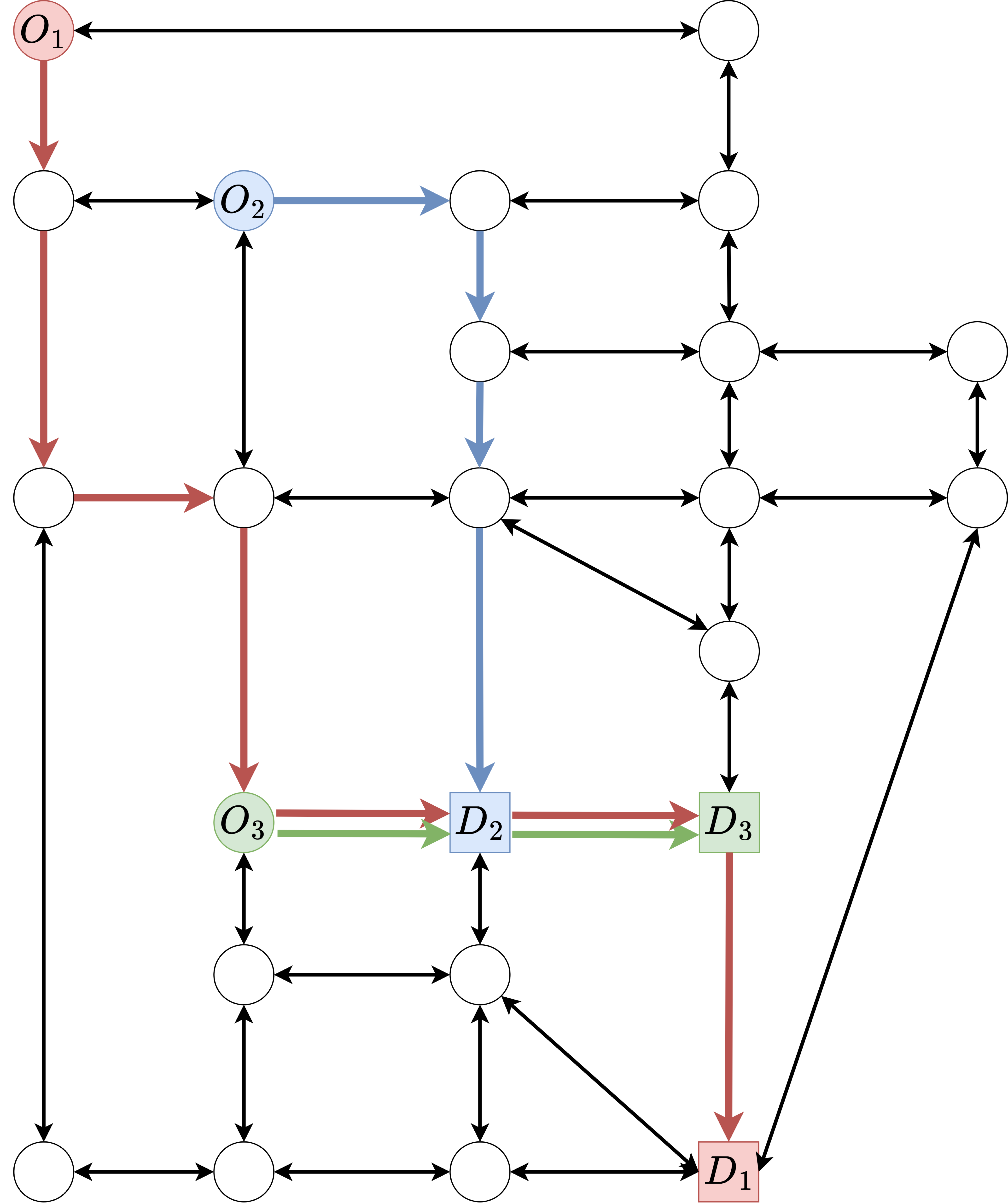}
    \caption{Sioux-Falls example (with drive-alone \DIFdelbeginFL \DIFdelFL{shortest }\DIFdelendFL \DIFaddbeginFL \DIFaddFL{cheapest }\DIFaddendFL routes marked)}
    \label{fig:9.siouxfalls_illustration}
\end{figure}

As shown in \cref{fig:9.siouxfalls_illustration}, we assume in this example, there are RD with vehicle capacity of 2, traveling between $(O_1, D_1)$ with demands of 10,000; and RP traveling between $(O_2, D_2)$ and  $(O_3, D_3)$, with demands of 20,000 for each OD pair. The cheapest drive-alone routes are also marked for each OD pair in \cref{fig:9.siouxfalls_illustration}. This illustrative example has in total of 12 matching sequences, and both RD and RP can choose to quit ridesharing. We show in the following the equilibrium flows and generalized costs of the matching sequences from the RD's perspective. 
\begin{table}[h]
\caption{Stable matching and route choice results (Top 5 cheapest)}
\label{tab:toy_results}
\centering
\setlength\tabcolsep{10pt}
\begin{tabular}{c|c|cccc}
  \toprule
   & Matching sequence                                      & Generalized cost & RD cost & RP cost & Flow \\
   &                                                       &     (RD Cost + Multipliers)    &             &             &             \\
   \midrule
1& $(O_1, O_3, O_3, D_3, D_3, D_1)$ & 45.59          & 43.55       &   22.62          &    5451.75         \\
2& $(O_1, O_2, O_2, D_2, D_2, D_1)$ & 45.59        & 45.59       & 19.58         & 4548.75           \\
3& Quit & 48.00        &   48.00      &    -      & 0.00            \\
4& $(O_1, O_2, D_2, D_1)$ & 48.59        &  47.09       & 21.08         & 0.00             \\
5& $(O_1, O_3, D_3, D_1)$ & 51.36        &    45.40     & 27.56        & 0.00             \\
\midrule
        \multicolumn{2}{c}{\textbf{Gap}$^{*}$} &           $3.37 \times 10^{-5}$ & & & \\
\bottomrule   
\end{tabular}
\begin{tablenotes}
  \item *: Gap between cheapest and with-flow most costly route costs in the hyper-network. 
\end{tablenotes}
\end{table}

As presented in \cref{tab:toy_results}, only the two cheapest (in terms of generalized cost) matching sequences 1 and 2 have flows, which correspond to serve 2 RPs between the same OD pair $(O_3, D_3)$and $(O_2, D_2)$, respectively. Also, the gap in generalized cost between the cheapest and \textit{with-flow} most costly routes, that correspond to these two matching sequences, is only $3.37 \times 10^{-5}$. From the equilibrium perspective, no RD can reduce his/her generalized cost by unilateral switching to another matching sequence, which captures matching stability. Moreover, RD/RP cost of the matching sequences are accumulated by traversing links in the network. In terms of route choice, such negligible gap in generalized costs also suggests equilibrium route choice.

Apart from showing equilibrium, these results also demonstrate the \textit{stability} in matching choices. From the driver's point of view, matching sequences 1 and 5 are preferable, since the driver need not make detours to pickup passengers from $(O_3, D_3)$ (i.e., on the driver's cheapest drive-alone path, marked in red). However, matching sequence 5 has the highest generalized cost among 5 options. This is because passengers from $(O_3, D_3)$ are not willing to form such matching, as they have a better option, matching sequence 1, available for them (i.e., this observation corresponds to the no \textit{blocking pair} in stable matching literature). Another interesting observation is that, although matching sequence 4 has a higher RD cost than matching sequence 5, it has a lower generalized cost than option 5. This is because passengers in matching sequence 5 have a higher tendency (in terms of differences in RP costs) to choose option 1 \DIFdelbegin \DIFdel{that }\DIFdelend \DIFaddbegin \DIFadd{which }\DIFaddend has lower RP costs, compared to passengers switching from matching sequence 4 to 2. This, again, shows our model is able to capture stable matching within a hyper-network modeling framework.

\subsection{Comparison between without and with ridesharing}
In this subsection, we compare between the base scenario, transportation system without ridesharing, and the ridesharing scenario, to illustrate the potential network benefits of ridesharing. Note that, for the following subsections, we apply the full demands of the Sioux Falls network (529 OD pairs and 360,600 trips) instead of selected OD pairs.
\begin{table}[h]
\caption{Network performance without/with ridesharing}
\label{tab:with_without_RS_results}
\centering
\setlength\tabcolsep{40pt}
\begin{tabular}{c|cc}
\toprule
      & \textbf{Without RS} & \textbf{With RS} \\
\midrule
VKT$^{*}$   & 2,651,200.47 & 2,597,822.91 \\
VHT$^{*}$   & 3,489,151.65 & 3,418,496.93 \\
Trip saved   & \multicolumn{2}{c}{4,963.96 (1.38\%)}  \\
VKT saved   & \multicolumn{2}{c}{53,377.56 (2.01\%)} \\
VHT saved   & \multicolumn{2}{c}{70,654.72 (2.02\%)} \\
\bottomrule   
\end{tabular}
\begin{tablenotes}
  \item *: Public transport passenger car equivalent (PCE) assumed 3. 
\end{tablenotes}
\end{table}

Before introducing ridesharing in the transportation system, DA has a modal split of 77.93\%, and PT is of 22.07\%. After the introduction of ridesharing, modal splits of DA dropped to 75.60\%, PT modal split reduced to 21.49\%, with 1.54\% as RD and the remaining 1.38\% as RP. Compared to PT (-0.58\%), there is a larger reduction in DA (-2.33\%) after the introduction of ridesharing, which suggests, under the assumed parameter setting, ridesharing is attracting more travelers from the DA. This is confirmed by the reduced VKT and VHT in \cref{tab:with_without_RS_results}, in which there are 1.38\% saved trips, and 2.01\% and 2.02\% savings in VKT and VHT, respectively. Such results illustrate the potential of ridesharing services to reduce traffic congestion in the network. 
\begin{table}[h]
\caption{Ridesharing origin-destination (OD) pairs}
\label{tab:RS_pairs}
\centering
\setlength\tabcolsep{14pt}
\begin{tabular}{c|cc}
\toprule
      & \textbf{OD pair} & \textbf{Avg. trip} \\
      & (origin, destination) & \textbf{length [km]} \\
\midrule
RD   & (1, 13), (6, 13), (10, 1), (10, 2), (11, 1), (13, 1), (13, 2), (13, 8), (14, 1), & 17.54 \\
   & (14, 2), (17, 1), (18, 1), (19, 1), (19, 4), (19, 5), (22, 4), (24, 6) &  \\
\midrule
RP   & (1, 13), (2, 13), (7, 1), (9, 1), (9, 2), (11, 1), (11, 2), (12, 2),  & 15.89 \\
   &  (12, 6), (13, 1), (13, 6), (15, 1), (15, 3), (15, 4), (15, 5), (16, 1) &  \\
\bottomrule   
\end{tabular}
\end{table}

We also observe higher percentage of VKT and VHT savings compared to trip savings (\cref{tab:with_without_RS_results}). As shown in \cref{tab:RS_pairs}, this is because the average ridesharing (RD and RP) trip lengths are significantly longer than overall average trip length (8.82 km). This observation is \DIFdelbegin \DIFdel{consistency }\DIFdelend \DIFaddbegin \DIFadd{consistent }\DIFaddend with our previous findings (\cite{yao2022}), in which, for a reasonable ridesharing price, only drivers and passengers with relatively longer trips will participate in ridesharing.
\begin{figure}[H]
    \centering
    \includegraphics[width=0.52\textwidth]{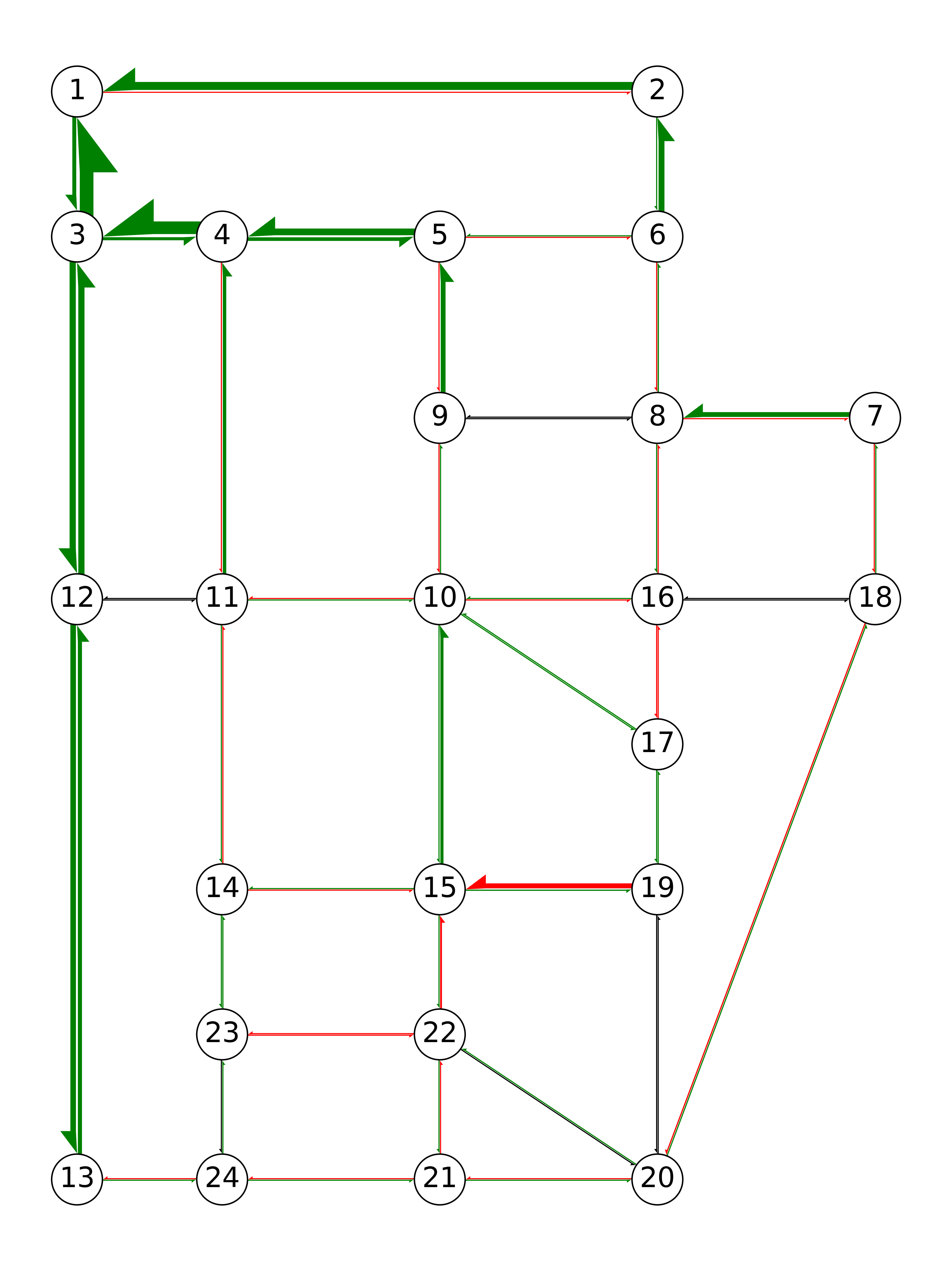}
    \caption{Sioux-Falls link flow changes (reduced flow in green, increased flow in red)}
    \label{fig:11.with_RD_link_flow}
\end{figure}

The impacts of ridesharing on the network link flows are further demonstrated in \cref{fig:11.with_RD_link_flow}. The link flow changes (compared to the without ridesharing scenario) are consistent with tables \eqref{tab:with_without_RS_results}-\eqref{tab:RS_pairs}, in which the magnitudes of the link flow reductions (links in green) are larger than the increased traffic flows (links in red). The increased traffic flows correspond to our ridesharing setting that RD may detour to pick up and drop off multiple passengers from different OD, with ridesharing matching objective of maximizing VKT savings. Moreover, we observe that most link flow reductions are around the main RP destinations (\cref{tab:RS_pairs}), in which many of these trips cross the network and have longer trip lengths.  

\subsection{Sensitivity analysis - ridesharing unit price}
In this subsection, we perform sensitivity analysis to investigate the impacts of ridesharing unit price on ridesharing modal splits and network benefits. For simplicity, we vary only the distance-based unit price, $\nu_d^{RD}$ between 0 and 1.5, and assume $\nu_d^{RP}=0.5\nu_d^{RD}$.
\begin{figure}[H]
    \centering
    \includegraphics[width=0.9\textwidth]{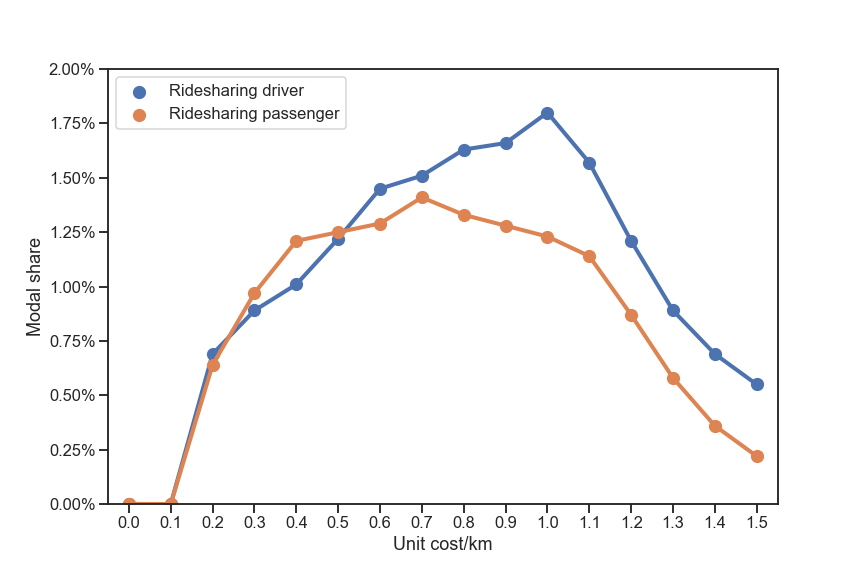}
    \caption{Impact of unit price on ridesharing modal shares}
    \label{fig:10a.sensitivity_pricing_1}
\end{figure}
As shown in \cref{fig:10a.sensitivity_pricing_1}, when unit price is low (0.0-0.1), there is neither RD nor RP at equilibrium. From the behavioral point of view, when the price is low, passengers are very much willing to participate in ridesharing, but drivers have no incentive to take passengers. Therefore, due to coupling between RD and RP, RP modal splits at equilibrium should be low. Our sensitivity analysis results demonstrate such coupling effects.

When the ridesharing unit price increase, more drivers start picking up passengers to compensate their costs. As a result, both RD and RP modal split increase. In this example, there is a passenger surplus for unit cost $\nu_d^{RD}< 0.5$, which suggests these prices are beneficial towards passengers. With larger unit cost ($\nu_d^{RD}> 0.5$), driver surplus is observed, which suggests such pricing strategy acts towards the drivers' benefits.

As the unit price gets too high, passengers start leaving the ridesharing system, while more drivers consider ridesharing attractive for its high compensations. But due to RP and RD coupling, when there is not enough RP, RD not matched with passengers would eventually leave the system. These results indicate the proposed GEM-mpr properly captures the interactions between drivers and passengers, so as their responses to different operational strategies.

\begin{figure}[H]
    \centering
    \includegraphics[width=0.9\textwidth]{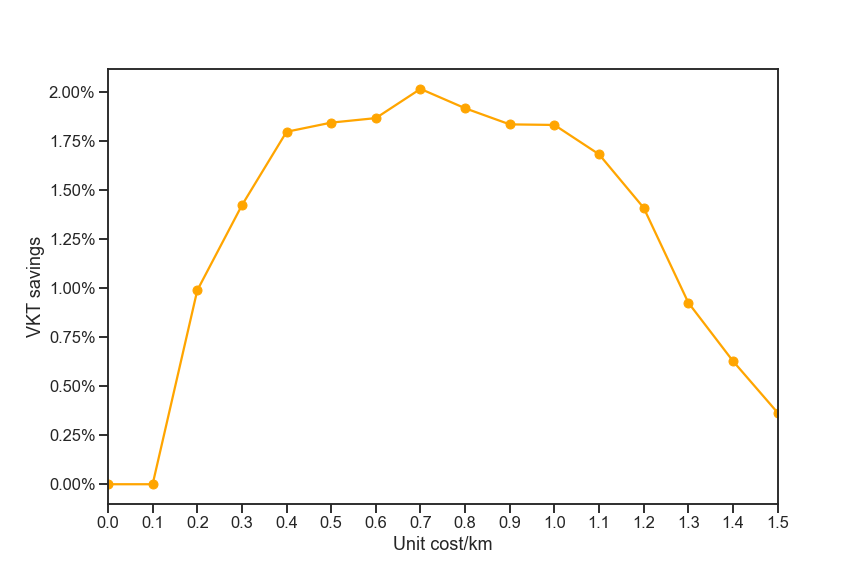}
    \caption{Impact of unit price on network benefits}
    \label{fig:10b.sensitivity_pricing_2}
\end{figure}

We show in \cref{fig:10b.sensitivity_pricing_2} the impacts of ridesharing pricing \DIFdelbegin \DIFdel{in }\DIFdelend \DIFaddbegin \DIFadd{on }\DIFaddend network benefits. As already suggested \DIFdelbegin \DIFdel{from }\DIFdelend \DIFaddbegin \DIFadd{in }\DIFaddend the previous subsection, ridesharing network benefits are related to the RP trips, in which RP travel with RD to reduce vehicle trips. VKT saving results are consistent with the RP modal share in \cref{fig:10a.sensitivity_pricing_1}. As the unit price increase, there are more passengers being matched, therefore the VKT savings increase. However, when the unit price is too high, RP start leaving the ridesharing systems, so as the VKT savings start decreasing. These results suggest ridesharing prices should be carefully designed to realize ridesharing network benefits.

\subsection{Sensitivity analysis - driver value of time (VOT)}
In this subsection, the impacts of the driver's (DA and RD) VOTs on ridesharing modal shares and network performances are analyzed. We assume the driver's VOTs vary between 1 and 2, and $\alpha^{DA}=\alpha^{RD}$ for simplicity.

As shown in \cref{fig:12a.sensitivity_VOT_1}, the RD, RP and PT modal shares decrease as the driver's VOTs increase, while only DA modal share increases. It implies that, under the proposed parameter setting, ridesharing compensations cannot attract travelers with high VOT to become RD. Consequently, RP, who cannot be matched with drivers, leave ridesharing. These results are consistent with the literature (\cite{alonso2021determinants}), in which high VOT individuals are less willing to rideshare, or take public transport. 
\begin{figure}[H]
    \centering
    \includegraphics[width=0.9\textwidth]{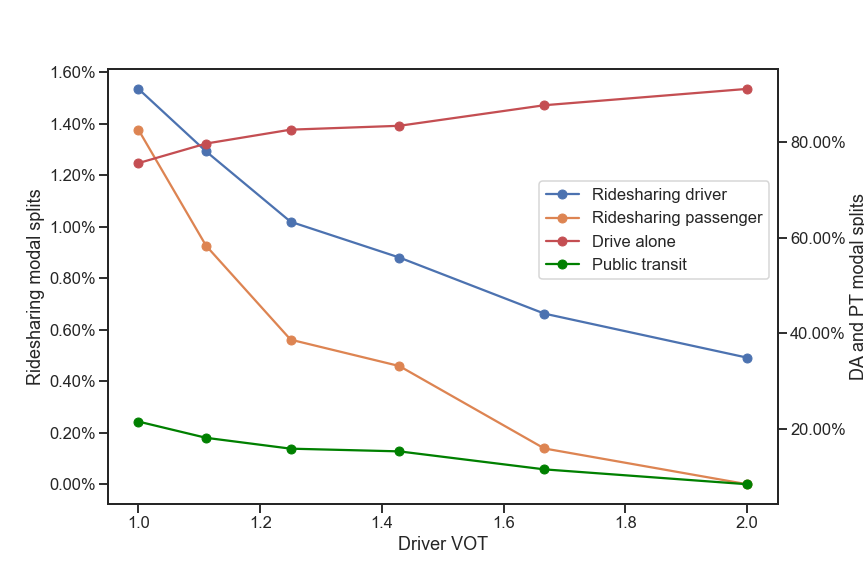}
    \caption{Impact of driver VOT on modal shares}
    \label{fig:12a.sensitivity_VOT_1}
\end{figure}
\begin{figure}[H]
    \centering
    \includegraphics[width=0.9\textwidth]{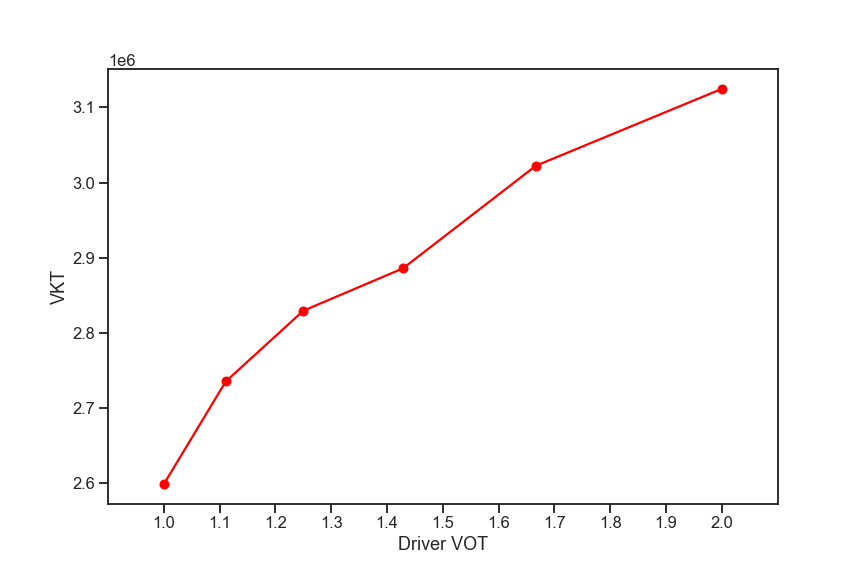}
    \caption{Impacts of driver VOT on network performances}
    \label{fig:12b.sensitivity_VOT_2}
\end{figure}

We further evaluate the impacts of driver's VOTs on network performances in \cref{fig:12b.sensitivity_VOT_2}, in which the VKT increases as the driver VOT increases. This is related to the increases in DA modal splits (\cref{fig:12a.sensitivity_VOT_1}), as more travelers choose to drive alone, the more congested the network becomes. These results suggest that ridesharing services might not be attractive for travelers with high VOT, and additional measures should be taken to promote these services and reduce traffic congestion. 

\subsection{Algorithm convergence}
In this subsection, we show \DIFdelbegin \DIFdel{in \mbox{
\cref{fig:convergence} }\hskip0pt
}\DIFdelend the algorithm convergence results for both the last inner iterations (sequence-bush assignment) and the outer iterations (augmented Lagrangian). Recall that, for all the runs with Sioux Falls network, we set the convergence criteria in terms of normalized gap between the costliest and cheapest mode option $\epsilon_M \leq 10^{-2}$ and route option $\epsilon_N \leq 10^{-2}$, that is, the inner iteration stops when both the mode and route cost differences are no more than 0.01\$. For the outer iterations, we consider the overall problem converges if the flow infeasibility $\epsilon_3 \leq 5\cdot10^{-3}$. 

As indicated in \DIFdelbegin \DIFdel{\mbox{
\cref{fig:convergence}}\hskip0pt
(a)}\DIFdelend \DIFaddbegin \DIFadd{\mbox{
\cref{fig:convergence_a}}\hskip0pt
}\DIFaddend , the route cost normalized gap reaches $10^{-3}$ whereas the modal cost normalized gap converges at $10^{-2}$ (\DIFdelbegin \DIFdel{\mbox{
\cref{fig:convergence}}\hskip0pt
b)}\DIFdelend \DIFaddbegin \DIFadd{\mbox{
\cref{fig:convergence_b}}\hskip0pt
}\DIFaddend). This result suggests the network model is important for the convergence of the inner iteration, and the sequence-bush assignment algorithm is able to provide sufficient accuracy for the algorithm convergence. The overall convergence result in \DIFdelbegin \DIFdel{\mbox{
\cref{fig:convergence}}\hskip0pt
}\DIFdelend \DIFaddbegin \DIFadd{\mbox{
\cref{fig:outer_convergence}}\hskip0pt
}\DIFaddend (c) indicates the augmented Lagrangian method is suitable for finding the overall equilibrium of the proposed model. 

\begin{figure}[H]
  \centering 
  {\includegraphics[width=0.85\linewidth]{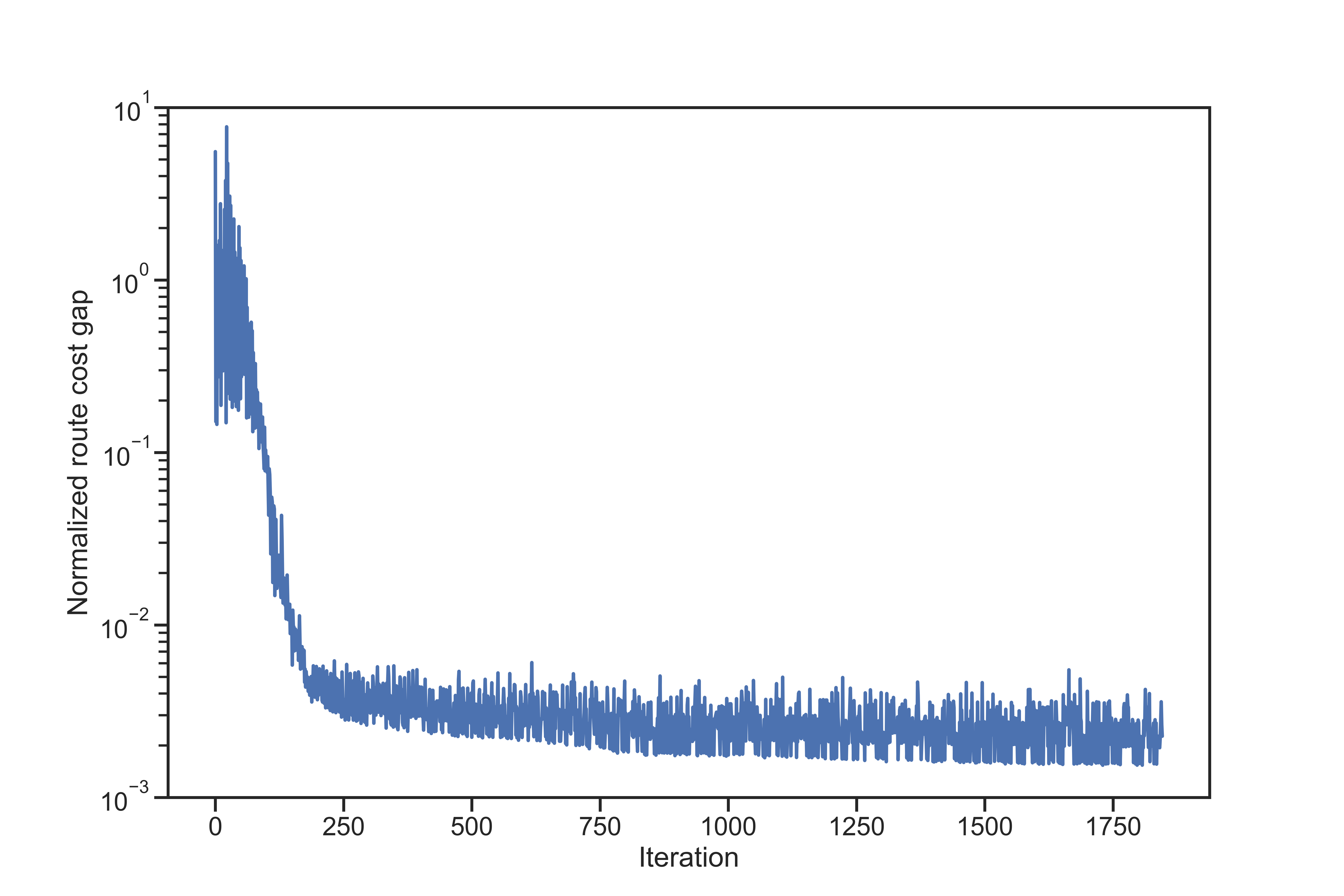}}
  \caption{Algorithm convergence: Normalized route cost gap (Sioux Falls with RS)}
  \label{fig:convergence_a}
\end{figure}
\begin{figure}[H]
  \centering 
  {\includegraphics[width=0.85\linewidth]{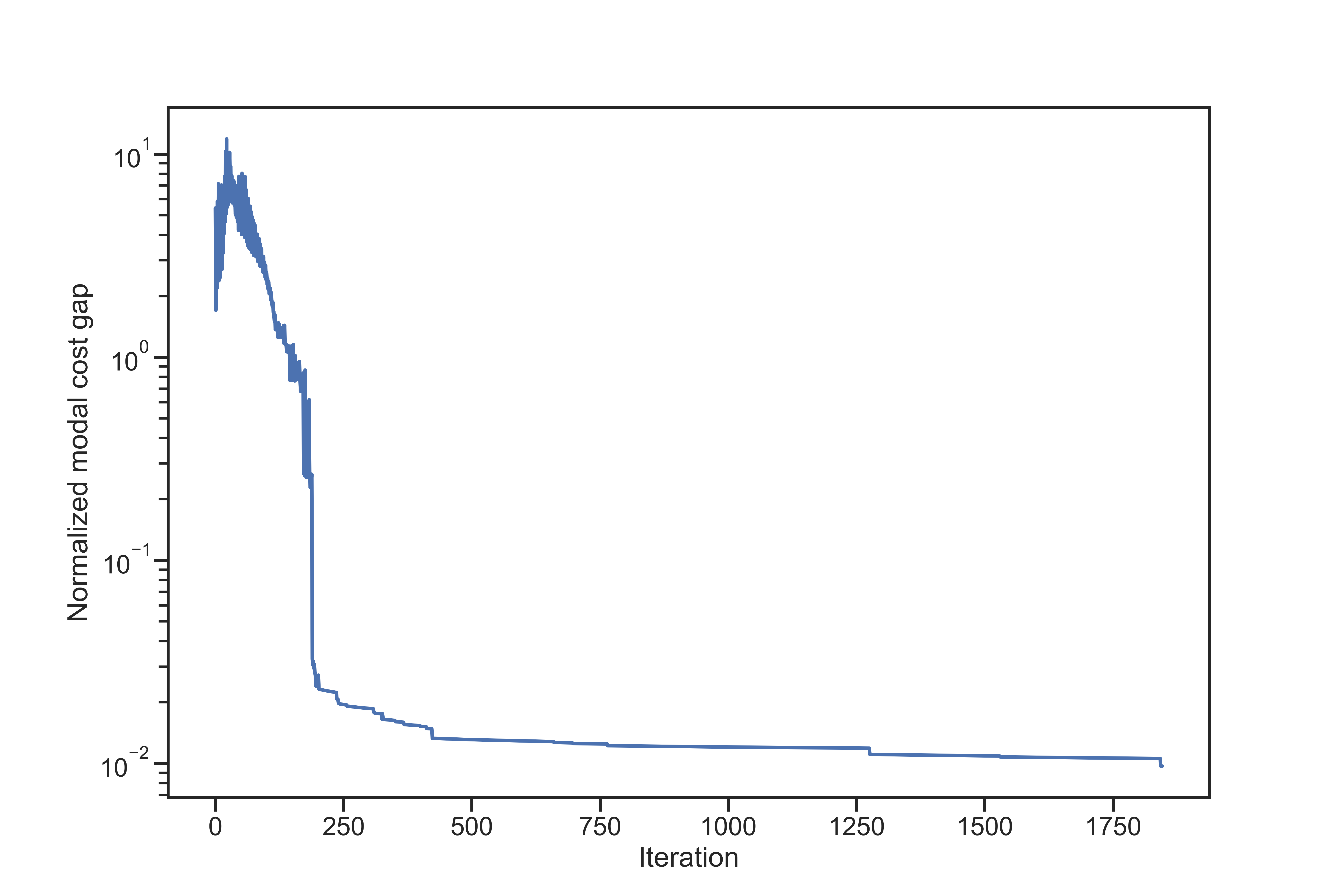}}
  \caption{Algorithm convergence: Normalized modal cost gap (Sioux Falls with RS)}
  \label{fig:convergence_b}
\end{figure} 
\begin{figure}[H]
  \centering 
  {\includegraphics[width=0.85\linewidth]{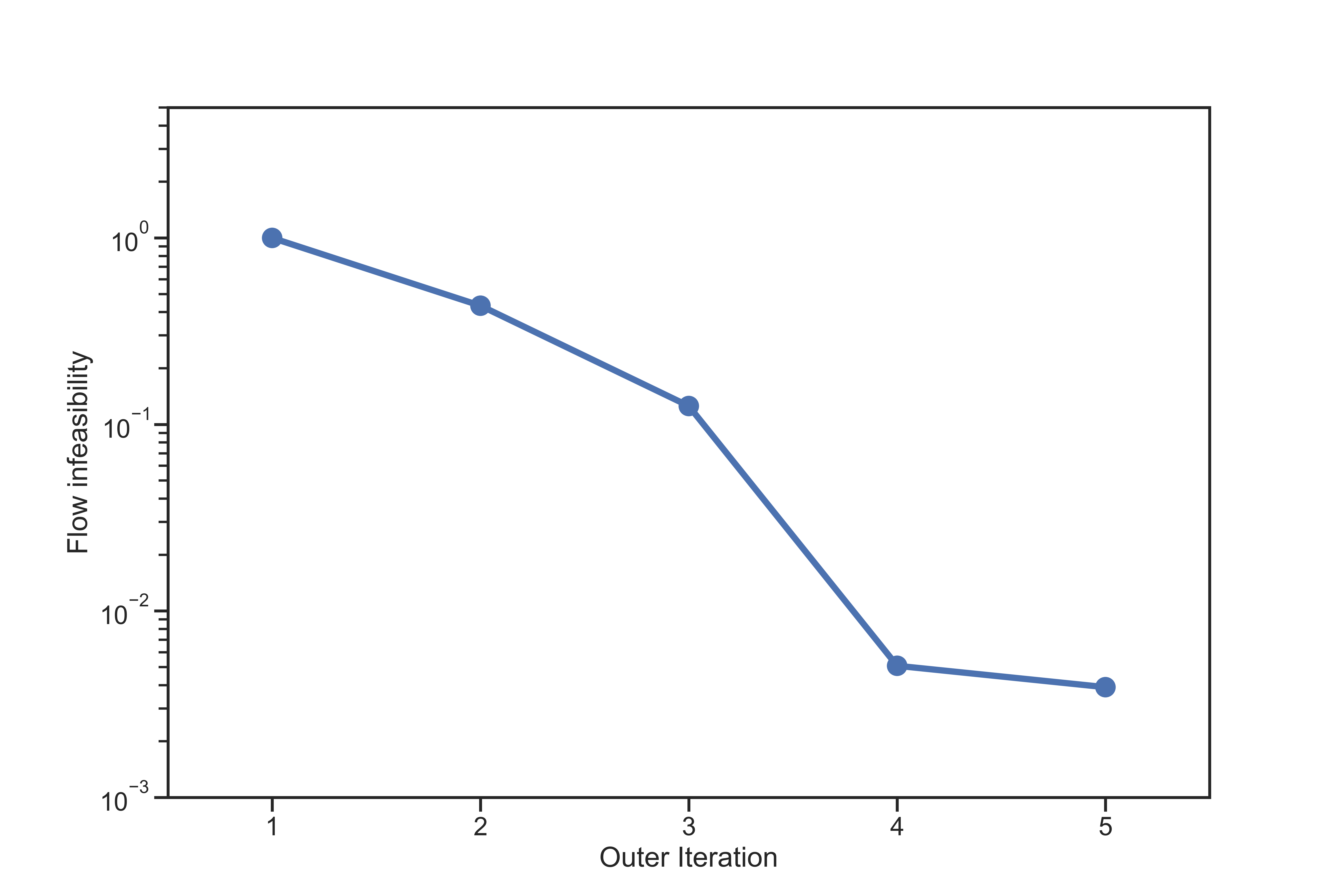}}
  \caption{Algorithm convergence: Overall convergence by Eq.~\eqref{eq:outer_convergence} (Sioux Falls with RS)}
  \label{fig:outer_convergence}
\end{figure}

\DIFaddbegin \DIFadd{Note that, similar to most first-order methods (e.g., gradient descent), the proposed algorithm also slows down when approaching the optimal, due to small gradient values. Future research could incorporate second-order information~}\citep[e.g.,][]{du2021faster} \DIFadd{or accelerated gradients~}\citep{nesterov1983method,beck2009fast} \DIFadd{to speedup the convergence. 
}

\DIFaddend \section{Discussion}
\label{sec:Discussion}
\noindent
\textbf{Extension with waiting time.} The proposed hyper-network approach (\cref{fig:5.hypernetwork}) could be extended to accommodate pickup waiting times, by imposing additional costs on virtual links $(\widehat{i_{RP}}, i_{RP}^l)$. The pickup times could be endogenously determined by computing RD's travel time to the RP's pickup location through traversing trajectories in the matching sequences. 
Whereas matching waiting times could be accounted for endogenously or exogenously in the RD and RP modal cost functions~\eqref{eq:RD_modal_costs}-\eqref{eq:RP_modal_costs}.

\noindent
\textbf{Comparison between sequence-bush and DARP.} The sequence-bush approach differs from the DARP in the following aspects: 1) entity and objective, and 2) travel time. In terms of entity and objective, DARPs are typically solved for an operator (or platform) who aims to optimize its objective (e.g., maximizing the profit), which corresponds to a \textit{system optimal} solution. Whereas sequence-bush solves for a group of travelers (RD and several RPs), in which each of them aims to minimize his/her own travel disutility (\textit{user equilibrium}). In terms of travel times, DARP typically uses flow-independent travel times~\citep[e.g.,][]{yao2021dynamic}, while sequence-bush approach uses flow-dependent travel times and assigns travelers to the network.
Hypothetically, if each RD and RP chooses to maximize the ridesharing platform's objective (e.g., profit), the sequence-bush solution could collapse to the DARP solution. Otherwise, RD and RP could deviate from the platform's matching and routing. 

\noindent
\textbf{Matching-sequence approach from the viewpoint of a joint assignment.} Ridesharing matchings are typically made between a group of RDs and RPs, such that each RD picks up a few RPs. Under such setting, RD and RPs are naturally considered as discrete (i.e., one RD and one RP). The proposed matching sequence corresponds to this discrete setting, in which one RD follows a matching sequence to pickup/drop-off one RP at each task node. Consequently, the proposed matching-sequence can be seen as a joint assignment, in which $F_n$ RD choose matching sequence $n$ and $1\cdot F_n$ RPs are served at each task node, whereas the continuous $F_n$ can be interpreted as the proportion of RD (and RP) choosing $n$.

\section{Summary}
\label{sec:Summary}
\noindent
In this paper, we propose a general equilibrium model for multi-passenger ridesharing systems. The proposed model considers the interactions between ridesharing drivers, passengers, platforms, and the transportation network. Most studies did not explicitly consider multi-passenger ridesharing matching in ridesharing equilibrium models. Therefore, we believe that it is necessary to explicitly integrate the ridesharing platform's multi-passenger matching problem into the model.

The proposed hyper-network approach relaxes the assumption that ridesharing passengers need to make transfers in a multi-passenger ridesharing system. Moreover, the matching stability between ridesharing drivers and passengers is extended to the multi-OD multi-passenger case, by transforming the stable matching problem into a route choice problem in the hyper-network. Solution existence for the proposed general equilibrium model is shown following a three-step procedure.

One of the key contributions of this paper is the hyper-network approach. Assuming travelers make mode choice decisions among 4 alternative modes: driver alone, ridesharing driver, ridesharing passenger, and public transport, the original networks are extended to the multi-modal setting. Furthermore, the multi-modal extended network is expanded for the matching sequences. The matching-sequence expanded network provides a tool that not only can handle multi-passenger ridesharing without the need to make transfers, but also \DIFdelbegin \DIFdel{allow }\DIFdelend \DIFaddbegin \DIFadd{allows }\DIFaddend handling the stable matching problem in the hyper-network and determining the ridesharing driver routes endogenously.

Within a non-cooperative game-theoretic framework, platforms and travelers can make their decisions. We explicitly consider these interactions occurred from their decisions. Given the platform's matching decisions, ridesharing drivers and passengers can choose their preferred matching sequence, such that stable matching is achieved at equilibrium. Through matching sequence, the ridesharing platform can also influence the drivers' routing decisions for picking up and dropping off passengers. 

The proposed equilibrium model falls into the generalized type of Nash equilibrium, in which solution existence is not trivial to show. A three-step procedure is developed to tackle the shared and asymmetric constraints and the lack of explicit bound on the decision variables, by adapting the VI solution existence theorem for unbounded sets (\cite{Nagurney2009}) and penalty-based method for relaxing the asymmetric constraints (\cite{ban2019general}).

A sequence-bush assignment algorithm is developed for solving the multi-passenger ridesharing equilibrium problem, in which the overall problem is decomposed with respect to RD-RP groups in terms of sequences. By the introduction of route-flow variables, the sequence-bush approach avoids path enumeration and storage by connecting bushes in sequence, with which the complicate ridesharing constraints can be handled implicitly. 

Results illustrate the proposed sequence-bush algorithm is able to find equilibria for the joint stable matching and route choice problem\DIFaddbegin \DIFadd{, and it outperforms general-purpose solver like GAMS/PATH in terms of runtime and number of iterations}\DIFaddend . Results indicate that the introduction of ridesharing can bring network benefits, and ridesharing trips are typically longer than average trip lengths. Sensitivity analysis suggests that a properly designed ridesharing unit price is necessary to achieve network benefits, and travelers with relatively lower values of time are more likely to participate in ridesharing.

The proposed general equilibrium model provides a planning tool for studying optimal pricing and regulatory policymaking problems. Further research is needed to design efficient and effective pricing strategies for transportation systems with shared mobility services.

\newpage
\printbibliography

\newpage
\appendix
\section*{Appendix}
\addcontentsline{toc}{section}{Appendices}
\renewcommand{\thesubsection}{A\Alph{subsection}}
\noindent

\subsubsection{Hyper-network definition}
\label{appendix.hyper_net_def}
In this subsection, we formally define the matching-sequence expanded hyper-network $\mathcal{G}$. Let $\mathcal{G}_{Veh}$ and $\mathcal{G}_{PT}$ denote the original vehicle and public transport networks, respectively; where $\mathcal{G}_{Veh} = (\mathcal{N}_{Veh}, \mathcal{E}_{Veh})$ and $\mathcal{G}_{PT} = (\mathcal{N}_{PT}, \mathcal{E}_{PT})$ with $\mathcal{N}_{Veh}, ~\mathcal{N}_{PT}$ and $\mathcal{E}_{Veh}, ~\mathcal{E}_{PT}$ denoting nodes and edges in $\mathcal{G}_{Veh}$ and $\mathcal{G}_{PT}$; and the set of traveler origins and destinations are denoted as $\mathcal{N}_O$  and $\mathcal{N}_D$, respectively.  The hyper-network $\mathcal{G}$ is composed of four subnetworks $\mathcal{G}_{DA}, ~\mathcal{G}_{RD}, ~\mathcal{G}_{RP}, ~\mathcal{G}_{PT}$ as summarized in \cref{tab:hypernet}:

\begin{table}[H]
    \caption{Components of the hyper-network.}
    \label{tab:hypernet}
    \centering
    \scriptsize
    \setlength\tabcolsep{6pt}
    \renewcommand{\arraystretch}{1.5}
    \begin{tabular}{cccc}
    \toprule
    \textbf{Subnetwork} & \textbf{Notation} & \textbf{Nodes} & \textbf{Links} \\
    \midrule
    \midrule
    Drive alone (DA) & 
      $\mathcal{G}_{DA}$ & 
        $\mathcal{N}_{DA} \coloneqq \mathcal{N}_{Veh}$ & 
          $\mathcal{E}_{DA} \coloneqq \mathcal{E}_{Veh}$ \\
    
    \midrule
    {\multirow{12}{*}{\makecell{Ridesharing driver (RD)}}} 
      & {\multirow{12}{*}{\makecell{$\mathcal{G}_{RD}$}}}
        & \multicolumn{1}{l}{\tabitem RD node}
          & \multicolumn{1}{l}{\tabitem RD links} \\
      & 
        & $\mathcal{N}_{RD} \coloneqq \{\mathcal{N}_{RD}^l\}$  
          & $\mathcal{E}_{RD, 0/1} \coloneqq \{\mathcal{E}_{RD, 0/1}^l \cup \mathcal{E}_{RD, 0/1}^{l,l+1}\}$ \\
      & 
        & \multicolumn{1}{l}{where,} 
          & \multicolumn{1}{p{0.25\linewidth}}{where, 0: without-passenger, and 1: with-passenger}\\
      & 
        & $\mathcal{N}_{RD}^l \coloneqq \mathcal{N}_{Veh}, \forall l \in \{1,...,L\}$ 
          & $\mathcal{E}_{RD,0/1}^l \coloneqq \mathcal{E}_{Veh}, \forall l \in \{1,...,L\}$ \\
      & 
        & \multicolumn{1}{l}{is the RD node at level $l$}
          & \multicolumn{1}{l}{is the RD links within level $l$, and} \\
      & 
        & \multicolumn{1}{l}{\tabitem Virtual RD origin node, $\widehat{O}$} 
          & $\mathcal{E}_{RD,0/1}^{l, l+1} \coloneqq \{(i_{RD}^{l}, i_{RD}^{l+1})\}, \forall l \in \{1,...,L-1\}$ \\
      & 
        & $\mathcal{N}_{RD, \widehat{O}} \coloneqq \mathcal{N}_{O}$
          & \multicolumn{1}{l}{is the virtual links connecting level $l$ and $l+1$} \\
      
      & 
        & 
          & \multicolumn{1}{l}{\tabitem Virtual RD links} \\
      & 
        & 
          & $\mathcal{E}_{\widehat{RD}} \coloneqq \{(i_{\widehat{RD}}, i_{RD}^{1}), (i_{\widehat{RD}}, i_{DA}) \}$ \\
      & 
        & 
          & \multicolumn{1}{l}{where,$(i_{\widehat{RD}}, i_{RD}^{1})$ represents travel as RD} \\
      & 
        & 
          & \multicolumn{1}{l}{$(i_{\widehat{RD}}, i_{DA})$ represents leaving ridesharing} \\
    
    \midrule
    {\multirow{12}{*}{\makecell{Ridesharing passenger (RP)}}} 
      & {\multirow{12}{*}{\makecell{$\mathcal{G}_{RP}$}}}
        & \multicolumn{1}{l}{\tabitem RP node}
          & \multicolumn{1}{l}{\tabitem RP links} \\
      & 
        & $\mathcal{N}_{RP} \coloneqq \{\mathcal{N}_{RP}^l\}$  
          & $\mathcal{E}_{RP} \coloneqq \{\mathcal{E}_{RP}^l \cup \mathcal{E}_{RP}^{l,l+1}\}$ \\
      & 
        & \multicolumn{1}{l}{where,} 
          & \multicolumn{1}{p{0.25\linewidth}}{where,}\\
      & 
        & $\mathcal{N}_{RP}^l \coloneqq \mathcal{N}_{Veh}, \forall l \in \{1,...,L\}$ 
          & $\mathcal{E}_{RP}^l \coloneqq \mathcal{E}_{Veh}, \forall l \in \{1,...,L\}$ \\
      & 
        & \multicolumn{1}{l}{is the RP node at level $l$}
          & \multicolumn{1}{l}{is the RP links within level $l$, and} \\
      & 
        & \multicolumn{1}{l}{\tabitem Virtual RP origin node, $\widehat{O}$} 
          & $\mathcal{E}_{RP}^{l, l+1} \coloneqq \{(i_{RP}^{l}, i_{RP}^{l+1})\}, \forall l \in \{1,...,L-1\}$ \\
      & 
        & $\mathcal{N}_{RP, \widehat{O}} \coloneqq \mathcal{N}_{O}$
          & \multicolumn{1}{l}{is the virtual links connecting level $l$ and $l+1$} \\
      
      & 
        & 
          & \multicolumn{1}{l}{\tabitem Virtual RP links} \\
      & 
        & 
          & $\mathcal{E}_{\widehat{RP}} \coloneqq \{(i_{\widehat{RP}}, i_{RP}^{l}), (i_{\widehat{RP}}, i_{PT}) \}$ \\
      & 
        & 
          & \multicolumn{1}{l}{where, $(i_{\widehat{RP}}, i_{PT})$ represents leaving ridesharing} \\
      & 
        & 
          & \multicolumn{1}{l}{$(i_{\widehat{RP}}, i_{RP}^{l}), \forall l \in \{1, ..., L\}$ represents picked up as $l^{th}$ task} \\
    \midrule
    Public transport (PT) & 
      $\mathcal{G}_{PT}$ & 
        $\mathcal{N}_{PT} \coloneqq \mathcal{N}_{PT}$ & 
          $\mathcal{E}_{PT} \coloneqq \mathcal{E}_{PT}$ \\

    \bottomrule
    \end{tabular}
\end{table}

\begin{proposition}
  \label{prop:equivalence_stable_matching}
  The proposed multi-passenger ridesharing stable matching problem formulation \eqref{eq:RD_stable_route}-\eqref{eq:stable_matching_constraint} is equivalent to the many-to-one stable matching problem of \textcite{peng2022many}, which is a special case of the college admission problem (\cite{gale1962college}).
\end{proposition}
\begin{proof}
(\cite{peng2022many}) \textit{Many-to-one stable matching problem}

Given a set of driver-passenger groups $b_i \in \mathcal{B}$, a set of drivers $t \in T$, and a set of passengers $p \in P$, decision variable $y_i=1$ if driver-passenger group $b_i$ is matched, and $y_i=0$ otherwise. The many-to-one stable matching problem is formulated as follows:
\begin{subequations}
  \begin{alignat}{2}
    &\max_{y}       & \qquad  & H(y)  \nonumber \\
    &\textrm{subject to}&       & \nonumber \\
    &           &       & \sum_{i:t \in b_i} {y_i} =1, \forall t \in T \label{eq:peng_T_match} \\
    &           &       & \sum_{i:p \in b_i} {y_i} =1, \forall p \in P \label{eq:peng_P_match} \\
    &           &       & y_i 
                      + \sum_{t \in b_i} {\sum_{b_j \succ_{t} b_i} {y_j}} 
                      + \sum_{p \in b_i} {\sum_{b_k \succ_{t} b_i} {y_k}} 
                      \geq 1, \forall b_i \in \mathcal{B}
                      \label{eq:peng_stable_constraint} \\
    &           &       & y_i \in \{0, 1\}, \forall b_i \in \mathcal{B}
\end{alignat}
\end{subequations}

\cref{eq:peng_T_match}-\eqref{eq:peng_P_match} state that, driver and passenger are matched to exactly one driver-passenger group $b_i$, and if $y_i = 1$, $\sum_{t \in b_i} {\sum_{b_j \succ_{t} b_i} {y_j}} =0$, and $\sum_{p \in b_i} {\sum_{b_k \succ_{t} b_i} {y_k}}=0$ (\cref{eq:peng_stable_constraint}). This is to say, for driver $t$, $b_i \succeq_{t} b_j, \forall j : t\in b_j$, which implies the disutility of $b_i$ is the minimum, $\pi_t^{i} \leq \pi_t^{j}, \forall j : t\in b_j$; and similarly, $\pi_p^{i} \leq \pi_p^{k}, \forall k : p\in b_j$. Such relation can be written as complementarity conditions as in proposed multi-passenger ridesharing stable matching problem formulation \eqref{eq:RD_stable_route}-\eqref{eq:RP_stable_route}. By equating \eqref{eq:peng_T_match} with \eqref{eq:peng_P_match} ($\sum_{i:t \in b_i} {y_i}=\sum_{i:p \in b_i} {y_i}$), we also obtain constraint \eqref{eq:stable_matching_constraint} in our formulation.

Suppose $x_{\widehat{i_{RD}}, i^1}^{n, RD}$ and $x_{\widehat{i_{RP}}, i^l}^{n, e, RP}$ are discretized as follows:
\begin{subequations}
  \begin{alignat}{2}
    & x_{\widehat{i_{RD}}, i^1}^{n, RD} = \sum_{0}^{x_{\widehat{i_{RD}}, i^1}^{n, RD}} {\sum_{n}{y_n}} & \label{eq:eqv_RD_disc}\\
    & x_{\widehat{i_{RP}}, i^l}^{n, e, RP} = \sum_{0}^{x_{\widehat{i_{RP}}, i^l}^{n, e, RP}} {\sum_{n}{y_n}} & \label{eq:eqv_RP_disc}\\
    & \sum_{n} {y_n}=1 & \label{eq:discrete}\\
    & y_n \in \{0, 1\}, \forall n &
  \end{alignat}
\end{subequations}
It is easy to see that, \cref{eq:eqv_RD_disc}-\eqref{eq:discrete} are equivalent to \eqref{eq:peng_T_match}-\eqref{eq:peng_P_match}. By complementarity conditions \eqref{eq:RD_stable_route}-\eqref{eq:RP_stable_route} (assume virtual link costs are zero), if $x_{\widehat{i_{RD}}, i^1}^{n, RD} > 0$,  $\pi_{i^1}^{n, RD} \leq \pi_{i^1}^{n', RD}, \forall n' : s_1(n',0,(i, \cdot))=1$; and similarly, $\pi_{i^l}^{n, e, RP} \leq pi_{i^{l'}}^{n', e, RP}, \forall (n',l') : s_1(n',l',(i,e))=1$. Together by \cref{eq:eqv_RD_disc}-\eqref{eq:discrete} and the complementarity conditions, we have that, if $y_n=1$, $n \succeq_{RD} n', \forall n' : s_1(n',0,(i, \cdot))=1$ and $(n, l) \succeq_{RP} (n',l'), \forall (n',l') : s_1(n',l',(i,e))=1$. 

Therefore, if $y_n=1$, 
\begin{flalign}
  \sum_{i \in \{i:s_1(n,0,(i, \cdot))=1\}} {\sum_{n' \succ_{RD_i} n} {y_{n'}}} =0 \nonumber \\
  \sum_{(i, e) \in \{(i, e)| s_1(n,l,(i,e))=1\}} {\sum_{(n', l') \succ_{t} (n, l)} {y_{n'}}}=0 \nonumber
\end{flalign}

\end{proof}

\subsubsection{MCP formulation for the joint stable matching and route choice problem}
\label{appendix.MCP_network_model}

\textbf{Stable matching multipliers}\; Let $\mu_{3}^{i, n}$ denote the multiplier for constraint \eqref{eq:RD_matching_cap}, $\mu_{4}^{i, e, n, l}$ denote the multiplier for constraint \eqref{eq:RP_matching_cap}, $\phi_{5}^{n, l},\widehat{\phi}_{5}^{n, l} $ denote the multiplier for constraint \eqref{eq:stable_matching_constraint} for RD, and for RP, respectively; and the set of passengers in matching sequence $n$, traveling between $(j, e)$ and being picked up as the $l^{th}$ task, denoted as $\zeta_{RP}^{n} \coloneqq \{(j, l, e)|s_1(n, l, (j, e)) = 1, 1 \leq l \leq L - 1\}$. 

\textbf{Route choice multipliers}\; Let $\pi_{i}^{DA}$ denote the dual variables (node potential) of DA flow conservation at node $i_{DA}$, and $\pi_{i^{l}}^{n, RD}$ denote the dual variables of RD conservation at node $i_{RD}^{l}$ for matching sequence $n$, and similarly $\pi_{i^{l}}^{n, e, RP}$ denote the dual variables of RP conservation at node $i_{RP}^{l}$ for matching sequence $n$ and RP with destination $e$, and $\pi_{i}^{PT}$ denote the dual variables of PT conservation at node $i_{PT}$.

\textbf{Multi-passenger ridesharing multipliers}\; We denote the dual variables for the coupling constraint \eqref{eq:RD_RP_coupling_upper} as $\lambda_{6}^{i^l, j^{l'}, n}$, and $\widehat{\lambda}_{6}^{i^l, j^{l'},n}$; and $\lambda_{7}^{i^l, j^{l'}, n}$, and $\widehat{\lambda}_{7}^{i^l, j^{l'}, n}$ for the coupling constraint \eqref{eq:RD_RP_coupling_lower}, for ridesharing drivers and passengers, respectively. Let $\lambda_{8}^{i^l, n}$, $\lambda_{9}^{i^l, n}$, $\lambda_{10}^{i^l, n}$, $\lambda_{11}^{i^l, n}$ denote the dual variables for constraints \eqref{eq:intermediate_conservation_dropoff},\eqref{eq:intermediate_conservation_pickup}, \eqref{eq:RD_no_pickup_remain}, and \eqref{eq:RD_pickup_cross}, respectively; and let $\lambda_{8}^{i^1, n}=0$, and $\lambda_{9}^{i^1, n}=0$ for completeness.

\noindent
[MCP - Joint stable matching and route choice model]
\\\textbf{Multi-passenger ridesharing constraints}

\textit{Ridesharing driver-passenger coupling constraints}
\begin{subequations}
\label{eq:mcp_group_coupling_constraints}
    \begin{alignat}{2}
    0 \leq \left[  
            \left( 
              \sum_{\omega}  {\sum_{l'' \leq l' - 1} \left(s_1(n, l'', \omega) - s_{-1}(n, l'', \omega)\right) } - 1 
            \right)
            \cdot 
              x_{i^{l}, j^{l'}, 1}^{n, RD}
            - \sum_{e \in \mathcal{D}}{x_{i^{l}, j^{l'}}^{n, e, RP}}
      \right] 
    & \perp 
            \lambda_{6}^{i^l, j^{l'}, n}
    \geq 0, \label{eq:mcp_group_coupling_1} \\
             &\forall i, n, l, j^{l'}:(i^l, j^{l'}) \in \mathcal{E}_{RD, 1} \nonumber\\
    0 \leq \left[  
            \sum_{e \in \mathcal{D}}{x_{i^{l}, j^{l'}}^{n, e, RP}}
            - x_{i^{l}, j^{l'}, 1}^{n, RD}
      \right] 
     \perp 
            \lambda_{7}^{i^l, j^{l'},n}
    \geq 0, 
             \forall i, n, l, j^{l'}:(i^l, j^{l'}) \in \mathcal{E}_{RD, 1}\label{eq:mcp_group_coupling_2} &.\\
    0 \leq \left[  
            \left( 
              \sum_{\omega}  {\sum_{l'' \leq l' - 1} \left(s_1(n, l'', \omega) - s_{-1}(n, l'', \omega)\right) } - 1 
            \right)
            \cdot 
              x_{i^{l}, j^{l'}, 1}^{n, RD}
            - \sum_{e \in \mathcal{D}}{x_{i^{l}, j^{l'}}^{n, e, RP}}
      \right] 
    & \perp 
            \widehat{\lambda}_{6}^{i^l, j^{l'},n}
    \geq 0, \label{eq:mcp_group_coupling_1_hat}\\
             &\forall i, n, l, j^{l'}:(i^l, j^{l'}) \in \mathcal{E}_{RD, 1} \nonumber\\
    0 \leq\left[  
            \sum_{e \in \mathcal{D}}{x_{i^{l}, j^{l'}}^{n, e, RP}}
            - x_{i^{l}, j^{l'}, 1}^{n, RD}
      \right] 
     \perp 
            \widehat{\lambda}_{7}^{i^l, j^{l'}, n}
    \geq 0, 
             \forall i, n, l, j^{l'}:(i^l, j^{l'}) \in \mathcal{E}_{RD, 1} \label{eq:mcp_group_coupling_2_hat} &.
\end{alignat}
\end{subequations}

\textit{Conservation of with-passenger RD flow at intermediate nodes}
\begin{subequations}
\label{eq:mcp_group_intermediate_conservation}
    \begin{alignat}{2}
    \begin{split}
      0 \leq & \left[ 
              \left(
                  \sum_{o \in \mathcal{O}}{s_{-1}(n, l, (o, i))}  
              \right)
              \cdot \left(
                \sum_{k^{l''}:(k^{l''}, i^{l}) \in \mathcal{E}_{RD, 0}} {x_{k^{l''}, i^{l}, 0}^{n, RD}}
                + \sum_{k^{l''}:(k^{l''}, i^{l}) \in \mathcal{E}_{RD, 1}} {x_{k^{l''}, i^{l}, 1}^{n, RD}}
                \right) \right. \\
            & \left.
            - \sum_{k^{l''}:(k^{l''}, i^{l}) \in \mathcal{E}_{RD, 1}} {x_{k^{l''}, i^{l}, 1}^{n, RD}}
            + \sum_{j^{l'}:(i^{l}, j^{l'}) \in \mathcal{E}_{RD, 1}} {x_{i^{l}, j^{l'}, 1}^{n, RD}}
        \right]
        \perp 
          \lambda_{8}^{i^l, n}
    , \forall i, n, l \geq 2
    \end{split} \label{eq:mcp_group_intermediate_conservation_1} \\
    \begin{split}
      0 \leq & \left[ 
              \left(
                  \sum_{o \in \mathcal{O}}{s_{1}(n, l, (o, i))}  
              \right)
              \cdot \left(
                \sum_{k^{l''}:(k^{l''}, i^{l}) \in \mathcal{E}_{RD, 0}} {x_{k^{l''}, i^{l}, 0}^{n, RD}}
                + \sum_{k^{l''}:(k^{l''}, i^{l}) \in \mathcal{E}_{RD, 1}} {x_{k^{l''}, i^{l}, 1}^{n, RD}}
                \right) \right. \\
            & \left.
            - \sum_{j^{l'}:(i^{l}, j^{l'}) \in \mathcal{E}_{RD, 1}} {x_{i^{l}, j^{l'}, 1}^{n, RD}}
            + \sum_{k^{l''}:(k^{l''}, i^{l}) \in \mathcal{E}_{RD, 1}} {x_{k^{l''}, i^{l}, 1}^{n, RD}}
        \right]
        \perp 
          \lambda_{9}^{i^l, n}
    , \forall i, n, l \geq 2
    \end{split} \label{eq:mcp_group_intermediate_conservation_2}
    \end{alignat}
\end{subequations}

\textit{Matching sequence intra-task constraint}
\begin{equation}    
\label{eq:mcp_group_no_pickup_remain}        
    \begin{split}
  0 \leq & \left[
             \left(
              \sum_{e\in \mathcal{D}} {s_1(n,l,(i, e))} 
              + \sum_{o\in \mathcal{O}} {s_{-1}(n,l,(o, i))}
            \right) \right. \\
            {}&\cdot \left( 
              x_{\widehat{i_{RD}}, i^{l}}^{n, RD}
              + \sum_{k^{l''}:(k^{l''}, i^{l}) \in \mathcal{E}_{RD, 0}} {x_{k^{l''}, i^{l}, 0}^{n, RD}}
              + \sum_{k^{l''}:(k^{l''}, i^{l}) \in \mathcal{E}_{RD, 1}} {x_{k^{l''}, i^{l}, 1}^{n, RD}} 
            \right) \\
            &\left.  - \sum_{j^{l+1}:(i^{l}, j^{l+1}) \in \mathcal{E}_{RD, 0}^{l, l+1}} {x_{i^{l}, j^{l+1}, 0}^{n, RD}}
            - \sum_{j^{l+1}:(i^{l}, j^{l+1}) \in \mathcal{E}_{RD, 1}^{l, l+1}} {x_{i^{l}, j^{l+1}, 1}^{n, RD}}   
        \right]
      \perp
            \lambda_{10}^{i^l, n}
      \geq 0
      ,  \forall i, n, l
  \end{split}
\end{equation}

\textit{Matching sequence inter-task constraint}
\begin{equation}  
\label{eq:mcp_group_pickup_cross}          
    \begin{split}
  0 = \left(
      \sum_{j^{l}:(i^{l}, j^{l}) \in \mathcal{E}_{RD, 0}^{l, l}} {x_{i^{l}, j^{l}, 0}^{n, RD}} \right.
        +& \left. \sum_{j^{l}:(i^{l}, j^{l}) \in \mathcal{E}_{RD, 1}^{l, l}} {x_{i^{l}, j^{l}, 1}^{n, RD}}
      \right) \\
      {}& \cdot
        \left(
        \sum_{e\in \mathcal{D}} {s_1(n,l,(i, e))}
        + \sum_{o\in \mathcal{O}} {s_{-1}(n,l,(o, i))}
       \right)
      \perp
            \lambda_{11}^{i^l, n}
      \quad \text{free}
      ,  \forall i, n, l
  \end{split}
\end{equation}
\\\textbf{Route choice model}
\\Drive alone (DA) route choice:
\begin{flalign} 
\label{eq:mcp_DA_route_choice_final}
    0 \leq \left[\pi_{j}^{DA} + {c}_{i, j}^{DA} - \pi_{i}^{DA}\right] \perp x_{i, j}^{e, DA} \geq 0
    \quad \quad , \forall (i, j) \in \mathcal{E}_{DA}, e \in \mathcal{D}
\end{flalign} 
Ridesharing driver (RD) route choice:

\textit{Ridesharing without passenger}
\begin{flalign} 
\label{eq:mcp_RD_route_choice_empty_final}
\begin{split}
    0 \leq & \left[
            \pi_{j^{l'}}^{n, RD} 
            + {c}_{i^{l},j^{l'}, 0}^{RD} \right. \\
            {}& - \sum_{o \in \mathcal{O}} {s_{-1}(n, l', (o, j))} \cdot \lambda_{8}^{j^{l'}, n} \text{~~Avoid intermediate node drop-off}\\
            {}& - \sum_{o \in \mathcal{O}} {s_{1}(n, l', (o, j))} \cdot \lambda_{9}^{j^{l'}, n} \text{~~Avoid intermediate node drop-off}\\
            {}& + \lambda_{10}^{i^l, n} - \left(
                              \sum_{e \in \mathcal{D}} {s_{1}(n, {l'}, (j, e))} 
                              + \sum_{o \in \mathcal{O}} {s_{-1}(n, {l'}, (o, j))}
                             \right) 
                             \cdot 
                                \lambda_{10}^{j^{l'}, n} \text{~~Ensure pickup/drop-off execution} \\
            {}&  + \left(
        \sum_{e\in \mathcal{D}} {s_1(n,l,(i, e))}
        + \sum_{o\in \mathcal{O}} {s_{-1}(n,l,(o, i))}
       \right) \cdot \lambda_{11}^{i^l, n} \text{~~Ensure execution follows matching sequence} \\
            {}& \left. 
            - \pi_{i^{l}}^{n, RD}
        \right] 
    \perp 
            x_{i^{l}, j^{l'}, 0}^{n, RD} 
    \geq 0
    \quad \quad , \forall i, n, l, j^{l'}:(i^{l}, j^{l'}) \in \mathcal{E}_{RD, 0}
\end{split}
\end{flalign}    

\textit{Ridesharing with passenger}
\begin{flalign} 
\label{eq:mcp_RD_route_choice_onboard_final}
\begin{split}
    0 \leq & \left[
            \pi_{j^{l'}}^{n, RD} 
            + {c}_{i^{l},j^{l'}, 1}^{RD} \right.\\
            {}&
            - \left( 
              \sum_{\omega}  {\sum_{l'' \leq l' - 1} \left(s_1(n, l'', \omega) - s_{-1}(n, l'', \omega)\right) } - 1 
              \right)
              \cdot 
                \lambda_{6}^{i^l, j^{l'}, n} 
            + \lambda_{7}^{i^l, j^{l'}, n} \text{~~RD-RP coupling} \\
            {}& - \lambda_{8}^{i^l, n} + (1 - \sum_{o \in \mathcal{O}} {s_{-1}(n, {l'}, (o, j))}) \cdot \lambda_{8}^{j^{l'}, n} \text{~~Avoid intermediate node drop-off}\\
            {}& + \lambda_{9}^{i^l, n} - (1 + \sum_{o \in \mathcal{O}} {s_{1}(n, {l'}, (o, j))}) \cdot \lambda_{9}^{j^{l'}, n} \text{~~Avoid intermediate node drop-off}\\
            {}& + \lambda_{10}^{i^l, n} - \left(
                          \sum_{e \in \mathcal{D}} {s_{1}(n, {l'}, (j, e))} 
                          + \sum_{o \in \mathcal{O}} {s_{-1}(n, {l'}, (o, j))}
                        \right) \cdot \lambda_{10}^{j^{l'}, n} \text{~~Ensure pickup/drop-off execution} \\
            {}&  + \left(
        \sum_{e\in \mathcal{D}} {s_1(n,l,(i, e))}
        + \sum_{o\in \mathcal{O}} {s_{-1}(n,l,(o, i))}
       \right) \cdot \lambda_{11}^{i^l, n} \text{~~Ensure execution follows matching sequence} \\
            {}& \left. - \pi_{i^{l}}^{n, RD}
        \right]
    \perp 
            x_{i^{l}, j^{l'}, 1}^{n, RD} 
    \geq 0
    \quad \quad , \forall i, n, l, j^{l'}:(i^{l}, j^{l'}) \in \mathcal{E}_{RD, 1}
\end{split}
\end{flalign} 
Ridesharing passenger (RP) route choice:
\begin{flalign} 
\label{eq:mcp_RP_route_choice_final}
\begin{split}
    0 \leq & \left[
            \pi_{j^{l'}}^{n, e, RP} 
            + c_{i^{l}, j^{l'}}^{RP} \right. \\
            {}&+ \widehat{\lambda}_{6}^{i^l, j^{l'}, n} - \widehat{\lambda}_{7}^{i^l, j^{l'}, n} \text{~~RD-RP coupling} \\
            {}& \left. 
            - \pi_{i^{l}}^{n, e, RP}
        \right] 
    \perp 
            x_{i^{l}, j^{l'}}^{n, e, RP} 
    \geq 0
    \quad \quad , \forall i, n, l, j^{l'}:(i^{l}, j^{l'}) \in \mathcal{E}_{RP}, e \in \mathcal{D}
\end{split}
\end{flalign}
Public transport passenger (PT) route choice:
\begin{flalign} 
\label{eq:mcp_PT_route_choice_final}
    0 \leq \left[\pi_{j}^{PT} + {c}_{i, j}^{PT} - \pi_{i}^{PT}\right] \perp x_{i, j}^{e, PT} \geq 0
    \quad \quad , \forall (i, j) \in \mathcal{E}_{PT}, e \in \mathcal{D}
\end{flalign} 

\noindent
\textbf{Stable matching}

\textit{For ridesharing driver:}
\begin{subequations}
\label{eq:mcP_RD_stable_route_final}
    \begin{alignat}{2}
    0 \leq \left[  
            \pi_{i^1}^{n, RD}
            + \mu_{3}^{i, n}
            + \sum_{l} {\phi_{5}^{n, l}}
            - 
        \sum_{e\in \mathcal{D}} {s_1(n,1,(i, e))}
       \cdot \lambda_{10}^{i^1, n}
            - \pi_{\widehat{i_{RD}}}^{e, RD}  
      \right] 
    & \perp 
            x_{\widehat{i_{RD}}, i^1}^{n, RD}
    \geq 0, 
            & \forall i, n  \label{eq:mcp_RD_seq_choice_final}  \\
    0 \leq \left[  
            \pi_{i}^{e, DA} 
            - \pi_{\widehat{i_{RD}}}^{e, RD}  
      \right] 
    & \perp 
            x_{\widehat{i_{RD}}, i_{DA}}^{e, RD} 
    \geq 0, 
            & \forall i, e \in \mathcal{D}  \label{eq:mcp_RD_to_DA_final} \\
    0 \leq \left[ 
            \sum_{n\in \{n|s_1(n,0,(i,e))=1\}} {x_{\widehat{i_{RD}}, i^1}^{n, RD}}
            + x_{\widehat{i_{RD}}, i_{DA}}^{e, RD}
            - q_{(i,e)}^{RD}
      \right] 
    & \perp 
            \pi_{\widehat{i_{RD}}}^{e, RD}
    \geq 0, 
            & \forall i, e \in \mathcal{D}  \label{eq:mcp_RD_conservation_final} \\ 
    x_{\widehat{i_{RD}}, i^1}^{n, RD} = x_{\widehat{j_{RP}}, j^l}^{n, e, RP} & \perp  \phi_{5}^{n, l} \text{~free}, 
             \begin{split}{}&\forall n,  \\
               &i:s_1(n, 0, (i, \cdot))=1,          \\
               &(j, l, e) \in \zeta_{RP}^{n}       
    \end{split}  
             \label{eq:mcp_RD_stable_matching_final} \\
    0 \leq \left[ 
            s_1(n, 0, (i, \cdot)) \cdot Z_n
            - x_{\widehat{i_{RD}}, i^1}^{n, RD}
      \right] 
    & \perp 
            \mu_{3}^{i, n}
    \geq 0, 
            & \forall n, i \in O  \label{eq:mcp_RD_matching_cap_final}
    \end{alignat}
\end{subequations}
\textit{For ridesharing passenger:}
\begin{subequations}
\label{eq:mcP_RP_stable_route_final}
    \begin{alignat}{2}
    0 \leq \left[  
            \pi_{i^l}^{n, e, RP} 
            + \mu_{4}^{i, e, n, l}
            - \widehat{\phi}_{5}^{n, l}
              - \pi_{\widehat{i_{RP}}}^{e, RP}
      \right] 
    & \perp 
            x_{\widehat{i_{RP}}, i^l}^{n, e, RP}
    \geq 0, 
            & \forall i, e, n, l  \label{eq:mcp_RP_seq_choice_final}  \\
    0 \leq \left[  
            \pi_{i}^{e, PT} 
            - \pi_{\widehat{i_{RP}}}^{e, RP}  
      \right] 
    & \perp 
            x_{\widehat{i_{RP}}, i_{PT}}^{e, RP} 
    \geq 0, 
            & \forall i, e \in \mathcal{D}  \label{eq:mcp_RP_to_PT_final} \\
    0 \leq \left[
            \sum_{n} {\sum_{1 \leq l \leq L-1} {x_{\widehat{i_{RP}}, i^l}^{n, e, RP}}} 
            + x_{\widehat{i_{RP}}, i_{PT}}^{e, RP}
            - q_{(i,e)}^{RP}
      \right] 
    & \perp 
            \pi_{\widehat{i_{RP}}}^{e, RP}
    \geq 0, 
            & \forall i, e \in \mathcal{D} \label{eq:mcp_RP_conservation_final} \\ 
    x_{\widehat{i_{RD}}, i^1}^{n, RD} = x_{\widehat{j_{RP}}, j^l}^{n, e, RP} & \perp  \widehat{\phi}_{5}^{n, l} \text{~free}, \begin{split}{}&\forall n,  \\
                  &i:s_1(n, 0, (i, \cdot))=1,         \\
                  &(j, l, e) \in \zeta_{RP}^{n}     
    \end{split}  
             \label{eq:mcp_RP_stable_matching_final} \\
    0 \leq \left[ 
            s_1(n, l, (i, e)) \cdot Z_n
            - x_{\widehat{i_{RP}}, i^l}^{n, e, RP}
      \right] 
    & \perp 
            \mu_{4}^{i, e, n, l}
    \geq 0, 
            & \forall n, (i, e): \omega, 1 \leq l \leq L-1  \label{eq:mcp_RP_matching_cap_final}
    \end{alignat}
\end{subequations}

Note that, multipliers for the \textit{multi-passenger ridesharing constraints} \eqref{eq:mcp_group_coupling_constraints}-\eqref{eq:mcp_group_pickup_cross} are related to routing decisions, which are included in the route choice mode \eqref{eq:mcp_DA_route_choice_final}-\eqref{eq:mcp_PT_route_choice_final}. Moreover, multiplier $\lambda_{10}^{i^1, n}$ for matching sequence intra-task constraint \eqref{eq:mcp_group_no_pickup_remain} is also incorporated into the ridesharing driver stable matching model \eqref{eq:mcp_RD_seq_choice_final}. $\lambda_{10}^{i^1, n}$ can be interpreted as the additional incentives for drivers to pick up passengers from the driver's origin. For drivers departing from node $i$, if there are passengers to be picked up at node $i$ as the $1^{st}$ task (i.e., $\sum_{e \in \mathcal{D}} {s_{1}(n, l, (i, e))} = 1$), all the flows enter to level $2$ at node $i^1$, due to Constraint \eqref{eq:RD_pickup_cross} and Constraint \eqref{eq:flow_conservation}.

The dual variables of the complementarity constraints \eqref{eq:mcp_RD_matching_cap_final} and \eqref{eq:mcp_RP_matching_cap_final} are typically interpreted as the "demand price" or "surge price" in the ridesharing literature. When there is a surplus of matching, both ridesharing drivers and passengers need not to wait; when the RD/RP demand is higher than the matching (which is jointly determined by the travelers and platform), drivers and passengers might experience matching waiting, or pay a surge pee for a timely ridesharing service. The un-signed multipliers for the stable matching constraints \eqref{eq:stable_matching_constraint}, $\phi_{5}^{n, l}$ and $\widehat{\phi}_{5}^{n, l}$, might be interpreted as the willingness-to-share of the the ridesharing travelers. If $\phi_{5}^{n, l}>0$, the more passengers in a matching sequence, the less likely the driver is willing to share a ride (e.g., considered too cumbersome). If $\phi_{5}^{n, l}<0$, it might suggest the driver is happy to share their ride with more passengers (e.g., picking up passengers from the same OD), $\phi_{5}^{n, l}=0$ might indicate the driver is indifferent to share or not share.

\subsubsection{Solution existence}
\label{sec:Existence}
The solution existence for the proposed general equilibrium are established in 3 main steps: first, the proposed general equilibrium model is cast into a VI with penalization of the asymmetric shared constraints \eqref{eq:RD_demand_constraint}-\eqref{eq:RP_demand_constraint}, and \eqref{eq:RD_matching_cap}-\eqref{eq:RP_matching_cap}; then, the solution existence is established for the VI problem with penalization; last, the proposed general equilibrium model is shown to be recovered when the penalty tending to infinity. Note that, compared to the solution existence for the path-based model in \textcite{ban2019general}, we extend the penalty-based method to the link-node formulation by using VI existence theorem for unbounded sets (\cite{Nagurney2009}). The following subsections correspond to the 3 main steps mentioned above.

\noindent
\textbf{VI formulation with penalization of the asymmetric shared constraints}

As outlined above, the proof of solution starts from formulating a VI with penalization of the asymmetric constraints for the general equilibrium model \eqref{eq:mcp_mode_choice}-\eqref{eq:mcp_mode_conservation}, \eqref{eq:mcp_matching_obj}-\eqref{eq:mcp_RP_constraint}, and \eqref{eq:mcp_group_coupling_constraints}-\eqref{eq:mcP_RP_stable_route_final}. Let $\mathbf{\rho} > 0$ denote the penalty factors, the following vector function $\mathbf{F}^{\rho}$ correspond to a set of penalty factors $\mathbf{\rho}$:

\begin{flalign}
\label{VI_1}
  \begin{split}
    &\mathbf{F}^{\rho}(\mathbf{w}) \triangleq \\
    & \begin{pmatrix*}[l]
      & \textbf{Mode choice} & & 
      \\
      & C_{\omega}^{m} - \pi_{\omega} & : & (\omega, m) \in \mathcal{W} \times \mathcal{M} \\
      \midrule
      & \textbf{Ridesharing matching} & & 
      \\
      & -R_{n} & : & n 
      \\[1.5ex]
      & \underbrace{\rho_1 \cdot \left(Y_{\omega}^{RD} - q_{\omega}^{RD} \right)}_{\text{Penalty term}} & : & \omega \in \mathcal{W} 
      \\[1.5ex]
      & \underbrace{\rho_2 \cdot \left(Y_{\omega}^{RP} - q_{\omega}^{RP} \right)}_{\text{Penalty term}} & : & \omega \in \mathcal{W} 
      \\
      \midrule
      & \textbf{Stable matching} & & 
      \\
      & \pi_{i}^{e, DA} - \pi_{\widehat{i_{RD}}}^{e, RD} & : & i, e \in \mathcal{D} 
      \\[1.5ex]
      & \underbrace{\eta_{5}^{n} \cdot \left(\pi_{i^{1}}^{n, RD} - \pi_{\widehat{i_{RD}}}^{e, RD} \right)}_{\text{Normalized term}} 
        + \underbrace{\rho_3 \cdot
            \left[
                  x_{\widehat{i_{RD}}, i^1}^{n, RD} 
                  - s_1(n, 0 (i, \cdot)) \cdot Z_n
            \right]}_{\text{Penalty term}}
      & : & n, i \in \mathcal{O} 
      \\[1.5ex]
      & \pi_{i^{l}}^{n, e, RP} - \pi_{\widehat{i_{RP}}}^{e, RP}
        + \underbrace{\rho_4 \cdot
            \left[
                  x_{\widehat{i_{RP}}, i^l}^{n, e, RP} 
                  - s_1(n, l (i, \cdot)) \cdot Z_n
            \right]}_{\text{Penalty term}}
      & : & n, (i, e) \in \mathcal{W}, 1 \leq l \leq L -1 
      \\[1.5ex]
      & \pi_{i}^{e, PT} - \pi_{\widehat{i_{RP}}}^{e, RP} & : & i, e \in \mathcal{D} 
      \\[1.5ex]
      \midrule
      & \textbf{Route choice} & & 
      \\
      & \pi_{j}^{e, DA} + c_{i, j}^{DA}- \pi_{i}^{e, DA} & : & (i,j) \in \mathcal{E}_{DA}, e \in \mathcal{D} 
      \\[1.5ex]
      & \underbrace{\eta_{67}^{i^l,j^{l'},n} \cdot \left(\pi_{j^{l'}}^{n, RD} + c_{i^{l}, j^{l'}, 0/1}^{RD}- \pi_{i^{l}}^{n, RD}\right)}_{\text{Normalized term}} & : & i, n, l, j^{l'}:(i^{l},j^{l'}) \in \mathcal{E}_{RD, 0/1} 
      \\[1.5ex]
      & \pi_{j^{l'}}^{n, e, RP} + c_{i^{l}, j^{l'}}^{RP}- \pi_{i^{l}}^{n, e, RP} & : & i, n, l, j^{l'}:(i^{l},j^{l'}) \in \mathcal{E}_{RP}, e \in \mathcal{D} 
      \\
      & \pi_{j}^{e, PT} + c_{i, j}^{PT}- \pi_{i}^{e, PT} & : & (i,j) \in \mathcal{E}_{PT}, e \in \mathcal{D} 
      \\[1.5ex]
      &\textit{Conservation constraints} & & 
      \\
      & \sum_{m \in \mathcal{M}} {q_{\omega}^{m}} - q_{\omega} & : & \forall \omega \in \mathcal{W}
      \\[1.5ex]
      & \sum_{n} {x_{\widehat{i_{RD}}, i^1}^{n, RD}} 
        + x_{\widehat{i_{RD}}, i_{DA}}^{e, RD}
        - q_{(i,e)}^{RD} & : & \forall (i, e) \in \mathcal{W}
      \\[1.5ex]
      & \sum_{n} {\sum_{1 \leq l \leq L-1} {x_{\widehat{i_{RP}}, i^l}^{n, e, RP}}} 
            + x_{\widehat{i_{RP}}, i_{PT}}^{e, RP}
            - q_{(i,e)}^{RP} & : & \forall (i, e) \in \mathcal{W}
      \\[1.5ex]
      &\sum_{j:(i,j)\in \mathcal{E}_{DA}} {x_{i, j}^{e, DA}} 
        - q_{i,e}^{DA} 
        - x_{\widehat{i_{RD}}, i_{DA}}^{e, RD} 
        - \sum_{k:(k,i)\in \mathcal{E}_{DA}} {x_{k, i}^{e, DA}} & : & \forall i, e \in \mathcal{D} 
      \\[1.5ex]
      &
        \sum_{j^{l'}:(i^{l}, j^{l'}) \in \mathcal{E}_{RD, 0}} {x_{i^{l}, j^{l'}, 0}^{n, RD}}
        + \sum_{j^{l'}:(i^{l}, j^{l'}) \in \mathcal{E}_{RD, 1}} {x_{i^{l}, j^{l'}, 1}^{n, RD}} 
        - x_{\widehat{i_{RD}}, i^1}^{n, RD}
      & :
      & \forall i, n, l 
      \\[1ex]
      & ~~ - \sum_{k^{l''}:(k^{l''}, i^{l}) \in \mathcal{E}_{RD, 0}} {x_{k^{l''}, i^{l}, 0}^{n, RD}}
        - \sum_{k^{l''}:(k^{l''}, i^{l}) \in \mathcal{E}_{RD, 1}} {x_{k^{l''}, i^{l}, 1}^{n, RD}}
      &
      & 
      \\[1.5ex]
      & \sum_{j^{l'}:(i^{l}, j^{l'}) \in \mathcal{E}_{RP}} {x_{i^{l}, j^{l'}}^{n, e, RP}}
        - \sum_{k^{l''}:(k^{l''}, i^{l}) \in \mathcal{E}_{RP}} {x_{k^{l''}, i^{l}}^{n, e, RP}}
        - x_{\widehat{i_{RP}}, i^l}^{n, e, RP}
      & :
      & \forall i, e \in \mathcal{D} , n, l 
      \\[1.5ex]
      & \sum_{j:(i,j)\in \mathcal{E}_{PT}} {x_{i, j}^{e, PT}}  
        - q_{i,e}^{PT} 
        - x_{\widehat{i_{RP}}, i_{PT}}^{e, RP} 
        - \sum_{k:(k,i)\in \mathcal{E}_{PT}} {x_{k, i}^{e, PT}} 
      & :
      & \forall i, e \in \mathcal{D}
    \end{pmatrix*}
  \end{split}
\end{flalign}
where $\mathbf{w} \triangleq (\mathbf{q} ,\mathbf{Z} ,\mathbf{Y} ,\mathbf{x}, \mathbf{\pi})\geq 0$ is contained in the set $\Omega$, and satisfying:

\begin{align}
  & \textbf{Ridesharing matching} & \nonumber \\ 
  & \underbrace{
    \sum_{n} {s_1{(n,0,\omega)} \cdot Z_n} = Y_{\omega}^{RD}
    }_{\text{Definitional constraint}} 
    & ,\forall \omega 
    \label{eq:proof1_def_constraint1} 
    \\
  & \underbrace{
    \sum_{n} {\sum_{1 \leq l \leq L-1} {s_1{(n,l,\omega)} \cdot Z_n} = Y_{\omega}^{RP}}
    }_{\text{Definitional constraint}} 
    & ,\forall \omega 
    \label{eq:proof1_def_constraint2} 
    \\[2ex]
  & \textbf{Stable matching constraint }
    & \text{\eqref{eq:stable_matching_constraint}}
    \nonumber \\
  & \textbf{Ridesharing driver-passenger coupling constraints }
    &  \text{\eqref{eq:RD_RP_coupling_upper}-\eqref{eq:RD_RP_coupling_lower}}
    \nonumber \\
  & \textbf{Ridesharing passenger-transfer avoidance constraints }
    &  \text{\eqref{eq:intermediate_conservation_dropoff}-\eqref{eq:RD_pickup_cross}}
    \nonumber
\end{align}
Note that, in the penalized VI problem \eqref{VI_1} $VI(\mathbf{F^{\rho}}, \Omega)$, the remaining coupling shared constraints are: stable matching constraints \eqref{eq:stable_matching_constraint}; RD-RP coupling constraints \eqref{eq:RD_RP_coupling_upper}-\eqref{eq:RD_RP_coupling_lower}; and the definitional constraints \eqref{eq:proof1_def_constraint1}-\eqref{eq:proof1_def_constraint2}. While the asymmetric shared constraints: RD and RP demand constraints for ridesharing matching \eqref{eq:RD_demand_constraint}-\eqref{eq:RP_demand_constraint}; and matched RD RP constraints \eqref{eq:RD_matching_cap}-\eqref{eq:RP_matching_cap} are transferred into the vector function $\mathbf{F^{\rho}}$ via penalizations. Also, the functions corresponding to multipliers $\lambda$ and $\phi$ are not included in $\mathbf{F^{\rho}}$, but can be retrieved from the constraints defining the set $\Omega$. Moreover, there are multiplicative factors $\eta$ within $\mathbf{F^{\rho}}$ but are not in the original complementarity problem (GEM-mpr). This construction is needed for the proof. The additional definitional constraints \eqref{eq:proof1_def_constraint1}-\eqref{eq:proof1_def_constraint2} (and the corresponding variable $Y$) are introduced for the convenience of the proof (to avoid redundant index in the penalization), which can be interpreted as the total number of RD and RP being matching for each OD pair $\omega$.
 
\noindent
\textbf{Solution existence of the penalized VI}

One of the challenges to directly prove solution existence for $VI(\mathbf{F^{\rho}}, \Omega)$ is the lack of explicit compactness of the set $\Omega$. Instead of assuming compactness a prior, we utilize the existence theorem for VI with unbounded set (\cite{Nagurney2009}). The procedure of the proof is twofold: we first consider a restricted problem $VI(\mathbf{F^{\rho}}, \Omega_R)$ by imposing additional constraints on the primary variable set for its compactness. Under the compactness assumption, solution existence can be readily shown. We later show that, these additional constraints are redundant at equilibrium, such that solution existence for the $VI(\mathbf{F^{\rho}}, \Omega)$ is also established for the unbounded set (due to \cite{Nagurney2009}).

\noindent
\textit{Solution existence of the restricted penalized $VI(\mathbf{F^{\rho}}, \Omega_R)$}

In this subsection, we first introduce additional constraints to the restricted $VI(\mathbf{F^{\rho}}, \Omega_R)$, then show solution existence for the $VI(\mathbf{F^{\rho}}, \Omega_R)$.

We assume the polyhedron $\Omega_R$ are constrained by:
\begin{align}
  & \text{Ridesharing matching definitional constraints} 
    & \text{\eqref{eq:proof1_def_constraint1}-\eqref{eq:proof1_def_constraint2}} 
    \nonumber \\ 
  & \text{Stable matching constraint }
    & \text{\eqref{eq:stable_matching_constraint}}
    \nonumber \\
  & \text{Ridesharing driver-passenger coupling constraints }
    &  \text{\eqref{eq:RD_RP_coupling_upper}-\eqref{eq:RD_RP_coupling_lower}}
    \nonumber \\
  & \text{Ridesharing passenger-transfer avoidance constraints }
    &  \text{\eqref{eq:intermediate_conservation_dropoff}-\eqref{eq:RD_pickup_cross}}
    \nonumber \\
  & \textbf{Compactness constraints }
    & 
    \nonumber \\
  & q_{\omega}^{m} \leq \kappa \cdot \underbrace{\sum_{\omega} {q_{\omega}}}_{\text{Model input}}
    & \forall m \in \mathcal{M}, \omega \in \mathcal{W} \label{eq:proof_compact_1}\\
  & Z_n \leq \kappa \cdot \sum_{\omega} {q_{\omega}}
    & \forall n \\
  & Y_{\omega}^{RD} \leq \kappa \cdot \sum_{\omega} {q_{\omega}}
    & \forall \omega \in \mathcal{W} \\
  & Y_{\omega}^{RP} \leq \kappa \cdot \sum_{\omega} {q_{\omega}}
    & \forall \omega \in \mathcal{W} \\
  & x_{i, j}^{e, m} \leq \kappa \cdot \sum_{\omega} {q_{\omega}}
    & \forall (i, j) \in \mathcal{E}, m \in \mathcal{M}, e \in \mathcal{D} \\
  & \pi_{i}^{e, m} \leq \kappa \cdot \sum_{(i, j) \in \mathcal{E} } {c_{ij} \textstyle\left( \sum_{\omega} {q_{\omega}} \right)}
    & \forall i \in \mathcal{N}, m \in \mathcal{M} , e \in \mathcal{D} \label{eq:proof_compact_last}
\end{align}

For some constants $\kappa>0$ (construction needed for dropping these compactness constraints later), the resulting polyhedron $\Omega_R$ is compact, such that the solution existence immediately follows:
\begin{lemma}
\label{restricted_VI_existence}
  If the link cost functions $c$ are nonnegative and continuous, and the set $\Omega_R$ is not empty, then $VI(\mathbf{F^{\rho}}, \Omega_R)$ has a solution, $\mathbf{w}_R^{*}$, for every $\rho > 0$.
\end{lemma}
\begin{proof}(Theorem 1, \cite{Nagurney2009})
  Since the polyhedron $\Omega_R$ is convex compact, and cost functions $c$ are continuous, $VI(\mathbf{F^{\rho}}, \Omega_R)$ admits at least one solution for every $\rho > 0$.
\end{proof}

Note that the compactness constraints \eqref{eq:proof_compact_1}-\eqref{eq:proof_compact_last} do not appear in the original MCP formulation (GEM-mpr), we need to show these bounds are redundant and establish the solution existence of the unrestricted $VI(\mathbf{F^{\rho}}, \Omega)$. In the following subsection, we first derive the upper bounds for the primary variables, and show that with some constants $\kappa$, these compactness constraints are not active, and establish solution existence for the unbounded set.

\noindent
\textit{Solution existence of the penalized $VI(\mathbf{F^{\rho}}, \Omega)$}

The goal of this subsection is to show that the compactness constraints \eqref{eq:proof_compact_1}-\eqref{eq:proof_compact_last} will not be active at equilibrium, such that solution existence for $VI(\mathbf{F^{\rho}}, \Omega)$ with the unbounded (unrestricted) set $\Omega$ is established. In the following paragraphs, we derive the upper bounds for the primary decision variables.

Since modal split $q_{\omega}^{m} \geq 0$, and $\sum_{m \in \mathcal{M}} {q_{\omega}^{m}} =  q_{\omega}, \forall \omega$, it follows that $q_{\omega}^{m} \leq q_{\omega} \leq \sum_{\omega} {q_{\omega}}, \forall m \in \mathcal{M}$. Also, from RD demand constraint $\sum_{n} {s_1{(n,0,\omega)} \cdot Z_n} \leq q_{\omega}^{RD}, \forall \omega$, by summing over all OD pair we can derive $\sum_{\omega} {\sum_{n} {s_1(n,0,\omega) \cdot Z_n}}=\sum_{n} {Z_n} \leq \sum_{\omega} {q_{\omega}}$, which results in $Z_n \leq \sum_{n} {Z_n} \leq \sum_{\omega} {q_{\omega}}$. With the same notion, by definitional constraints \eqref{eq:proof1_def_constraint1}-\eqref{eq:proof1_def_constraint2}, we also have $Y_{\omega}^{RD} \leq \sum_{\omega} {q_{\omega}}$ and $Y_{\omega}^{RP} \leq \sum_{\omega} {q_{\omega}}$. 

For link flow variables $x$ , we derive in the following separately for virtual link flows, and network link flows. For virtual link flows, by conservation constraints for the RD/RP virtual origin nodes \eqref{eq:RD_stable_matching_conservation}-\eqref{eq:RP_stable_matching_conservation} and non-negativity conditions, virtual RD link flows $x_{\widehat{i_{RD}}, i'}^{n, RD} \leq q_{(i,e)}^{RD} \leq \sum_{\omega:(i,e)} {q_{\omega}}, \forall (\widehat{i_{RD}}, i') \in \mathcal{E}_{\widehat{RD}}$, and virtual RP link flows $x_{\widehat{i_{RP}}, i'}^{n, e, RP} \leq q_{(i,e)}^{RP} \leq \sum_{\omega:(i,e)} {q_{\omega}}, \forall (\widehat{i_{RP}}, i') \in \mathcal{E}_{\widehat{RP}}$.

To derive upper bounds for the network link flow variables, we first show that, under the Wardrop's first principle, if link cost $c_{ij}>0$, links with positive flows will never form a cycle, such that network link flows are bounded. We delay the formal proof of the following proposition in Appendix \ref{proof.proposition1}.

\begin{proposition}
  \label{bounded_link_flows}
  If link cost 
  \begin{align}
  c_{ij}>0 \quad \quad, \forall (i, j) \in \mathcal{E} \setminus (\mathcal{E}_{\widehat{RD}} \cup \mathcal{E}_{\widehat{RP}})
  \end{align}
  then, the network link flow $x_{i, j}^{m}$ is bounded by the overall demands:
  \begin{align}
  x_{i, j}^{m} \leq \sum_{\omega} {q_{\omega}} \quad \quad, \forall m \in \mathcal{M} ,(i, j) \in \mathcal{E} \setminus (\mathcal{E}_{\widehat{RD}} \cup \mathcal{E}_{\widehat{RP}}) \nonumber
  \end{align}
\end{proposition}

Based on the \textit{no-cycle} observation, the upper bound can be derived for the node potential variables $\pi^{*}$at equilibrium by showing that the node potentials are accumulated recursively in the network with monotone link costs. The full proof is stated in Appendix \ref{proof.proposition2}.
\begin{proposition}
  \label{bounded_node_potentials}
  If link cost functions are monotone and
  \begin{align}
  c_{ij}>0 \quad \quad, \forall (i, j) \in \mathcal{E} \setminus (\mathcal{E}_{\widehat{RD}} \cup \mathcal{E}_{\widehat{RP}})
  \end{align}
  then, there exists ${\pi_{i}^{e,m}}^{*}$ at equilibrium such that:
  \begin{align}
  {\pi_{i}^{e,m}}^{*} \leq \sum_{(i, j) \in \mathcal{E}} {c_{ij}(\textstyle\sum_{\omega}{q_{\omega}})} \quad \quad, \forall i, m \in \mathcal{M} ,e \in \mathcal{D}  \nonumber
  \end{align}
\end{proposition}

After deriving upper bounds for all the primary variables, the following proposition states that, for some constants $\kappa$, the compactness constraints \eqref{eq:proof_compact_1}-\eqref{eq:proof_compact_last} are not active. This is proposition is needed for showing solution existence for VI problem with unbounded sets. 

\begin{proposition}
\label{redundancy}
  Under the assumptions of link cost functions \eqref{eq:mcp_proof_node_potentials_link_complementarity} and Propositions \eqref{bounded_link_flows}-\eqref{bounded_node_potentials}, for some constants $\kappa > 1$ , at equilibrium, we have the following:
  \begin{align}
    & {q_{\omega}^{m}}^{*} < \kappa \cdot \sum_{\omega} {q_{\omega}}
    & \forall m \in \mathcal{M}, \omega \in \mathcal{W} \nonumber\\
  & {Z_n}^{*} < \kappa \cdot \sum_{\omega} {q_{\omega}}
    & \forall n \nonumber\\
  & {Y_{\omega}^{RD}}^{*} < \kappa \cdot \sum_{\omega} {q_{\omega}}
    & \forall \omega \in \mathcal{W} \nonumber\\
  & {Y_{\omega}^{RP}}^{*} < \kappa \cdot \sum_{\omega} {q_{\omega}}
    & \forall \omega \in \mathcal{W} \nonumber\\
  & {x_{i, j}^{e, m}}^{*} < \kappa \cdot \sum_{\omega} {q_{\omega}}
    & \forall (i, j) \in \mathcal{E}, m \in \mathcal{M}, e \in \mathcal{D} \nonumber\\
  & {\pi_{i}^{e, m}}^{*} < \kappa \cdot \sum_{(i, j) \in \mathcal{E} } {c_{ij} \textstyle\left( \sum_{\omega} {q_{\omega}} \right)}
    & \forall i \in \mathcal{N}, m \in \mathcal{M} , e \in \mathcal{D} \nonumber
  \end{align}
  
\end{proposition}
\begin{proof}
See \cref{bounded_link_flows} and \cref{bounded_node_potentials}.
\end{proof}

\begin{lemma}
\label{VI_existence}
Under the assumptions of \cref{restricted_VI_existence} and \cref{redundancy}, the variational inequality problem $VI(\mathbf{F^{\rho}}, \Omega)$ admits a solution.
\end{lemma}
\begin{proof}(Theorem 2, \cite{Nagurney2009})
  $VI(\mathbf{F^{\rho}}, \Omega)$ admits a solution if and only if there exists a bound $R>0$ on the unbounded feasible set, and a solution of $VI(\mathbf{F^{\rho}}, \Omega_R)$, $\mathbf{w_R^{*}}$, such that $||\mathbf{w_R^{*}}|| < R$.
\end{proof}

With \cref{VI_existence}, we establish the solution existence for the penalized $VI(\mathbf{F^{\rho}}, \Omega_R)$. The goal of the next subsection is to show that, the original MCP formulation (GEM-mpr) can be recovered from the $VI(\mathbf{F^{\rho}}, \Omega_R)$, when the penalty is tending to infinity.

\noindent
\textbf{Solution existence for the proposed general equilibrium model}

The solution existence of the proposed the general equilibrium model is shown using the penalized $VI(\mathbf{F^{\rho}}, \Omega_R)$. We first formulate an equivalent MCP for the penalized $VI(\mathbf{F^{\rho}}, \Omega_R)$; then show the relaxed asymmetric constraints can be recovered from the penalized $VI(\mathbf{F^{\rho}}, \Omega_R)$ by taking the limiting arguments.

\noindent
\textit{MCP formulation for the penalized VI}

Recall that, the penalized VI is formulated for the normalized equilibrium problem, in which the multipliers corresponding to a coupled shared constraint are proportional, and defined as $\widehat{\phi}_5^{n,l} =\eta_5^{n}{\phi}_5^{n,l}$, $\widehat{\lambda}_{6}^{i^l,j^{l'},n}=\eta_{67}^{i^l,j^{l'},n}{\lambda}_{6}^{i^l,j^{l'},n}$ , and $\widehat{\lambda}_{7}^{i^l,j^{l'},n}=\eta_{67}^{i^l,j^{l'},n}{\lambda}_{7}^{i^l,j^{l'},n}$, with $\eta$ being some positive constants. We now rewrite the penalized asymmetric constraints in the form of multipliers:
\begin{align}
  & \mu_{1}^{\omega, \rho_1} = \rho_1 \cdot \left(Y_{\omega}^{RD} - q_{\omega}^{RD} \right) & \forall \omega \in \mathcal{W} \nonumber \\
  & \mu_{2}^{\omega, \rho_2} = \rho_2 \cdot \left(Y_{\omega}^{RP} - q_{\omega}^{RP} \right) & \forall \omega \in \mathcal{W} \nonumber \\
  & \mu_{3}^{i, n, \rho_3} = \frac{\rho_3}{\eta_5^{n}} \cdot \left[
                x_{\widehat{i_{RD}}, i^1}^{n, RD} 
                - s_1(n, 0 (i, \cdot)) \cdot Z_n
          \right] & \forall i, n \nonumber \\
  & \mu_{4}^{i, e, n, l, \rho_4} = \rho_4 \cdot \left[
                x_{\widehat{i_{RP}}, i^l}^{n, e, RP} 
                - s_1(n, l (i, \cdot)) \cdot Z_n
          \right] & \forall n, (i, e) \in \mathcal{W}, 1 \leq l \leq L-1 \nonumber
\end{align}
where, the multipliers $\mu^{\rho}$ for each asymmetric constraint depend on its corresponding penalty parameters $\rho$. By rewriting the asymmetric multipliers, the MCP for the penalized VI is shown in the following:

\noindent
[MCP - Penalized VI]
\begin{align}
  & \textbf{Mode choice model} 
    & \text{\eqref{eq:mcp_mode_choice}-\eqref{eq:mcp_mode_conservation} }
    \nonumber \\[1.5ex] 
  & \textbf{Ridesharing matching model}
    & 
    \nonumber \\
  & \sum_{n} {s_1{(n,0,\omega)} \cdot Z_n} = Y_{\omega}^{RD}
    \perp
    \phi_{12}
    \quad \text{free}
    & , \forall \omega \in \mathcal{W} \label{eq:mcp_proof_def1}
     \\
  & \sum_{n} {\sum_{1 \leq l \leq L-1} {s_1{(n,l,\omega)} \cdot Z_n} = Y_{\omega}^{RP}}
    \perp
    \phi_{13}
    \quad \text{free}
    & , \forall \omega \in \mathcal{W} \label{eq:mcp_proof_def2}
     \\
  & 0 \leq \left[ 
                \mu_{1}^{\omega, \rho_1} - \phi_{12}
    \right]
    \perp
    Y_{\omega}^{RD}
    \geq 0
    & , \forall \omega \in \mathcal{W} \label{eq:mcp_proof_YRD}
     \\
  & 0 \leq \left[ 
                \mu_{2}^{\omega, \rho_2} - \phi_{13}
    \right]
    \perp
    Y_{\omega}^{RP}
    \geq 0
    & , \forall \omega \in \mathcal{W} \label{eq:mcp_proof_YRP}
     \\
  & 0 \leq \left[ 
                -R_n + \phi_{12} \cdot s_1(n,0,\omega) + \phi_{13} \cdot \sum_{1 \leq l \leq L-1} {s_1(n,l,\omega)}
    \right]
    \perp
    Z_n
    \geq 0
    & , \forall n \label{eq:mcp_proof_Z}
\end{align}

By complementarity conditions \eqref{eq:mcp_proof_YRD}-\eqref{eq:mcp_proof_YRP}, if $Y_{\omega}^{RD} > 0$ and $Y_{\omega}^{RP} > 0$, we have $\phi_{12}=\mu_{1}^{\omega, \rho_1}$ and $\phi_{13}=\mu_{2}^{\omega, \rho_2}$. This means the RD and RP demand constraints (in the penalty form) are passed to the decision variables $Z_n$, in which complementarity condition \eqref{eq:mcp_proof_Z} becomes:
\begin{align}
& 0 \leq \left[ 
                -R_n +\mu_{1}^{\omega, \rho_1} \cdot s_1(n,0,\omega) + \mu_{2}^{\omega, \rho_2} \cdot \sum_{1 \leq l \leq L-1} {s_1(n,l,\omega)}
    \right]
    \perp
    Z_n
    \geq 0
    & , \forall n \nonumber
\end{align}

As discussed in the previous subsection, the definitional constraints \eqref{eq:proof1_def_constraint1}-\eqref{eq:proof1_def_constraint2} are only for the convenience of the proof. The above complementarity condition for the ridesharing matching problem is equivalent to the original ridesharing matching problem of GEM-mpr, with asymmetric multipliers $\mu_{1}^{\omega}$ and $\mu_{2}^{\omega}$ in their penalized form $\mu_{1}^{\omega, \rho_1}$ and $\mu_{2}^{\omega, \rho_2}$.

\textbf{Stable matching model}
\begin{subequations}
\label{eq:mcP_RD_stable_route_final}
    \begin{align}
    0 \leq \left[  
            \pi_{i^1}^{n, RD}
            + \mu_{3}^{i, n, \rho_3}
            + \sum_{l} {\phi_{5}^{n, l}}
            - \lambda_{10}^{i^1, n}
            - \pi_{\widehat{i_{RD}}}^{e, RD}  
      \right] 
    & \perp 
            x_{\widehat{i_{RD}}, i^1}^{n, RD}
    \geq 0, 
            & \forall i, n  \label{eq:mcp_RD_seq_choice_final}  \\
    0 \leq \left[  
            \pi_{i^l}^{n, e, RP} 
            + \mu_{4}^{i, e, n, l, \rho_4}
            - \eta_{5}^{n} \cdot {\phi}_{5}^{n, l}
              - \pi_{\widehat{i_{RP}}}^{e, RP}
      \right] 
    & \perp 
            x_{\widehat{i_{RP}}, i^l}^{n, e, RP}
    \geq 0, 
            & \forall i, e, n, l  \label{eq:mcp_RP_seq_choice_final}  \\
    0 \leq \left[  
            \pi_{i}^{e, DA} 
            - \pi_{\widehat{i_{RD}}}^{e, RD}  
      \right] 
    & \perp 
            x_{\widehat{i_{RD}}, i_{DA}}^{e, RD} 
    \geq 0, 
            & \forall i, e \in \mathcal{D}  \label{eq:mcp_RD_to_DA_final} \\
    0 \leq \left[  
            \pi_{i}^{e, PT} 
            - \pi_{\widehat{i_{RP}}}^{e, RP}  
      \right] 
    & \perp 
            x_{\widehat{i_{RP}}, i_{PT}}^{e, RP} 
    \geq 0, 
            & \forall i, e \in \mathcal{D}  \label{eq:mcp_RP_to_PT_final} \\
    0 \leq \left[ 
            \sum_{n\in \{n|s_1(n,0,(i,e))=1\}} {x_{\widehat{i_{RD}}, i^1}^{n, RD}}
            + x_{\widehat{i_{RD}}, i_{DA}}^{e, RD}
            - q_{(i,e)}^{RD}
      \right] 
    & \perp 
            \pi_{\widehat{i_{RD}}}^{e, RD}
    \geq 0, 
            & \forall i, e \in \mathcal{D}  \label{eq:mcp_RD_conservation_final} \\
    0 \leq \left[
            \sum_{n} {\sum_{1 \leq l \leq L-1} {x_{\widehat{i_{RP}}, i^l}^{n, e, RP}}} 
            + x_{\widehat{i_{RP}}, i_{PT}}^{e, RP}
            - q_{(i,e)}^{RP}
      \right] 
    & \perp 
            \pi_{\widehat{i_{RP}}}^{e, RP}
    \geq 0, 
            & \forall i, e \in \mathcal{D} \label{eq:mcp_RP_conservation_final} \\ 
    x_{\widehat{i_{RD}}, i^1}^{n, RD} = x_{\widehat{j_{RP}}, j^l}^{n, e, RP} & \perp  \phi_{5}^{n, l} \text{~free}, 
             \begin{split}{}&\forall n,  \\
               &i:s_1(n, 0, (i, \cdot))=1,          \\
               &(j, l, e) \in \zeta_{RP}^{n}      
    \end{split}  
             \label{eq:mcp_RD_stable_matching_final}
    \end{align}
\end{subequations}
We illustrate in the following how the penalized VI \eqref{VI_1} with normalized factor is formulated into the above MCP. Recall that, the vector function associated with $x_{\widehat{i_{RD}}, i^1}^{n, RD}$ has the form of:
\begin{align}
F_{x_{\widehat{i_{RD}}, i^1}^{n, RD}}^{\rho}  = \eta_{5}^{n} \cdot \left(\pi_{i^{1}}^{n, RD} - \pi_{\widehat{i_{RD}}}^{e, RD} \right) 
        + \rho_3 \cdot
            \left[
                  x_{\widehat{i_{RD}}, i^1}^{n, RD} 
                  - s_1(n, 0 (i, \cdot)) \cdot Z_n
            \right] \nonumber
\end{align}
, dividing the above equation with $\eta_{5}^{n}$, it follows:
\begin{align}
F_{x_{\widehat{i_{RD}}, i^1}^{n, RD}}^{\rho}  = \left(\pi_{i^{1}}^{n, RD} - \pi_{\widehat{i_{RD}}}^{e, RD} \right) 
        + \frac{\rho_3}{\eta_{5}^{n}} \cdot
            \left[
                  x_{\widehat{i_{RD}}, i^1}^{n, RD} 
                  - s_1(n, 0 (i, \cdot)) \cdot Z_n
            \right] = \pi_{i^{1}}^{n, RD} - \pi_{\widehat{i_{RD}}}^{e, RD} + \mu_{3}^{i, n, \rho_3}
\end{align}
By incorporating multipliers for the remaining constraints, the MCP for the stable matching problem in the penalized VI is formulated.

\textbf{Route choice model}
\begin{align}
  & \text{DA route choice} 
    & \text{\eqref{eq:mcp_DA_route_choice_final} }
    \nonumber \\
  & \text{RD route choice} 
    & \text{\eqref{eq:mcp_RD_route_choice_empty_final}-\eqref{eq:mcp_RD_route_choice_onboard_final} }
    \nonumber \\
  & \text{RP route choice} 
    &
    \nonumber \\
  & 0 \leq \left[
            \pi_{j^{l'}}^{n, e, RP} 
            + c_{i^{l}, j^{l'}}^{RP}
            + \eta_{67}^{i^l,j^{l'},n}{\lambda}_{6}^{i^l, j^{l'}, n} - \eta_{67}^{i^l,j^{l'},n}{\lambda}_{7}^{i^l, j^{l'}, n}
            - \pi_{i^{l}}^{n, e, RP}
        \right] 
    \perp 
            x_{i^{l}, j^{l'}}^{n, e, RP} 
    \geq 0
    &
     \\
  & 
    \quad \quad \quad \quad \quad \quad \quad \quad \quad \quad \quad \quad\quad \quad \quad \quad \quad \quad \quad \quad \quad, \forall i, n, l, j^{l'}:(i^{l}, j^{l'}) \in \mathcal{E}_{RP}, e \in \mathcal{D}
    &
    \nonumber \\
  & \text{PT route choice} 
    & \text{\eqref{eq:mcp_PT_route_choice_final} }
    \nonumber \\
  & \text{Multi-passenger ridesharing constraints} 
    & \text{\eqref{eq:mcp_group_intermediate_conservation}-\eqref{eq:mcp_group_pickup_cross} }
    \nonumber \\
  & \text{Flow conservation constraints} 
    & \text{\eqref{eq:flow_conservation} }
    \nonumber
\end{align}

As shown above, the resulting penalized MCP is one step closer to the original MCP for GEM-mpr, except for the asymmetric RD/RP demand constraints, and the number of matched RD and RP constraints. In the following subsection, a limiting argument is applied to recover the relaxed asymmetric constraints and the multipliers.

\noindent
\textit{Recovery of the relaxed asymmetric constraints}

In the previous subsection, solution existence has been established for the penalized VI problem, in which the asymmetric shared constraints are relaxed and transferred to the vector function. The goal of this subsection is to show that, these relaxed constraints can be recovered when the penalty tends to infinity, such that solution of the penalized VI is indeed a desired equilibrium solution of the original GEM-mpr problem. Similar to \textcite{ban2019general}, a limiting argument is applied on an arbitrary sequence $\left\{\mathbf{\rho_{\alpha}}: \left[\rho_{1,\alpha}, \rho_{2,\alpha}, \rho_{3,\alpha}, \rho_{4,\alpha} \right]  \right\}$ of (vector of) positive scalars tending to $\infty$, with $\alpha$ denote its index in the sequence. 

For every $\alpha$, the $VI(\mathbf{F^{\rho_{\alpha}}}, \Omega)$ is solved with respect to the penalty factor $\rho_{\alpha}$, which corresponds to a solution $\mathbf{w^{\alpha}} \in \Omega$. As a result, the solution sequence  $\left\{\mathbf{w_{\alpha}} \right\}$ also belongs to the set $\Omega$. If there exist bounds for the primary variables in the set $\Omega$ (i.e., the set is compact), the solution sequence contains a convergent subsequence. However, only the set defining $(q_{\alpha}, Z_{\alpha}, Y_{\alpha}, x_{\alpha})$ is compact, while bounds for the node potential variables $\pi$ are established at equilibrium (see \cref{bounded_node_potentials}). Therefore, only $\left\{\mathbf{w_{\alpha}^{'}}:(q_{\alpha}, Z_{\alpha}, Y_{\alpha}, x_{\alpha}) \right\}$ is assumed to converge to a limit $\mathbf{w_{\infty}^{'}} \triangleq (q_{\infty}, Z_{\infty}, Y_{\infty}, x_{\infty})$. To deal with the possible unboundedness of $\pi$ and the multipliers $\lambda, \phi$, we first resort to the following linear complementarity theorem:
\begin{lemma} (\cite{cottle2009linear})
\label{linear_complementarity}
Consider the following complementarity conditions:
\begin{flalign}
    0 \leq G(y) + A^{T} \cdot \phi + B^{T} \cdot \lambda
    \perp 
            y 
    \geq 0 \\
    0 = Cy - e
     \perp 
            \phi \quad \text{free} \label{eq:linear_comp_1}\\
    0 \leq Dy - f 
     \perp 
            \lambda 
    \geq 0 \label{eq:linear_comp_2}
\end{flalign}
where, $G$ is a continuous function. Let $\left\{(y_{k}, \phi_{k}, \lambda_{k}) \right\}$ be a sequence of solutions corresponding to the convergent sequence $\left\{(e_{k}, f_{k}) \right\}$ with $\lim_{k \rightarrow \infty}{e_k} = e_{\infty}$ and $\lim_{k \rightarrow \infty}{f_k} = f_{\infty}$. If $\lim_{k \rightarrow \infty}{y_k} = y_{\infty}$, then there exists $\left(\phi_{\infty}, \lambda_{\infty}\right) $ such that the tuple $(y_{\infty}, \phi_{\infty}, \lambda_{\infty})$ is a solution corresponding to $(e_{\infty}, f_{\infty})$.
\pushQED{\qed} 
{}\qedhere
\popQED
\end{lemma} 

Following \cref{linear_complementarity}, existence of $\pi_{\infty}, \phi_{\infty}$ and $\lambda_{\infty}$ can be established by constructing a tuple of solution $(x_{\infty}, \pi_{\infty}, \phi_{\infty}, \lambda_{\infty})$ corresponding to $(q_{\infty}, Z_{\infty}, Y_{\infty})$, such that $\lim_{k \rightarrow \infty}{x_k} = x_{\infty}$ and $\pi_{\infty}, \phi_{\infty}, \lambda_{\infty}$ correspond to the linear complementarity conditions (in the forms of \cref{eq:linear_comp_1}-\eqref{eq:linear_comp_2}). Let $\mathbf{w_{\infty}} \triangleq (q_{\infty}, Z_{\infty}, Y_{\infty}, x_{\infty}, \pi_{\infty})$ denote a solution to $VI(\mathbf{F^{\rho_{\infty}}}, \Omega)$, we would like to verify $\mathbf{w_{\infty}}$, or more specifically $(Z_{\infty}, Y_{\infty}, x_{\infty})$, satisfies the asymmetric constraints using the following lemma:
\begin{lemma}
\label{asymmetric_recovery}
Suppose there exists nonnegative tuple $q$ satisfies conservation constraints \eqref{eq:mcp_mode_conservation}, and nonnegative $Z$ satisfies definitional constraints \eqref{eq:proof1_def_constraint1}-\eqref{eq:proof1_def_constraint2}, then the limiting tuple $(Z_{\infty}, Y_{\infty}, x_{\infty})$ satisfies the following:
\begin{align}
& {Y_{\omega}^{RD}}_{\alpha} \leq {q_{\omega}^{RD}}_{\alpha} & \forall \omega \in \mathcal{W} \label{proof_start} \\
& {Y_{\omega}^{RP}}_{\alpha} \leq {q_{\omega}^{RP}}_{\alpha} & \forall \omega \in \mathcal{W}  \\
&x_{\widehat{i_{RD}}, i^1}^{n, RD} \leq s_1(n, 0, (i, \cdot)) \cdot Z_n, & \forall n, i \in O  \\
&x_{\widehat{i_{RP}}, i^l}^{n, e, RP} \leq s_1(n, l, (i, e)) \cdot Z_n,  & \forall n, (i, e): \omega, 1 \leq l \leq L-1 \label{proof_end}
\end{align}
\end{lemma}
\begin{proof}
Given a pair of feasible $(q_{\alpha}, Z_{\alpha})$, let $\bar{x}_{\alpha}$ and $\bar{Y}_{\alpha}$ being the variables such that the tuple $(Z_{\alpha}, \bar{Y}_{\alpha}, \bar{x}_{\alpha})$ satisfies \eqref{proof_start}-\eqref{proof_end}. Following the definition of $VI(\mathbf{F^{\rho_{\alpha}}}, \Omega)$, we have
\begin{align}
F^{\rho_{\alpha}}\left(\mathbf{w_{\alpha}} \right)^\top \left(\mathbf{w} - \mathbf{w_{\alpha}} \right) \geq 0, \forall \mathbf{w} \in \Omega \nonumber
\end{align}
Let $\bar{\mathbf{w}}_{\alpha} = (q_{\alpha}, Z_{\alpha}, \bar{Y_{\alpha}}, \bar{x_{\alpha}}, \pi_{\alpha}) \in \Omega$, differs from $\mathbf{w_{\alpha}}$ only in the $x$ and $Y$ variables, more specifically, $Y_{\omega}^{RD}, Y_{\omega}^{RP}, x_{\widehat{i_{RD}}, i^1}^{n, RD}, x_{\widehat{i_{RP}}, i^l}^{n, e, RP}$, and it follows that:
\begin{align}
0 &\leq F^{\rho_{\alpha}}\left(\mathbf{w_{\alpha}} \right)^\top \left(\mathbf{w} - \mathbf{w_{\alpha}} \right) \nonumber \\
  & = \sum_{\omega} 
    {\rho_{1, \alpha} 
    \left[{Y_{\omega, \alpha}^{RD}}
    - {q_{\omega, \alpha}^{RD}} \right] 
    \left({\bar{Y}_{\omega, \alpha }^{RD}}
    - {Y_{\omega, \alpha}^{RD}} \right)} \nonumber\\
  & + \sum_{\omega} 
    {\rho_{2, \alpha} 
    \left[{Y_{\omega, \alpha}^{RP}}
    - {q_{\omega, \alpha}^{RP}} \right] 
    \left({\bar{Y}_{\omega, \alpha }^{RP}}
  - {Y_{\omega, \alpha}^{RP}} \right)} \nonumber\\
  & + \sum_{(i,e)} {\sum_{n} \left[ {\eta_{5}^{n} \cdot \left(\pi_{i^{1}, \alpha}^{n, RD} - \pi_{\widehat{i_{RD}}, \alpha}^{e, RD} \right)
        + \rho_{3, \alpha} \cdot
            \left(
                  x_{\widehat{i_{RD}}, i^1, \alpha}^{n, RD} 
                  - s_1(n, 0 (i, \cdot)) \cdot Z_{n, \alpha}
            \right)} \right] \left(\bar{x}_{\widehat{i_{RD}}, i^1, \alpha}^{n, RD} - x_{\widehat{i_{RD}}, i^1, \alpha}^{n, RD} \right) } \nonumber \\
  & + \sum_{(i, e)} {
          \sum_{n} {
              \sum_{l} {
                \left[
                \pi_{i^{l}, \alpha}^{n, e, RP} - \pi_{\widehat{i_{RP}}, \alpha}^{e, RP}
                + \rho_{4, \alpha} \cdot
                    \left(
                          x_{\widehat{i_{RP}}, i^l, \alpha}^{n, e, RP} 
                          - s_1(n, l (i, \cdot)) \cdot Z_{n, \alpha}
                    \right)
                \right] \left(\bar{x}_{\widehat{i_{RP}}, i^l, \alpha}^{n, e, RP} - x_{\widehat{i_{RP}}, i^l, \alpha}^{n, e, RP} \right)
              }
          }
  } \nonumber \\
  & = \rho_{1, \alpha} \cdot \sum_{\omega} 
    { 
    \left[{Y_{\omega, \alpha}^{RD}}
    - {q_{\omega, \alpha}^{RD}} \right] 
    \left({\bar{Y}_{\omega, \alpha }^{RD}}
    - {Y_{\omega, \alpha}^{RD}} \right)} \nonumber\\
  & + \rho_{2, \alpha} \cdot \sum_{\omega} 
    { 
    \left[{Y_{\omega, \alpha}^{RP}}
    - {q_{\omega, \alpha}^{RP}} \right] 
    \left({\bar{Y}_{\omega, \alpha }^{RP}}
  - {Y_{\omega, \alpha}^{RP}} \right)} \nonumber\\
  & + \sum_{(i,e)} {\sum_{n} \left[ {\eta_{5}^{n} \cdot \left(\pi_{i^{1}, \alpha}^{n, RD} - \pi_{\widehat{i_{RD}}, \alpha}^{e, RD} \right)} \right] \left(\bar{x}_{\widehat{i_{RD}}, i^1, \alpha}^{n, RD} - x_{\widehat{i_{RD}}, i^1, \alpha}^{n, RD} \right) } \nonumber \\
  & + \rho_{3, \alpha} \cdot \sum_{(i,e)} {\sum_{n}  {
            \left(
                  x_{\widehat{i_{RD}}, i^1, \alpha}^{n, RD} 
                  - s_1(n, 0 (i, \cdot)) \cdot Z_{n, \alpha}
            \right)} \left(\bar{x}_{\widehat{i_{RD}}, i^1, \alpha}^{n, RD} - x_{\widehat{i_{RD}}, i^1, \alpha}^{n, RD} \right) } \nonumber \\
  & + \sum_{(i, e)} {
          \sum_{n} {
              \sum_{l} {
                \left[
                \pi_{i^{l}, \alpha}^{n, e, RP} - \pi_{\widehat{i_{RP}}, \alpha}^{e, RP}
                \right] \left(\bar{x}_{\widehat{i_{RP}}, i^l, \alpha}^{n, e, RP} - x_{\widehat{i_{RP}}, i^l, \alpha}^{n, e, RP} \right)
              }
          }
  } \nonumber \\
  & + \rho_{4, \alpha} \cdot \sum_{(i, e)} {
          \sum_{n} {
              \sum_{l} {
                \left[
                    \left(
                          x_{\widehat{i_{RP}}, i^l, \alpha}^{n, e, RP} 
                          - s_1(n, l (i, \cdot)) \cdot Z_{n, \alpha}
                    \right)
                \right] \left(\bar{x}_{\widehat{i_{RP}}, i^l, \alpha}^{n, e, RP} - x_{\widehat{i_{RP}}, i^l, \alpha}^{n, e, RP} \right)
              }
          }
  } \nonumber  
\end{align}
Since, by the choice of ${\bar{Y}_{\omega, \alpha }^{RD}}$, it follows ${\bar{Y}_{\omega, \alpha }^{RD}} \leq q_{\omega, \alpha}^{RD}$, therefore:
\begin{align}
\rho_{1, \alpha} \cdot \sum_{\omega} 
    { 
    \left[{Y_{\omega, \alpha}^{RD}}
    - {q_{\omega, \alpha}^{RD}} \right] 
    \left({\bar{Y}_{\omega, \alpha }^{RD}}
    - {Y_{\omega, \alpha}^{RD}} \right)} 
    & \leq \rho_{1, \alpha} \cdot \sum_{\omega} 
    { \left[{Y_{\omega, \alpha}^{RD}}
    - {q_{\omega, \alpha}^{RD}} \right] 
    \left(q_{\omega, \alpha}^{RD}
    - {Y_{\omega, \alpha}^{RD}} \right)} \nonumber \\
    & = -\rho_{1, \alpha} \cdot \sum_{\omega} {\left({Y_{\omega, \alpha}^{RD}}
    - {q_{\omega, \alpha}^{RD}} \right)^{2} } \nonumber
\end{align}
Similarly, we have:
\begin{align}
\rho_{2, \alpha} \cdot \sum_{\omega} 
    { 
    \left[{Y_{\omega, \alpha}^{RP}}
    - {q_{\omega, \alpha}^{RP}} \right] 
    \left({\bar{Y}_{\omega, \alpha }^{RP}}
    - {Y_{\omega, \alpha}^{RP}} \right)} 
    & \leq -\rho_{2, \alpha} \cdot \sum_{\omega} {\left({Y_{\omega, \alpha}^{RP}}
    - {q_{\omega, \alpha}^{RP}} \right)^{2} } \nonumber
\end{align}
Also, by the choice of $\bar{x}_{\widehat{i_{RD}}, i^1, \alpha}^{n, RD}$, it follows $\bar{x}_{\widehat{i_{RD}}, i^1, \alpha}^{n, RD} \leq  s_1(n, 0 (i, \cdot)) \cdot Z_{n, \alpha}$, therefore:
\begin{align}
\rho_{3, \alpha} & \cdot \sum_{(i,e)} {\sum_{n}  {
            \left(
                  x_{\widehat{i_{RD}}, i^1, \alpha}^{n, RD} 
                  - s_1(n, 0 (i, \cdot)) \cdot Z_{n, \alpha}
            \right)} \left(\bar{x}_{\widehat{i_{RD}}, i^1, \alpha}^{n, RD} - x_{\widehat{i_{RD}}, i^1, \alpha}^{n, RD} \right) } \nonumber \\
            & \leq -\rho_{3, \alpha} \sum_{(i,e)} {\sum_{n}  {
            \left(
                  x_{\widehat{i_{RD}}, i^1, \alpha}^{n, RD} 
                  - s_1(n, 0 (i, \cdot)) \cdot Z_{n, \alpha}
            \right)^{2}} } \nonumber
\end{align}
Similarly, it follows:
\begin{align}
\rho_{4, \alpha} & \cdot \sum_{(i, e)} {
          \sum_{n} {
              \sum_{l} {
                \left[
                    \left(
                          x_{\widehat{i_{RP}}, i^l, \alpha}^{n, e, RP} 
                          - s_1(n, l (i, \cdot)) \cdot Z_{n, \alpha}
                    \right)
                \right] \left(\bar{x}_{\widehat{i_{RP}}, i^l, \alpha}^{n, e, RP} - x_{\widehat{i_{RP}}, i^l, \alpha}^{n, e, RP} \right)
              }
          }
      }\nonumber \\
    & \leq - \rho_{4, \alpha} \sum_{(i, e)} {
          \sum_{n} {
              \sum_{l} {
                    \left(
                          x_{\widehat{i_{RP}}, i^l, \alpha}^{n, e, RP} 
                          - s_1(n, l (i, \cdot)) \cdot Z_{n, \alpha}
                    \right)^{2}
              }
          }
      }\nonumber
\end{align} 
To summarize, we have:
\begin{align}
0 &\leq F^{\rho_{\alpha}}\left(\mathbf{w_{\alpha}} \right)^\top \left(\mathbf{w} - \mathbf{w_{\alpha}} \right) \nonumber \\
  & \leq -\rho_{1, \alpha} \cdot \sum_{\omega} {\left({Y_{\omega, \alpha}^{RD}}
    - {q_{\omega, \alpha}^{RD}} \right)^{2} } \nonumber \\
  & - \rho_{2, \alpha} \cdot \sum_{\omega} {\left({Y_{\omega, \alpha}^{RP}}
    - {q_{\omega, \alpha}^{RP}} \right)^{2} } \nonumber \\
  & -\rho_{3, \alpha} \sum_{(i,e)} {\sum_{n}  {
              \left(
                  x_{\widehat{i_{RD}}, i^1, \alpha}^{n, RD} 
                  - s_1(n, 0 (i, \cdot)) \cdot Z_{n, \alpha}
              \right)^{2}
              } 
            } \nonumber \\
  & - \rho_{4, \alpha} \sum_{(i, e)} {
          \sum_{n} {
              \sum_{l} {
                    \left(
                          x_{\widehat{i_{RP}}, i^l, \alpha}^{n, e, RP} 
                          - s_1(n, l (i, \cdot)) \cdot Z_{n, \alpha}
                    \right)^{2}
              }
          }
      } + \text{Remaining part} \nonumber
\end{align}
With sequence $\left\{\mathbf{\rho_{\alpha}}: \left[\rho_{1,\alpha}, \rho_{2,\alpha}, \rho_{3,\alpha}, \rho_{4,\alpha} \right]  \right\}$ of positive scalars tending to $\infty$, and the remaining part has not affected by $\rho_{\alpha}$, this shows that, for $F^{\rho_{\alpha}}\left(\mathbf{w_{\alpha}} \right)^\top \left(\mathbf{w} - \mathbf{w_{\alpha}} \right) \geq 0$, we have:
\begin{align}
& \lim_{\alpha \rightarrow \infty} {\left({Y_{\omega, \alpha}^{RD}} - {q_{\omega, \alpha}^{RD}} \right) = 0} &, \forall \omega  \nonumber  \\
& \lim_{\alpha \rightarrow \infty} {\left({Y_{\omega, \alpha}^{RP}} - {q_{\omega, \alpha}^{RP}} \right) = 0} &, \forall \omega  \nonumber  \\
& \lim_{\alpha \rightarrow \infty} {\left(
                  x_{\widehat{i_{RD}}, i^1, \alpha}^{n, RD} 
                  - s_1(n, 0 (i, \cdot)) \cdot Z_{n, \alpha}
              \right)} = 0 &,  \forall n, i \in O  \nonumber  \\
& \lim_{\alpha \rightarrow \infty} {\left(
                          x_{\widehat{i_{RP}}, i^l, \alpha}^{n, e, RP} 
                          - s_1(n, l (i, \cdot)) \cdot Z_{n, \alpha}
                    \right)} = 0 &,  \forall n, (i, e): \omega, 1 \leq l \leq L-1 \nonumber 
\end{align}
By substituting the definitional constraints \eqref{eq:proof1_def_constraint1}-\eqref{eq:proof1_def_constraint2} into the limiting argument, we also obtain:
\begin{align}
& \lim_{\alpha \rightarrow \infty} {\left(\sum_{n} {s_1{(n,0,\omega)} \cdot Z_{n, \alpha}} - {q_{\omega, \alpha}^{RD}} \right) = 0} &, \forall \omega  \nonumber  \\
& \lim_{\alpha \rightarrow \infty} {
          \sum_{n} {
                    \left(\sum_{1 \leq l \leq L-1} {s_1{(n,l,\omega)} \cdot Z_{n, \alpha}}- {q_{\omega, \alpha}^{RP}}\right) 
                    = 0
                          } 
        }&, \forall \omega  \nonumber
\end{align}
The above shows that, when penalty goes to infinity, the relaxed constraints are satisfied.
\end{proof}

\subsubsection{Proof of \cref{bounded_link_flows}}
\label{proof.proposition1}
\begin{proof}
  Assuming a connected network, there exists a path from any node $i$ to any node $j$. For \textit{a} path, let $T_{ij}^{e, m}$ denote the path travel cost between node $i$ and $j$ with destination $e$ and mode $m$, and similarly $f_{ij}^{e, m}$ denote the path flow, which are defined as follows:
  \begin{align}
  & T_{ij}^{e, m} = \sum_{(a,b)} {\delta_{ab}^{ij, e, m} \cdot c_{ab}} &, \forall i,j,e,m \nonumber \\
  & f_{ij}^{e, m} \leq x_{ab}^{e, m} &, \forall i,j,e,m,  \delta_{ab}^{ij, e, m}=1\nonumber
  \end{align}
where, $\delta_{ab}^{ij, e, m}=1$ if link $(a,b)$ is part of the path between node $i$ and $j$ with destination $e$ and mode $m$; and $\delta_{ab}^{ij, e, m}=0$ otherwise. We have the following complementarity condition for the path flow:

\begin{align}
  & 0 \leq \left[ 
                    \pi_{je}^{m} + T_{ij}^{e, m} - \pi_{ie}^{m}
  \right] 
  \perp
  f_{ij}^{e, m}
  \geq 0
  &, \forall i,j,e,m
  \label{eq:mcp_proof_path_complementarity}
  \end{align}
Suppose link $(j, i) \in \mathcal{E}$, we have also the following complementarity condition for the link flow:
\begin{align}
  & 0 \leq \left[ 
                    \pi_{ie}^{m} + c_{ji}^{m} - \pi_{je}^{m}
  \right] 
  \perp
  x_{ji}^{e, m}
  \geq 0
  &, \forall i,j,e,m
  \label{eq:mcp_proof_link_complementarity}
  \end{align}
If link flow $x_{ji}^{e, m} > 0$, by complementarity condition \eqref{eq:mcp_proof_link_complementarity}, we have $\pi_{ie}^{m} + c_{ji}^{m} = \pi_{je}^{m}$. By substituting $\pi_{je}^{m}$ into \eqref{eq:mcp_proof_path_complementarity}, it follows:
\begin{align}
  & \pi_{je}^{m} + T_{ij}^{e, m} - \pi_{ie}^{m} = \pi_{ie}^{m} + c_{ji}^{m} + T_{ij}^{e, m} - \pi_{ie}^{m} = c_{ji}^{m} + T_{ij}^{e, m} > 0
  &, \forall i,j,e,m
  \nonumber
  \end{align}
, which indicates $f_{ij}^{e, m} =0$. 

This implies, at equilibrium, there is no cycle formed by links with positive flows (i.e., either $x_{ji}^{e, m} > 0$ or $f_{ij}^{e, m} > 0$; but not both).

Relying on this \textit{no-cycle} observation, we derive the link flow upper bounds for each mode, with $\delta_{ij}^{o,e} \in [0, 1]$ represents the portion of demands from $o$ to $e$ traversing link $(i, j)$, as follows:
\begin{align}
  x_{ij}^{e, DA} &= \sum_{o \in \mathcal{O}} {\delta_{ij}^{o,e} \cdot \left(q_{(o, e)}^{DA} + x_{\widehat{o_{RD}}, o_{DA}}^{e, RD} \right)  } \leq \sum_{o \in \mathcal{O}} {\left(q_{(o, e)}^{DA} + x_{\widehat{o_{RD}}, o_{DA}}^{e, RD} \right)  }
  &
  \nonumber \\
  & \leq \sum_{o \in \mathcal{O}} {\left(q_{(o, e)}^{DA} + q_{(o, e)}^{RD} \right)  } \leq \sum_{\omega :(o, e)} {q_{\omega}}
  &, \forall i,j,e,m
\end{align}
\begin{align}
  x_{ij}^{e, PT} &= \sum_{o \in \mathcal{O}} {\delta_{ij}^{o,e} \cdot \left(q_{(o, e)}^{PT} + x_{\widehat{o_{RP}}, o_{PT}}^{e, RP} \right)  } \leq \sum_{o \in \mathcal{O}} {\left(q_{(o, e)}^{PT} + x_{\widehat{o_{RP}}, o_{PT}}^{e, RP} \right)  }
  &
  \nonumber \\
  & \leq \sum_{o \in \mathcal{O}} {\left(q_{(o, e)}^{PT} + q_{(o, e)}^{RP} \right)  } \leq \sum_{\omega :(o, e)} {q_{\omega}}
  &, \forall i,j,e,m
\end{align}
\begin{align}
  x_{i^{l}j^{l'}}^{n, RD} &= \delta_{i^{l}j^{l'}}^{o,e} \cdot x_{\widehat{o_{RD}}, o^1}^{n, RD} \leq x_{\widehat{o_{RD}}, o^1}^{n, RD} \leq \sum_{\omega} {q_{\omega}}
  &, \forall i, n, l, j^{l'}: (i^{l}, j^{l'}) \in \mathcal{E}_{RD}, (o,e):s_1(n, 0, (o,e)) = 1
\end{align}
\begin{align}
  x_{i^{l}j^{l'}}^{n, e, RP} &= \sum_{o \in \widehat{\mathcal{O}}} {\delta_{i^{l}j^{l'}}^{o,e} \cdot \sum_{1\leq l'' \leq L-1} {x_{\widehat{o_{RP}}, o^{l''}}^{n, e, RP}}}
  &, \forall i, n, l, j^{l'}: (i^{l}, j^{l'}) \in \mathcal{E}_{RP}, e \in \mathcal{D} \nonumber \\
  & \leq \sum_{o \in \widehat{\mathcal{O}}} {\sum_{1\leq l'' \leq L-1} {x_{\widehat{o_{RP}}, o^{l''}}^{n, e, RP}}} 
  \leq \sum_{o \in \widehat{\mathcal{O}}} {q_{(o, e)}^{RP}} \leq  \sum_{\omega} {q_{\omega}}
  & ,\widehat{\mathcal{O}} \coloneqq \left\{o | s_1(n, l'', (o,e)) = 1, 1 \leq l'' \leq L-1 \right\}
\end{align}
\end{proof}

\subsubsection{Proof of \cref{bounded_node_potentials}}
\label{proof.proposition2}
\begin{proof}
  Suppose link $(i, j) \in \mathcal{E}$, and we have the following complementarity condition:
\begin{align}
  & 0 \leq \left[ 
                    \pi_{je}^{m} + c_{ij}^{m} - \pi_{ie}^{m}
  \right] 
  \perp
  x_{ij}^{e, m}
  \geq 0
  &, \forall i,j,e,m
  \label{eq:mcp_proof_node_potentials_link_complementarity}
  \end{align}
If there is link flow ${x_{ij}^{e, m}}^{*} > 0$ at equilibrium with link cost ${c_{ij}^{m}}^{*}$, by complementarity condition \eqref{eq:mcp_proof_node_potentials_link_complementarity}, we have:
\begin{align}
        {\pi_{je}^{m}}^{*} + {c_{ij}^{m}}^{*} = {\pi_{ie}^{m}}^{*} > 0 \nonumber
  \end{align}
By conservation at node $i$, if ${x_{ij}^{e, m}}^{*} > 0$ and $j \neq e$, there exists link flow ${x_{jk}^{e, m}}^{*} > 0$. Therefore, by similar complementarity conditions, it follows:
\begin{align}
        {\pi_{ke}^{m}}^{*} + {c_{jk}^{m}}^{*} = {\pi_{je}^{m}}^{*} \nonumber
  \end{align}
Then, the node potential at node $i$, ${\pi_{ie}^{m}}^{*}$ can be recursively defined by traversing links with positive flows, and accumulating their link costs:
\begin{align}
        {\pi_{ie}^{m}}^{*} = {c_{ij}^{m}}^{*} + {c_{jk}^{m}}^{*} + ... = \sum_{(j, k): {x_{jk}^{e, m}}^{*} > 0} {c_{jk}^{m}}^{*} \leq \sum_{(j, k) \in \mathcal{E} } {c_{jk}^{m}}^{*} \nonumber
\end{align}
Furthermore, under the monotone link cost assumption, we have:
\begin{align}
    {c_{jk}^{m}}^{*} = {c_{jk}^{m}}({x_{ij}^{e, m}}^{*}) \leq {c_{jk}^{m}}(\textstyle\sum_{\omega}{q_{\omega}})     \nonumber
\end{align}
Therefore,
\begin{align}
    {\pi_{ie}^{m}}^{*} \leq \sum_{(i, j) \in \mathcal{E} } {{c_{ij}^{m}}(\textstyle\sum_{\omega}{q_{\omega}})}     \nonumber
\end{align}
\end{proof}

\begin{longtable}{ll}
\caption{Definitions of key index, parameters, and decision variables} \label{tab:index_list}\\ \toprule
\textbf{Notation} & \textbf{Definition}  \\ \midrule
\endfirsthead
\caption{Definitions of key index, parameters, and decision variables (Cont.)} \label{tab:index_list}\\
\toprule
\textbf{Notation} & \textbf{Definition} \\ \midrule
\endhead
\hline
\multicolumn{2}{r}{Continued on next page}\\   \bottomrule
\endfoot
\bottomrule
\endlastfoot
\textbf{Index} & ~ \\ 
$o$ & origin \\ 
$d, e$ & destination \\ 
$n$ & matching sequence \\ 
$w:(o, d)$ & origin-destination pair \\ 
$l$ & level in the hyper-network \\ 
$i,j$ & node \\ 
$\widehat{i}$, $\widehat{j}$ & virtual node \\ 
$(i,j)$ & link \\ 
$m$ & mode \\ 
$k$ & route \\ 
$o_{nl}, d_{nl}$ & OD of matching sequence $n$ at level $l$ \\ 
$p$ & iteration number \\ 
$\bar{\cdot}$ & costliest \\ 
$\underline{\cdot}$ & cheapest \\ 
$\kappa$ & sequence-route \\ 
~ & ~ \\ 
\textbf{Sets} & ~ \\ 
$\mathcal{O}$ & origins \\ 
$\mathcal{D}$ & destinations \\ 
$\mathcal{W}$ & OD pairs \\ 
$\mathcal{N}$ & nodes \\ 
$\mathcal{E}$ & links \\ 
$\mathcal{M}$ & modes \\ 
$\mathcal{S}_w$ & matching sequences for $w$ \\ 
$\Psi^{w,m}$ & flow class for mode $m$ \\ 
$K_{od}$ & unrestricted path set \\ 
$B$ & bushes \\ 
~ & ~ \\ 
\textbf{Parameters} & ~ \\ 
$s_1$ & Pickup incident \\ 
$s_{-1}$ & drop-off incident \\ 
$L$ & Max sequence length \\ 
$q^w$ & OD demand \\ 
$C_w^m$ & modal cost \\ 
$C_n^w$ & matching sequence cost \\ 
$\alpha^m$ & VOT of traveler using $m$ \\ 
$\beta$ & unit vehicle operation cost \\ 
$\gamma$ & shared portion of a RD trip \\ 
$\tau_t^m$, $\tau_d^m$ & unit inconvenience cost \\ 
$\nu_t$, $\nu_d$ & unit ridesharing price \\ 
$R_n$ & matching sequence objective value \\ 
$c_{ij}^m$ & mode-specific link cost \\ 
$\tilde{c}_{ij}^m$ & generalized link cost \\ 
$d_{ij}$ & link length \\ 
$\mathcal{A}$ & node-link incident \\ 
$B_{n,l}$ & RD occupancy status \\ 
$s(n,l)$ & task node mapping \\ 
$\rho, \tilde{\mu}$ & penalty factor for augmented Lagrangian \\ 
$\theta$ & step sizes \\ 
$\mathcal{T}_i^o$ & topology order of $i$ in bush $o$ \\ 
$\tilde{\pi}$ & bush min-cost label \\ 
$\tilde{U}$ & bush max-cost label \\ 
$\sigma$ & augmented Lagrangian parameter \\ 
$\mathcal{L}$ & Lagrangian function of network subproblem \\ 
~ & ~ \\ 
\textbf{Decision variables} & ~ \\ 
$q_w^m$ & modal demand \\ 
$\pi$ & multiplier for the conservation constraints \\ 
$Z_n$ & platform matching decision \\ 
$\mu$ & multiplier for inequality constraints (e.g., ridesharing matching demand) \\ 
$x$ & link flow \\ 
$t_{ij}$ & link travel time \\ 
$F_n$ & sequence flow \\ 
$f_k$ & route-flow variable for unrestricted path set \\
\end{longtable}

\DIFaddbegin 
\begin{figure}[H]
    \centering
    \includegraphics[width=0.95\textwidth]{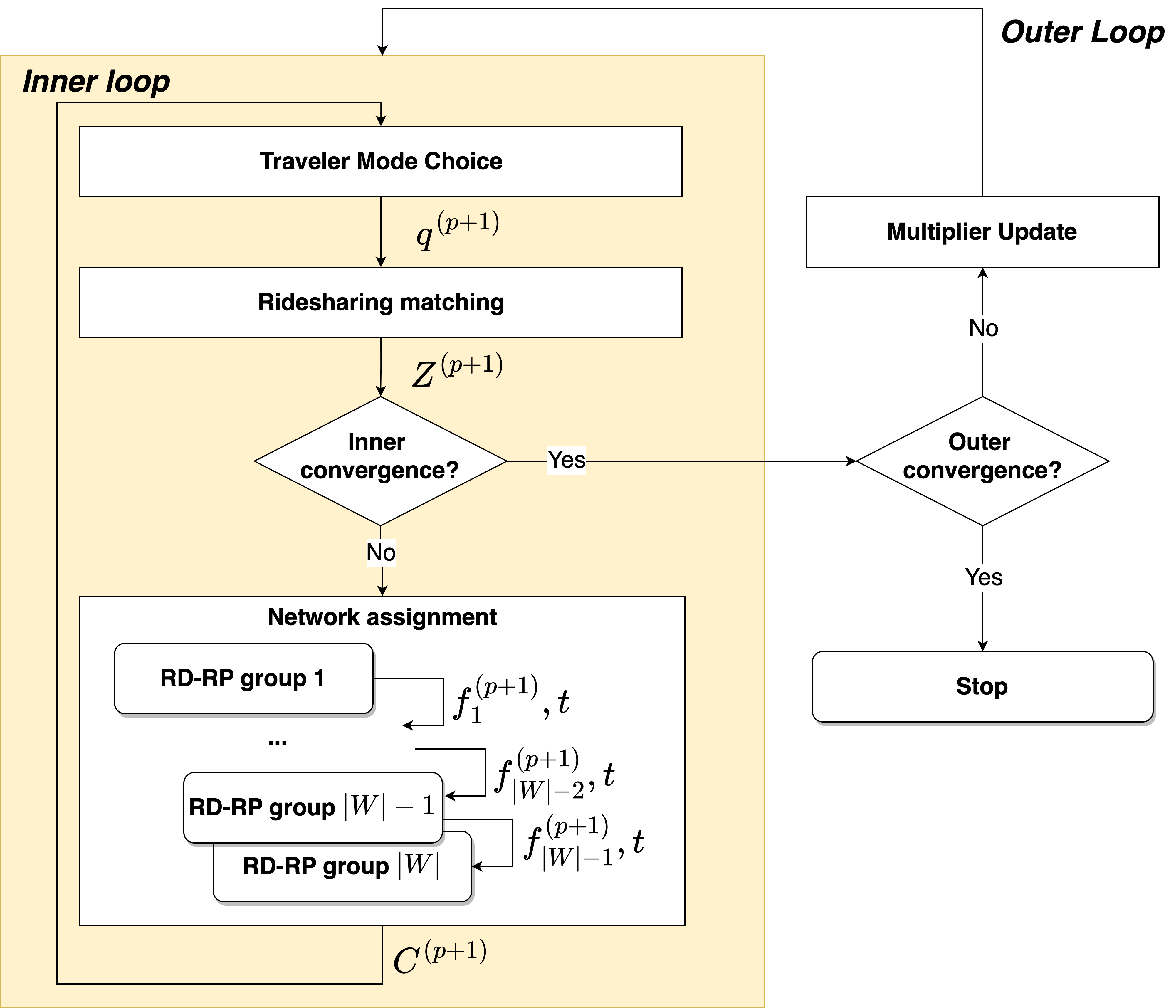}
    \caption{\DIFaddFL{Flowchart of sequence-flow assignment algorithm}}
    \label{fig:algorithm_flowchart}
\end{figure}
\DIFaddend 

\end{document}